\documentclass[12pt]{article}

\newcommand{\upd}{\mathrm{d}}
\newcommand{\Lbar}{\underline{L}}

\usepackage{amsmath}
\usepackage{amssymb}
\usepackage{amsthm}
\usepackage{mathrsfs}
\usepackage{slashed}
\usepackage{graphicx}
\usepackage{hyperref}
\usepackage{enumitem}
\usepackage{upgreek}

\usepackage[utf8]{inputenc}
\usepackage[backend=biber, style=numeric-comp, sorting=none]{biblatex}
\renewbibmacro*{in:}{%
	\ifentrytype{article}{}{\printtext{\bibstring{in}\intitlepunct}}}
\DeclareMathOperator*{\esssup}{ess\,sup}

\newtheorem{theorem}{Theorem}[section]
\newtheorem{lemma}[theorem]{Lemma}
\newtheorem{proposition}[theorem]{Proposition}
\newtheorem{corollary}[theorem]{Corollary}
\newtheorem{claim}[theorem]{Claim}

\newtheorem{assumption}{Assumption}
\newtheorem{definition}[theorem]{definition}

\newcounter{assumpcount}
\theoremstyle{remark}
\newtheorem{remark}[theorem]{Remark}

\newcommand{\dVol}{\mathrm{d}\textit{vol}}

\title{Evanescent ergosurface instability}

\begin{document}
\author{Joe Keir \\ \\
{\small DAMTP, Centre for Mathematical Sciences, University of Cambridge}, \\
{\small \sl Wilberforce Road, Cambridge CB3 0WA, UK} \\ \\
\small{j.keir@damtp.cam.ac.uk}}
\maketitle

\begin{abstract}
	Some exotic compact objects, including supersymmetric microstate geometries and certain boson stars, possess \emph{evanescent ergosurfaces}: timelike submanifolds on which a Killing vector field, which is timelike everywhere else, becomes null. We show that any manifold possessing an evanescent ergosurface but no event horizon exhibits a linear instability of a peculiar kind: \emph{either} there are solutions to the linear wave equation which concentrate a finite amount of energy into an arbitrarily small spatial region, \emph{or} the energy of waves measured by a stationary family of observers can be amplified by an arbitrarily large amount. In certain circumstances we can rule out the first type of instability. We also provide a generalisation to asymptotically Kaluza-Klein manifolds. This instability bears some similarity with the ``ergoregion instability'' of Friedman \cite{Friedman1978}, and we use many of the results from the recent proof of this instability by Moschidis \cite{Moschidis2016}.
\end{abstract}

\tableofcontents

\section{Introduction}

With the recent experimental detection of gravitational waves \cite{GWdetection} there has been a great deal of interest in \emph{exotic compact objects} and their properties. These objects, which are often solutions to various speculative theories, are supposed to ``mimic'' certain aspects of black holes: they are extremely compact, with a strong localised gravitational field, while having a similar asymptotic structure to black holes. On the other hand, many of these objects are supposed to avoid some of the ``pathologies'' of black holes: in particular, they are often non-singuler. From an observational point of view, many of these exotic objects have a compact region in which null geodesics can be ``trapped'' (as at the ``photon sphere'' in Schwarzschild), and this can lead to similar gravitational wave signals to those emitted by black holes (\cite{Cardoso2016}, \cite{Cardoso2016A}), at least on short time scales. Hence the recent interest in this subject: it is clearly of great importance to be able to distinguish these objects from genuine black holes, both from a theoretical and an observational point of view.

Despite mimicking black holes to some extent, in other ways these geometries can differ drastically from black holes, and this might provide a way to distinguish between them. For example, the gravitational wave signal from many of the exotic objects is expected to exhibit ``echoes'', in a way which black holes do not (see \cite{Cardoso2017} for an overview). In addition (and in some ways related to the ``echoes''), many of these exotic objects are classically unstable (e.g.\ \cite{Cardoso2005, Keir2016, Eperon2016, Keir2016A}), whereas black holes are expected to be classically nonlinearly stable. The question of stability will be the focus of this work.

An unusual geometric feature that is present in some of these exotic geometries is an ``evanescent ergosurface'' - for example, this is present in \emph{supersymmetric microstate geometries} (see \cite{Maldacena2000, Balasubramanian2000, Lunin2002, Giusto2004, Giusto2004A, Bena2005, Berglund2005, Gibbons2013}), as well as in boson stars\footnote{In this context, the evanescent ergoregion has been called a ``light point''. This terminology comes about because the null geodesic has constant spatial coordinates in any coordinate system in which the spatial coordinates are Lie transported with respect to the asymptotically timelike Killing vector field $T$.} which are sufficiently compact and rotating at a particular rate \cite{Grandclement2016}. An evanescent ergosurface is a timelike submanifold where an asymptotically-timelike Killing vector field, which is timelike everywhere else, becomes null\footnote{For asymptotically Kaluza-Klein manifolds, a different (but functionally similar) definition can be given.}. Thus, these submanifolds are similar to the boundary of an ergoregion, however, unlike an ergoregion, there is no ``interior'' where the asymptotically-timelike Killing vector fields becomes spacelike. Evanescent ergosurfaces are intimately related to questions of stability: for example, non-supersymmetric microstates have an ergoregion but no horizon, and so are susceptible to the ``ergoregion instability'' of Friedman \cite{Friedman1978} (recently proved rigorously in \cite{Moschidis2016}). However, \emph{supersymmeric} microstates do not have an ergoregion\footnote{There is some subtlety here: some supersymmetric microstate geometries also have an ergoregion, but this ergoregion only allows for negative energy waves if those waves also have some nonzero momentum in the Kaluza-Klein directions. In this work, when dealing with asymptotically Kaluza-Klein spacetimes, we will restrict our attention to waves which are invariant in the Kaluza-Klein directions.} but only an evanescent ergosurface, so it might be hoped that they avoid an instability, at least on the linear level. Similar comments hold in the boson star case: compact stars which rotate more rapidly than some critical rate admit an ergoregion and so are susceptible to the ergoregion instability, but stars rotating at precisely the critical rate only admit an evanescent ergoregion.

In \cite{Eperon2016} and \cite{Keir2016A}, particular geometries with evanescent ergosurfaces were studied, and various properties of waves propagating on these geometries were discussed. In particular, it was shown that a ``stable trapping'' phenomena occurs, causing waves to decay extremely slowly, and it is conjectured that this might lead to a nonlinear instability (see also \cite{Holzegel2014, Cunha2017}). In fact, waves on these geometries decay even slower than waves on other geometries exhibiting stable trapping, a feature which is related directly to the presence of the evanescent ergosurface \cite{Keir2016A}.

Here, we will take a much more general approach. Rather than studying a particular geometry, we will study a \emph{general} manifold with an evanescent ergosurface. As far as possible, we avoid placing other restrictions on the manifold: we require a suitable asymptotic structure and smoothness properties, and we also require \emph{either} a certain kind of discrete isometry (satisfied, for example, by $(t-\phi)$-symmetric spacetimes) \emph{or} an additional Killing field with suitable properties. Note that we are able to deal with both asymptotically flat and asymptotically Kaluza-Klein manifolds. Under these very general conditions, we are able to show that a kind of instability is present, which is (in a sense) stronger than the ``slow decay'' results of \cite{Eperon2016} and \cite{Keir2016A}, but weaker than the ergoregion instability. The geometries that we study can also be expected to exhibit very slow decay of linear waves, and this alone might lead to the expectation of a nonlinear instability. However, the new instability which we find is of a different nature, and already appears at the \emph{linear} level.

The instability that we exhibit has a lot in common with the ``ergosphere instabiliity'' originally discovered by Friedman \cite{Friedman1978} and recently proved rigorously by Moschidis \cite{Moschidis2016}. This instability occurs in all asymptotically flat spacetimes with ergoregions \emph{but no event horizon}. Indeed, we can view the ``ergosurface instability'' as what is left over of the ergosphere instability, when the ergoregion degenerates into an evanescent ergosurface.

\vspace{4mm}

Let us now make some comments on the nature of the instability we show in this paper. First, we are focussing on \emph{scalar} perturbations, that is, we are examining solutions to the linear wave equation. This can either be viewed as a model for the Einstein(-matter) equations, which typically involve a set of nonlinear wave equations, or it can be viewed as a model for scalar fields or scalar modes of the geometry. 

Next, note that our instability is \emph{not} associated with an exponentially growing mode solution to the wave equation. Indeed, under the geometric conditions we assume, such a solution can in fact be ruled out. Nevertheless, we believe that the kind of behaviour we demonstrate can justifiably be called an instability, as we aim to show below.

Specifically, we are able to show that at least one of the following two cases occur:
\begin{enumerate}[label=(\Alph*)]
	\item \label{case (A)} Given a stationary family\footnote{i.e. a family of observers moving along integral curves of some vector field $N$, where the Lie derivative of $N$ along the asymptotically timelike Killing vector field vanishes. Note that this does \emph{not} mean that each member of the family moves parallel to the asymptotically timelike Killing vector field!} of observers moving along timelike curves, and given any constant $C > 0$, there exist waves, arising from smooth, compactly supported initial data (depending on the constant $C$), such that initially the \emph{total} energy measured by the entire family of observers is arbitrarily small, but, after some time has passed, the total energy measured by these observers is at least $C$. Moreover, the energy density measured by the observers in a neighbourhood of the ergosurface is $\mathcal{O}(C)$.
	\item \label{case (B)} The spacetime exhibits an \emph{Aretakis-type} instability (see \cite{Aretakis2011}, \cite{Aretakis2012a}), where there are waves arising from smooth, compactly supported initial data, whose local energy \footnote{That is, the energy measured by a subset of the family of observers mentioned above, which is such that the worldlines of these observers intersect a spacelike hypersurface in a compact set.} \emph{fails to decay} in a neighbourhood of the ergosurface, although it decays everywhere else. In fact, a non-zero amount of energy is concentrated in a smaller and smaller region, leading to pointwise blow-up.
\end{enumerate}

Note that the ``energy'' measured by the family of observers referred to above is \emph{not} the energy measured with respect to the asymptotically timelike Killing vector field, which (since it is a Killing field) is conserved. In fact, the family of observers referred to above \emph{cannot} move parallel to the asymptotically timelike Killing vector field, since this is vector field is null (rather than timelike) on the evanescent ergosurface.

Note also that, although we cannot rule it out in general, we do not know of a particular case where behaviour of type \ref{case (B)} is exhibited. This is in contrast to the behaviour of type \ref{case (A)}, which (as we will show) is exhibited by the supersymmetric microstate geometries studied in \cite{Maldacena2000, Balasubramanian2000, Lunin2002, Giusto2004, Giusto2004A, Bena2005, Berglund2005, Gibbons2013, Eperon2016, Keir2016A}.

\vspace{4mm}

Of the two possible instability scenarios outlined above, we can guarantee that we have an instability of type \ref{case (A)} if there exists \emph{another} Killing vector field (in addition to the asymptotically timelike one) such that the span of these two Killing vector fields is timelike in a neighbourhood of the evanescent ergosurface\footnote{Note that this is the case for the supersymmetric microstate geometries investigated in \cite{Eperon2016}, \cite{Keir2016A}.}. In fact, the presence of an extra Killing field of this kind allows us to show a number of other details of the instability. In particular, we can show that
\begin{itemize}
	\item Despite the behaviour outlined in point \ref{case (A)} above, the local energy \emph{is} bounded, but it is bounded in terms of a \emph{higher-order} initial energy, and not in terms of the initial energy.
	\item However, if we know the initial higher-order energy, then at later times this same higher-order energy can become arbitrarily large.
	\item If we want the energy measured by our family of observers to be amplified by a factor of $C$, then this can be achieved in some time which is bounded by $\exp(C^{1+\epsilon})$ for some $\epsilon > 0$.
	\item There exist (possibly non-smooth and non-compactly supported) initial data leading to a (weak) solution of the wave equation with \emph{unbounded} local energy.
	\item Also, in the presence of this additional symmetry, we do not require the discrete isometry.	
\end{itemize}

\subsection{Brief overview of the instability}

We will very briefly sketch the construction, made rigorous later in this paper, which underlies our instability result.

First, we need to invoke two notions of the energy of a solution to the wave equation: the \emph{non-degenerate energy}, which we call $\mathcal{E}^{(N)}$, which is the energy measured by a family of timelike-moving observers, and the \emph{conserved energy} associated with the Killing field, which we call $\mathcal{E}^{(T)}$. We also define the \emph{local, non-degenerate energy} $\mathcal{E}^{(N)}_{\mathcal{U}}$, which is the same as the non-degenerate energy except that it is only evaluated on some subset of the space at each moment of time\footnote{Specifically, the local energy $\mathcal{E}^{(N)}_{\mathcal{U}}$ is evaluated on the subset $\mathcal{U}\cap\Sigma$, where $\Sigma$ is the spacelike hypersurface defining ``space at a given time'', and $\mathcal{U}$ is some subset of the manifold which is invariant under the flow generated by the stationary Killing field. We require that $\mathcal{U}\cap\Sigma$ is (pre)compact.}. We choose this subset to be a neighbourhood of the ergosurface at each point in time.

Because of the presence of the ergosurface, we find that it is possible to construct initial data for the wave equation for which $\mathcal{E}^{(N)}_{\mathcal{U}}$ (and hence $\mathcal{E}^{(N)}$) is very large, and yet $\mathcal{E}^{(T)}$ is very small. Let us begin with such data, at some time far in the future, say at time $\tau$. We will then evolve this data \emph{backwards} in time.

There are then two options: either the wave disperses in the past, and $\mathcal{E}^{(N)}_{\mathcal{U}} \rightarrow 0$ as we go backwards in time, or else the local energy does not disperse. Let us discuss each of these cases in turn.

In the first case, the local energy $\mathcal{E}^{(N)}_{\mathcal{U}}$ decays as we evolve backwards in time. Since we began with data in the future for which the conserved energy $\mathcal{E}^{(T)}$ is very small, it will remain the case that $\mathcal{E}^{(T)}$ is small as we evolve backwards in time. At the same time, by assumption, the local energy $\mathcal{E}^{(N)}_{\mathcal{U}}$ also becomes small, say at time $t = 0$. This means that, at time $0$, both $\mathcal{E}^{(N)}_{\mathcal{U}}$ and $\mathcal{E}^{(T)}$ are very small. But we also find that the \emph{total} non-degenerate energy, $\mathcal{E}^{(N)}$, can be expressed as
\begin{equation*}
	\mathcal{E}^{(N)} \sim \mathcal{E}^{(N)}_{\mathcal{U}} + \mathcal{E}^{(T)}
\end{equation*}
Hence, at time $0$, the total, non-degenerate energy $\mathcal{E}^{(N)}$ is also very small.

Hence we have arrived at data at time $0$ with a very small total, non-degenerate energy $\mathcal{E}^{(N)}$. In fact, all three energies, $\mathcal{E}^{(N)}$, $\mathcal{E}^{(N)}_{\mathcal{U}}$ and $\mathcal{E}^{(T)}$ are very small at time $0$. And yet, if we evolve this data forward in time, we know that at time $\tau$, the local energy $\mathcal{E}^{(N)}_{\mathcal{U}}$ becomes very large. If we now interpret this solution as a solution to the forward-in-time problem, then we find that both the non-degenerate energy $\mathcal{E}^{(N)}$ and the local non-degenerate energy $\mathcal{E}^{(N)}_{\mathcal{U}}$ have been ``amplified'' by a very large factor. See figure \ref{figure instability} for a sketch explaining this case (case \ref{case (A)}).

Alternatively, it might be the case that there are some solutions to the wave equation which do not disperse as we evolve them to the past, and so $\mathcal{E}^{(N)}_{\mathcal{U}}$ does not approach zero as we evolve backwards in time. In this case, we can invoke the discrete isometry to obtain a solution to the wave equation for which the local energy does not approach zero as we evolve \emph{to the future}. However, it is possible to show that the local energy does, in fact, tend to zero in the future, at least in every compact set which is positioned away from the ergosurface. Hence, in this case, a finite amount of energy must eventually be contained within an arbitrarily small region of space near the ergosurface. We call this an Aretakis-type instability.

\begin{figure}[htbp]
	\centering
	\includegraphics[width = 0.5\linewidth, keepaspectratio]{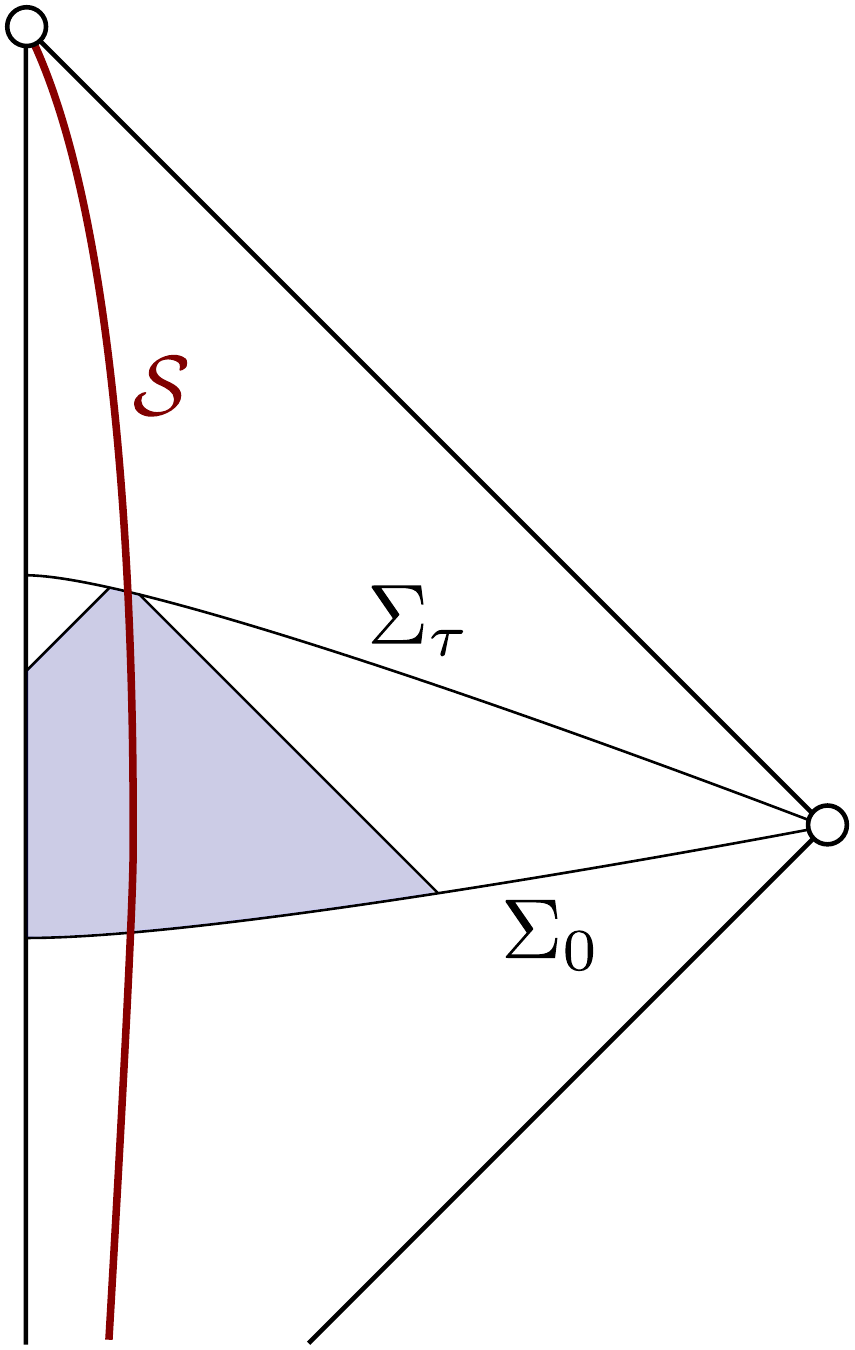}
	\caption[figure1]{
		A sketch of a Penrose diagram illustrating the construction of the instability in case \ref{case (A)}, that is, when the local energy of waves decays to the past. We begin by constructing initial data on the hypersurface $\Sigma_{\tau}$ such that the local non-degenerate energy $\mathcal{E}^{(N)}_{\mathcal{U}}$ at time $\tau$ is very large, say $\epsilon^{-1}$, for some arbitrarily small $\epsilon$. At the same time, the conserved energy $\mathcal{E}^{(T)}$ at time $\tau$ is very small, say $\epsilon$. This data can also be chosen to be supported only in a neighbourhood of the ergosurface $\mathcal{S}$.
	\\	
		We then evolve this data backwards in time to the hypersuface $\Sigma_0$. If $\Sigma_\tau$ is sufficiently far in the future, then we know that the local energy $\mathcal{E}^{(N)}_{\mathcal{U}}$ on the hypersurface $\Sigma_0$ will be very small - say $\epsilon$. The conserved energy $\mathcal{E}^{(T)}$ on the hypersurface $\Sigma_0$ is also $\epsilon$. It follows that the total, non-degenerate energy $\mathcal{E}^{(N)}$ measured on the hypersurface $\Sigma_0$ is $\mathcal{O}(\epsilon)$. Hence, between the hypersurfaces $\Sigma_0$ and $\Sigma_\tau$, the non-degenerate energy $\mathcal{E}^{(N)}$ has been ``amplified'' by a factor of $\epsilon^{-2}$.
	\\	
		 Note that a solution to the wave equation constructed in this way is supported only in the shaded (blue) region. In particular, the initial data on the hypersurface $\Sigma_0$ is compactly supported.
	}
	\label{figure instability}
\end{figure}

\subsection{Comparison with the ergosphere instability}
The ``evanescent ergosurface instability'' is somewhat weaker than the ``ergosphere instability'' of \cite{Friedman1978, Moschidis2016}. In particular, \cite{Moschidis2016} showed that, when an ergoregion is present but no event horizon exists, then there are solutions to the wave equation, arising from smooth, compactly supported initial data, whose local energy is unbounded. This is not the case for the evanescent ergosurface instability - indeed, in the case where an extra symmetry is present, we can actually rule out this kind of behaviour. On the other hand, if we allow for non-compactly supported data and we do not require that ``higher order'' energies are finite, then we can recover similar behaviour, although the rate of growth of the energy will generally be much slower in the evanescent ergosurface case.

The key idea behind the ergosphere instability of \cite{Friedman1978, Moschidis2016} is to use the ergoregion to construct initial data for the wave equation with \emph{negative conserved energy}. Then, under the assumption that the non-degenerate energy $\mathcal{E}^{(N)}$ remains uniformly bounded over time, it is possible to show (see \cite{Moschidis2016}) that the local non-degenerate energy $\mathcal{E}^{(N)}_{\mathcal{U}}$ must decay, at least away from the ergoregion. It is then possible to  derive a contradiction with the conservation of the (negative) conserved energy, which ensures that some part of the wave always remains trapped within the ergoregion.

Our approach is similar in many ways, and for this reason we shall make use of many of the results of \cite{Moschidis2016}. However, since we only have an evanescent ergosurface rather than a full ergoregion, it is not possible to produce waves with negative conserved energy. Instead, we can make use of the evanescent ergosurface to construct data for the wave equation such that its conserved energy is \emph{much smaller} than its non-degenerate energy. This is the key fact which, as we show, leads to some kind of instability.

\subsection{Comparison with extremal black holes}

One might wonder whether the ideas in this paper can be applied to extremal black holes. After all, the event horizon of an extremal black hole bears many similarities with an evanescent ergosurface. However, there are several technical reasons why our construction fails in this case: for example, the horizon is a null hypersurface rather than a timelike hypersurface.

Heuristically, we can understand the failure of our construction in the case of black holes in the following way. Our result relies crucially on being able to evolve data \emph{backwards in time}: it must be the case that the time-reversed manifold ``looks similar'' to the original manifold in a suitable sense. This is guaranteed if the manifold admits a discrete isometry of the required kind, or if it has an additional Killing field with certain properties. However, in the case of a black hole, it fails spectacularly. If we begin with a hypersurface which intersects the future horizon and then evolve a solution to the wave equation \emph{backwards} in time, then it behaves very differently from solutions which are evolved forwards in time. For example, at least in subextremal black holes, the redshift effect means that the energy of a wave near the event horizon decays when evolved to the future;  when evolved to the past, the energy will instead be blue-shifted. If, instead, we begin with a hypersurface that intersects the bifurcation sphere, then we do not expect the energy to decay, since we are not really ``evolving'' the data in this region when we flow along curves of the stationary Killing field.

Nevertheless, one kind of instability (case \ref{case (B)}) that might be exhibited by spacetimes with an evanescent ergosurface has a lot in common with the \emph{Aretakis instability} of extremal horizons \cite{Aretakis2011, Aretakis2012a}. In both cases, there is some non-decaying quantity on a specific hypersurface (either the event horizon or the evanescent ergosurface) which decays everywhere else, and this is responsible for a certain kind of blow-up. Note, however, that while we cannot rule out this kind of behaviour in general, we can rule it out on manifolds which have some extra symmetry. In many of the explicit examples of spacetimes with evanescent ergosurfaces, this extra symmetry is present, and so we can actually rule out this kind of instability. Instead, on these manifolds we have a different kind of instability, wherein the local energy of waves can be amplified by an arbitrarily large amount.

\section{Notation}

In this section, we will often refer to ``the asymptotically timelike Killing vector field $T$''. In all these cases, when considering the asymptotically Kaluza-Klein case, the vector field $T$ should be replaced by the vector field $V$.

We use the following notation for inequalities: we write $A \lesssim B$ if there is some constant $C > 0$, \emph{independent} of all of the parameters which we are varying, such that
\begin{equation*}
	A \leq CB
\end{equation*}
Similarly, we write $A \gtrsim B$ if there is a constant $C > 0$, \emph{independent} of all of the parameters which we are varying, such that
\begin{equation*}
	B \leq CA
\end{equation*}
Also, we write $A \sim B$ is there are constants $c > 0$, $C > 0$, again independent of the parameters which are varying, such that
\begin{equation*}
	cB \leq A \leq CB
\end{equation*}

\vspace{3mm}

When integrating over a spacelike hypersurface $\Sigma_t$, we will use the notation $\dVol$ for the volume form induced on $\Sigma_t$ by the spacetime metric $g$. Also, given a set $\mathcal{U}_t \subset \Sigma_t$ we define
\begin{equation*}
	\text{Volume}(\mathcal{U}_t) := \int_{\mathcal{U}_t} \dVol
\end{equation*}

\vspace{3mm}

Given some subset $U_t \subset \Sigma_t \subset \mathcal{M}$ of a hypersurface $\Sigma_t$, we define the corresponding set $U$ as the set of $T$-translations of $U_t$, where $T$ is the asymptotically timelike Killing vector field. In other words, we define
\begin{equation*}
	U := \left\{ x \in \mathcal{M} \ \big| \  \exists \text{ an integral curve of } T \text{ through } x \text{ and } U_t \right\}
\end{equation*}

Similarly, if we have a foliation of $\mathcal{M}$ by hypersurfaces $\Sigma_t$, and if $U_t \subset \Sigma_t$, then we define $U_{\tau}$ by
\begin{equation*}
	U_{\tau} := \left\{ x \in \Sigma_{\tau} \ \big| \  \exists \text{ an integral curve of } T \text{ through } x \text{ and } U_t \right\}
\end{equation*}
See figure \ref{figure spacetime regions} for a sketch of these regions.

\begin{figure}[htbp]
	\centering
	\includegraphics[width = 0.7\linewidth, keepaspectratio]{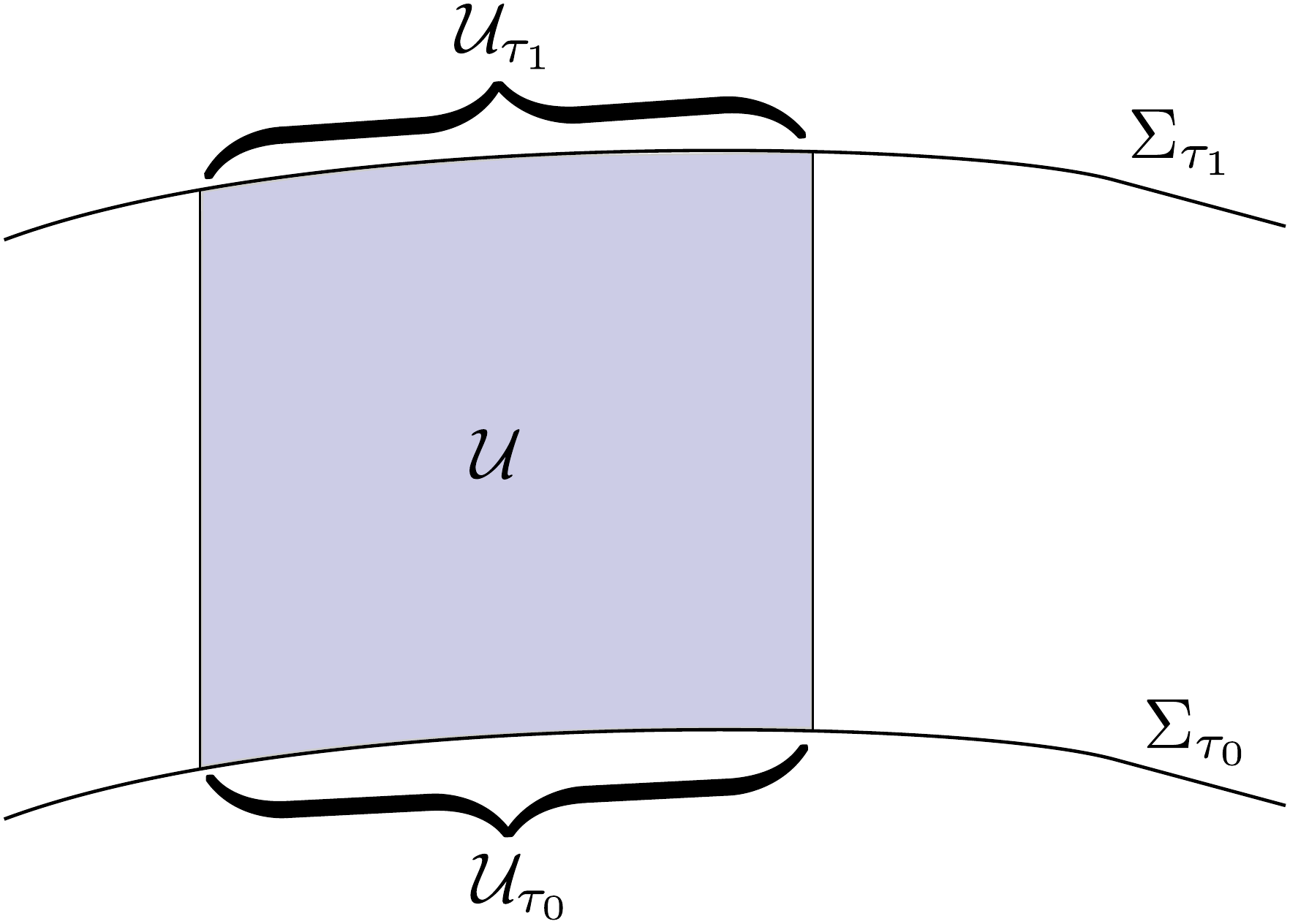}
	\caption{
		A sketch showing two spacelike hypersurfaces $\Sigma_{\tau_0}$ and $\Sigma_{\tau_1}$, the regions $\mathcal{U}_{\tau_0}$ and $\mathcal{U}_{\tau_1}$ and  the part of the region $\mathcal{U}$ (in blue) which lies between these two regions. In this sketch, the isometry generated by the vector field $T$ is represented by translation up the page.
	}
	\label{figure spacetime regions}
\end{figure}

Also, we define the ``$\delta$-thickening'' of the set $\mathcal{U}$ as follows: for $\delta > 0$, we define
\begin{equation*}
\begin{split}
	(\mathcal{U}_{(\delta)})_0 &:= \left\{ x \in \Sigma_0 \ \big| \ \text{dist}(x, \mathcal{U}_0) < \delta \right\}
\end{split}
\end{equation*}
where $\text{dist}(x, \mathcal{U})$ is the distance from $x$ to the set $\mathcal{U}$ defined using the Riemannian metric induced on $\Sigma_0$. Then, we define $\mathcal{U}_{(\delta)}$ to be the set consisting of all the $T$-translates of the set $(\mathcal{U}_{(\delta)})_0$.

\vspace{3mm}

We use $\nabla$ to denote the covariant derivative induced by the metric $g$. We also write the geometric wave operator as
\begin{equation*}
	\Box_g \phi := (g^{-1})^{\mu\nu} \nabla_\mu \nabla_\nu \phi
\end{equation*}
and we will also use the notation $\mathcal{L}_X$ to denote the Lie derivative with respect to a vector field $X$.

\vspace{3mm}

The notation $\partial \phi$, for some scalar field $\phi$, will be used to denote (schematically) the collection of first derivatives of $\phi$. To be more precise: assuming that the asymptotically timelike Killing field $T$ is transverse to the spacelike hypersurface $\Sigma$, we define
\begin{equation*}
	\partial \phi := \{ (T\phi) \ , \ (e_a \phi) \}
\end{equation*}
where the $e_a$ are an orthonormal frame for the tangent space of $\Sigma$. Similarly, we define norms:
\begin{equation*}
	|\partial \phi| := \sqrt{ |T\phi|^2 + \sum_a |e_a \phi|^2 }
\end{equation*}

\vspace{3mm}

Given a vector field $X$, we will say that $X$ is \emph{uniformly timelike} if there are positive constants $c$ and $C$ such that
\begin{equation*}
 c \leq -g(X, X) \leq C
\end{equation*}

Also, given a vector $X$, we define the covector $X^{\flat}$ by defining its action on a arbitrary vector field $Y$:
\begin{equation*}
	X^{\flat}(Y) := g(X, Y)
\end{equation*}
Similarly, given a covector $\omega$, we define the vector $\omega^{\sharp}$ by the prescription
\begin{equation*}
	g(\omega^{\sharp}, Y) := \omega(Y)
\end{equation*}
for all vectors $Y$. 

Finally, given a set $\mathcal{U}_0 \subset \Sigma_0$ and the associated set $\mathcal{U} \subset \mathcal{M}$, we define the notation
\begin{equation*}
	|| f ||_{L^p[\mathcal{U}]}(\tau) := \left( \int_{\mathcal{U} \cap \Sigma_\tau} |f|^p \dVol \right)^{\frac{1}{p}} 
\end{equation*}
and similarly
\begin{equation*}
	|| f ||_{L^\infty[\mathcal{U}]}(\tau) := \esssup_{\mathcal{U} \cap \Sigma_\tau} |f|
\end{equation*}

\section{Asymptotic structure}
\label{section asymptotic structure}
We are interested in smooth\footnote{But not necessarily analytic!}, stationary metrics that are either asymptotically flat or asymptotically Kaluza-Klein. We shall restrict attention to manifolds without an event horizon or a black hole region, so we require the manifold $\mathcal{M}$ to be identical to the causal past of future null infinity. In the Kaluza-Klein case, we also restrict to those asymptotically Kaluza-Klein metrics that are ``homogeneous'' in the compact directions, a notion that we will make precise below.

\subsection{Asymptotically flat, stationary spacetimes}
\label{subsection asymptotically flat}

Since we make use of the results of \cite{Moschidis2016}, we shall use the same definition of an ``asymptotically flat'' stationary manifold given in \cite{Moschidis2016}. That is, we shall consider a stationary $(d+1)$-dimensional manifold in to be \emph{asymptotically flat} if, there is an open region $\mathcal{U}^{\text{as}}$, with compact complement, called the ``asymptotic region'', which is diffeomorphic to $\mathbb{R} \times (\mathbb{R}^{d}\setminus B)$, for some integer $d \geq 3$, where $B$ denotes the unit open ball in $\mathbb{R}^d$. Moreover, we require that there exist coordinates for $\mathcal{U}^{\text{as}}$ such that the metric on $\mathcal{M}$ takes the form
\begin{equation}
\label{equation asymptotically flat}
 g = -\left( 1 - \frac{2M}{r} + h_1(r, \sigma) \right) \upd t^2 + \left( 1 + \frac{2M}{r} + h_2(r,\sigma) \right) \upd r^2 + r^2 \left( g_{\mathbb{S}^{d-1}} + h_3(r,\sigma) \right) + h_4(r,\sigma) \upd t
\end{equation}
In the expression above $r$ is the pull-back of the standard Euclidean radial function on $(\mathbb{R}^d\setminus B)$ by the diffeomorphism mentioned above. Moreover, there is a (related) diffeomorphism from the asymptotic region to $\mathbb{R}\times \mathbb{R}_+ \times \mathbb{S}^{d-1}$, which can be understood as a ``polar coordinate chart'' for $\mathcal{U}^{\text{as}}$. The function $\sigma$ is then the projection $\sigma : \mathcal{U}^{\text{as}} \rightarrow \mathbb{S}^{d-1}$ naturally defined by this diffeomorphism.

Note that the constant $M$ in the above formula corresponds to the mass of the spacetime \emph{if} the spacetime dimension is 4, i.e.\ if $d = 3$.

In equation (\eqref{equation asymptotically flat}) above, $h_1$ and $h_2$ are smooth functions from $\mathcal{U}^{\text{as}}$ to $\mathbb{R}$. $g_{\mathbb{S}^{d-1}}$ denotes the standard round metric on the \emph{unit} $(d-1)$-sphere. For each value of $r$, $h_3(r, \cdot)$ is a smooth, symmetric rank $(0,2)$ tensor field on the $(d-1)$-sphere. Finally, for each value of $r$, $h_4$ is a smooth one-form on the $(d-1)$ sphere.

In addition, the functions $h_1$, \ldots $h_4$ are required to decay at suitable rates as $r \rightarrow \infty$. Using the round metric $g_{\mathbb{S}^{d-1}}$ we can measure the norm of tensor fields on the sphere $\mathbb{S}^{d-1}$ in the obvious way, and the decay rates required are as follows: for some $\alpha \in (0, 1]$,
\begin{equation*}
 \begin{split}
  h_1, h_2, h_3 &= \mathcal{O}_4(r^{-1-\alpha}) \\
  h_5 &= \mathcal{O}_4(r^{-\alpha})
 \end{split}
\end{equation*}
where a function $h$ is said to be $\mathcal{O}_n(r^{k})$ if there is a constant $C$ such that
\begin{equation*}
 \sum_{j = 0}^{n} \, \sum_{j_1 + j_2 + j_3 = j} r^{j_1 + j_2} \left|\partial_r^{j_1} \partial_t^{j_2} \partial_{\sigma}^{j_3} h \right| \leq C r^{k}
\end{equation*}
where $\partial_\sigma$ is schematic notation for the action of some vector field on $\mathbb{S}^{d-1}$ with unit norm (measured, as usual, with respect to the round metric $g_{\mathbb{S}^{d-1}}$).

Finally, the diffeomorphism mapping $\mathcal{U}^{\text{as}} \rightarrow \mathbb{R} \times (\mathbb{R}^{d}\setminus B)$ will be labelled as $\varphi_{\text{as}}$.

\subsection{Asymptotically Kaluza-Klein, stationary manifolds}
\label{subsection asymptotically KK}
As mentioned above, the asymptotically Kaluza-Klein manifolds we shall study will be ``homogeneous'' in the compact directions. We now make this notion precise:

\begin{definition}[Kaluza-Klein manifolds]
Let $\mathcal{M}$ be a smooth manifold. Let $\mathcal{G}$ be a compact Lie algebra $\mathcal{G}$. Let 
\begin{equation*}
 \varphi_{\mathcal{G}}: \mathcal{G}\times \mathcal{M} \rightarrow \mathcal{M}
\end{equation*}
be a smooth action of $\mathcal{G}$ on the manifold $\mathcal{M}$ \emph{by isometries}.

We require\footnote{It would be possible to relax this requirement and instead require only that the leaves of the foliation are $\mathcal{G}$-invariant in a neighbourhood of the evanescent ergosurface.} that there exists a foliation of $\mathcal{M}$ by spacelike hypersurfaces which are invariant under the action of $\mathcal{G}$. That is, $\mathcal{M}$ is foliated by hypersurfaces $\Sigma_t$ such that $\Sigma_t$ are the level sets of a function $t$, where $t$ is invariant under the action of $\mathcal{G}$. Hence the smooth action of $\mathcal{G}$ on $\mathcal{M}$ descends to a smooth action on $\Sigma_t$.

Let $X$ be a smooth compact manifold that is a homogeneous space for $\mathcal{G}$, that is, there is a transitive action of $G$ on $X$. Then $X$ is diffeomorphic to $\mathcal{G}/\mathcal{H}$ for some subgroup $\mathcal{H}$ of $\mathcal{G}$ (specifically, $\mathcal{H}$ is the isotropy subgroup of a point in $X$). So the points in $X$ can be represented as left cosets of $\mathcal{H}$.

As before, we require the existence of an ``asymptotic region'': an open set $\mathcal{U}^{\text{as}} \subset \mathcal{M}$, with a compact complement, such that $\mathcal{U}^{\text{as}}$ is diffeomorphic to $\mathbb{R}\times(\mathbb{R}^{d}\setminus B)\times X$, for some $d \geq 3$. Note that $\mathcal{M}$ is not required to be \emph{globally} a product space. The action of $\mathcal{G}$ on the asymptotic region is given explicitly by
\begin{equation*}
 \begin{split}
  \mathcal{G} : \mathbb{R}\times(\mathbb{R}^{d}\setminus B)\times X &\rightarrow \mathbb{R}\times(\mathbb{R}^{d}\setminus B)\times X \\
  g_1: (t, x, g_2H) &\mapsto (t, x, g_1 g_2 H)
 \end{split}
\end{equation*}
i.e.\ $\mathcal{G}$ acts on the left on $X$ viewed as a coset space, and does not affect the $t$ or $x$ coordinates. Moreover, since $\mathcal{G}$ acts by isometries, we require $X$ to be equipped with a $\mathcal{G}$-invariant, Riemannian metric $g_X$. We can then use $g_X$ to define the norm of tensor fields on $X$, and for tensor fields on $(\mathbb{S}^{d-1}\times X)$ we can measure the norms using a combination of both $g_{\mathbb{S}^{d-1}}$ and $g_X$.

Now, in $\mathcal{U}^{\text{as}}$ we require there to exist coordinates such that the metric on $\mathcal{M}$ takes the form
\begin{equation}
 \label{equation asymptotically KK}
 \begin{split}
   g &= -\left( 1 - \frac{2M}{r} + h_1(r, \sigma) \right) \upd t^2 + \left( 1 + \frac{2M}{r} + h_2(r,\sigma) \right) \upd r^2 + r^2 \left( g_{\mathbb{S}^{d-1}} + h_3(r,\sigma) \right) + h_4(r,\sigma) \upd t \\
   &\phantom{=} + \left( g_X + h_5(r,\sigma) \right) + h_6(r,\sigma)\upd t + h_7(r,\sigma)\upd r + h_8(r,\sigma)
  \end{split} 
\end{equation}
where the functions $h_1$, \ldots $h_4$ are of the same type, and satisfy the same bounds, as in the asymptotically flat case, while the other functions are defined as follows:
\begin{itemize}
 \item for each $r$, $\sigma$, $h_5$ is a symmetric, $\mathcal{G}$-invariant rank $(0,2)$ tensor on $X$ satisfying the bound $h_5 = \mathcal{O}_4(r^{-1-\alpha})$
 \item for each $r$, $\sigma$, $h_6$ and $h_7$ are $\mathcal{G}$-invariant one-forms on $X$, all of which are $\mathcal{O}_4(r^{-1-\alpha})$
 \item for each $r$, $h_8(r, \cdot)$ is a smooth, symmetric rank $(0,2)$ tensor on $(\mathbb{S}^{d-1}\times X)$ that is $\mathcal{O}_4(r^{-\alpha})$ and which is invariant under the natural action of $\mathcal{G}$ on such tensor fields
\end{itemize}

The asymptotically timelike Killing vector field, $T$, is given in the coordinates above by $\partial_t$. Note that the above definition ensures that this vector field is invariant under the action of $\mathcal{G}$ in the asymptotic region; we require that it is, in fact, globally invariant under the action of the group $\mathcal{G}$. In other words, for any $g \in \mathcal{G}$, we require that
\begin{equation*}
 (\varphi_g )_* (T) = T
\end{equation*}
To put this in yet another way, we require the action of the real line $\mathbb{R}$ generated by $T$ to commute with the action of the group $\mathcal{G}$.

\end{definition}

In the definition above, we view the tensor $g_X$ as the ``limiting'' metric on the Kaluza-Klein fibres: asymptotically, the metric approaches $g_{\infty} = -\upd t^2 + \upd r^2 + r^2 g_{\mathbb{S}^{d-1}} + g_X$.

Note that we can view the continuous action of $\mathcal{G}$ on $\mathcal{M}$ as being generated by a set of vector fields $L_{\mathcal{G}}$, which, since $\mathcal{G}$ acts by isometries, are Killing fields. Note that, since $\mathcal{G}$ acts transitively on $X$, the pushforward of the vector fields $L_{\mathcal{G}}$ by the diffeomorphism $\varphi_{\text{as}}$ span the tangent space of $X$ at each point in the asymptotic region. In fact, if $\mathfrak{g}$ is the Lie algebra associated with $\mathcal{G}$, then to each element $A \in \mathfrak{g}$ we have a smooth one-parameter family of maps
\begin{equation*}
 \begin{split}
  \phi_s(A): \mathcal{M} &\rightarrow \mathcal{M} \\
  \phi_s(A): x &\mapsto \varphi\left(\exp(s A) , x \right)
 \end{split}
\end{equation*}
for $s \in \mathbb{R}$. The vector field generating the map $\phi_s(A)$ is denoted by $L_A$. Then we have
\begin{equation*}
 L_{\mathcal{G}} = \left\{ L_A \big| A \in \mathfrak{g} \right\}
\end{equation*}
Although the vector fields $L_A$ do not necessarily commute with one another, they do all commute with $T$, i.e.\ $[L_{\mathcal{G}}, T] = 0$.

Note that the condition that the tensor fields $h_1, \ldots, h_8$ are invariant under the action of $\mathcal{G}$ mean that the Lie derivatives of these tensor fields with respect to the vector fields $L_A$, $A \in \mathfrak{g}$ vanish, and so this gives the required bounds on the derivatives in the ``Kaluza-Klein directions''.

\subsection{An example: 3-charge microstate geometries}

Our definition of an ``asymptotically Kaluza-Klein'' manifold is somewhat technical: we will illustrate it with the example of $3$-charge microstate geometries of \cite{Maldacena2000, Balasubramanian2000, Lunin2002, Giusto2004, Giusto2004A, Bena2005, Berglund2005, Gibbons2013}. We will also refer back to this example when discussing evanescent ergosurfaces.

The metric is given by
\begin{equation}
\label{equation microstate geometry}
	\begin{split}
		g &=
		-\frac{1}{h} (\upd t^2 - \upd z^2)
		+ \frac{Q_p}{hf} (\upd t - \upd z)^2
		+ hf \left( \frac{ \upd r^2}{r^2 + (\tilde{\gamma}_1 + \tilde{\gamma}_2)^2 \eta} + \upd \theta^2 \right)
		\\
		&\phantom{=}
		+ h\left(
			r^2
			+ \tilde{\gamma}_1 (\tilde{\gamma}_1 + \tilde{\gamma}_2)\eta
			- \frac{(\tilde{\gamma}_1^2 - \tilde{\gamma}_2^2)\eta Q_1 Q_2 \cos^2 \theta}{h^2 f^2}
		\right) \cos^2 \theta \upd \psi^2
		\\
		&\phantom{=}
		+ h\left(
			r^2
			+ \tilde{\gamma}_2 (\tilde{\gamma}_1 + \tilde{\gamma}_2)\eta
			- \frac{(\tilde{\gamma}_1^2 - \tilde{\gamma}_2^2)\eta Q_1 Q_2 \sin^2 \theta}{h^2 f^2}
		\right) \sin^2 \theta \upd \phi^2
		\\
		&\phantom{=}
		+ \frac{Q_p (\tilde{\gamma}_1 + \tilde{\gamma}_2)^2 \eta^2}{hf} (\cos^2 \theta \upd \psi + \sin^2 \theta \upd \phi)^2
		\\
		&\phantom{=}
		-2\frac{\sqrt{Q_1 Q_2}}{hf} \left( \tilde{\gamma}_1 \cos^2 \theta \upd \psi + \tilde{\gamma}_2 \sin^2 \theta \upd \phi \right) (\upd t - \upd z)
		\\
		&\phantom{=}
		-2\frac{ \tilde{\gamma_1} + \tilde{\gamma}_2)\eta \sqrt{Q_1 Q_2}}{hf} \left( \cos^2 \theta \upd \psi + \sin^2 \theta \upd \phi \right) \upd z
		+ \sqrt{\frac{H_1}{H_2}} \sum_{i=1}^4 \upd x_i^2
	\end{split}
\end{equation}
where
\begin{equation*}
	\begin{split}
		\eta &= \frac{Q_1 Q_2}{Q_1 Q_2 + Q_1 Q_p + Q_2 Q_p}
		\\
		\tilde{\gamma}_1 &= -an
		\\
		\tilde{\gamma}_2 &= a(n+1)
		\\
		f &= r^2 + (\tilde{\gamma}_1 + \tilde{\gamma}_2) \eta (\tilde{\gamma}_1 \sin^2 \theta + \tilde{\gamma}_2 \cos^2 \theta )
		\\
		H_1 &= 1 + \frac{Q_1}{f}
		\\
		H_2 &= 1 + \frac{Q_2}{f}
		\\
		h &= \sqrt{H_1 H_2}
	\end{split}
\end{equation*}
Here, $Q_1$, $Q_2$ and $a$ are independent constants, while $Q_p = a^2 n(n+1)$.

The coordinate $z \in [0, 2\pi R_z]$ parametrizes a circle of radius $R_z$, where $R_z = \frac{\sqrt{Q_1}{Q_2}}{a}$. The coordinates $x^i$, $i = 1,\ldots 4$ are also periodic: they parametrize a torus $T^4$, and their precise ranges are unimportant. The time coordinate is given by $t \in \mathbb{R}$, the radial coordinate is $r > 0$, and the coordinates $\theta$, $\phi$, $\psi$ parametrize a $3$-sphere, with $\theta \in [0, \pi/2]$, and $\phi$, $\psi \in [0, 2\pi]$.

This metric is asymptotically Kaluza-Klein. We have
\begin{equation*}
	g =
	-\upd t^2
	+ \upd r^2
	+ r^2 g_{\mathbb{S}^3}
	+ g_X
	+ \mathcal{O}(r^{-2})
\end{equation*}
where $g_X$ is given by
\begin{equation*}
	g_X = \upd z^2 + \sum_{i = 1}^4 \upd x_i^2
\end{equation*}
Note that $M = 0$ in this case.

For this metric, the group $\mathcal{G}$ is given by $\left( U(1) \right)^5$, and it acts by rotating the $z$ and $x_i$ coordinates. The space $X$ can be taken to be simply $\left( U(1) \right)^5$, parametrized by the coordinates $(z, x_i)$: we do not need to quotient out by some isotropy subgroup, since the isotropy subgroup is trivial in this case. In other words, the only group element which fixes a point on $X$ is the identity.

The vector fields $L_A$ can be written in terms of the coordinates defined above as
\begin{equation*}
	\begin{split}
		L_z &:= \frac{\partial}{\partial_z}
		\\
		L_i &:= \frac{\partial}{\partial x_i}
	\end{split}
\end{equation*}
Note that these are Killing fields.

Note also that, in addition to the asymptotically timelike Killing vector field $T = \frac{\partial}{\partial t}$, the metric also possesses a \emph{globally null} Killing vector field:
\begin{equation*}
	V = \frac{\partial}{\partial t} + \frac{\partial}{\partial z}
\end{equation*}

\section{Evanescent ergosurfaces}
\label{section evanescent ergosurfaces}

We shall say that an asymptotically flat (or asymptotically Kaluza-Klein) manifold posesses an evanescent ergosurface\footnote{Note that the conditions we impose also exclude manifolds with black hole regions or event horizons.} if either of the following conditions holds:

\begin{enumerate}[label=(ES\arabic*)]
\item \label{condition asymp flat}$\mathcal{M}$ is asymptotically flat and stationary, and the asymptotically timelike vector field $T$ (given by $\partial_t$ in the asymptotic region) is \emph{globally causal} and nonvanishing. Moreover, there is a submanifold $\mathcal{S}$ such that
	\begin{enumerate}[label=(\arabic*)]
		\item $\mathcal{S}$ is spatially compact (i.e.\ $\mathcal{S} \cap \Sigma_0$ is compact for some spacelike Cauchy surface $\Sigma_0$) and codimension at least one
		\item $T$ is null on $\mathcal{S}$ and timelike on $\mathcal{M}\setminus \mathcal{S}$
		\item for every open set $\mathcal{U}$ such that $\mathcal{S} \subset \mathcal{U}$, there exists some constant $c_{(T,\mathcal{U})} > 0$ such that 
			\begin{equation*}
				\inf_{x \in \mathcal{M}\setminus\mathcal{U}} |g(T,T)| \geq c_{(T,\mathcal{U})}
			\end{equation*}
		\item either $\mathcal{M}\setminus\mathcal{S}$ consists of a single connected component, \emph{or} $\mathcal{M}\setminus\mathcal{S}$ consists of at least two components, one\footnote{We can also deal with the case where $\mathcal{M}\setminus\mathcal{S}$ consists of a multiple components, if each one of the components includes an ``asymptotic region'', but for simplicity we stick with a single asymptotic region.} of which (which we call $\mathcal{M}_{(\text{ext})}$) includes the asymptotic region, and where $\mathcal{M} \setminus \mathcal{M}_{(\text{ext})}$ is precompact, and either
		\begin{itemize}
			\item there is some other Killing field $\Phi$ such that $[T, \Phi] = 0$ and the span of $T$ and $\Phi$ is timelike on $\mathcal{S}$, \emph{or}
			\item the manifold is real analytic in a neighbourhood of $\mathcal{S}$, \emph{or}
			\item $\mathcal{M} \setminus \mathcal{S}$ consists of two\footnote{Again, if there are multiple asymptotic regions, then this condition can be modified to allow for more components.} connected components: one which includes the asymptotic region, which we label $\mathcal{U}_{\text{asy}}$, and an ``enclosed'' region, which we label $\mathcal{U}_{\text{enc}}$. Then, we need the following unique continuation criteria: for every solution $\phi$ to the wave equation $\Box_g \phi = 0$, if $\phi \equiv 0$ on $\mathcal{U}_{\text{asy}}$, then $\phi = 0$ on all of $\mathcal{M}$.
		\end{itemize}
	\end{enumerate} 
	\item \label{condition asymp KK} $\mathcal{M}$ is asymptotically Kaluza-Klein in the sense of section \ref{section asymptotic structure}, and there is a \emph{globally causal} Killing vector field $V$ (which is not necessarily identical to the asymptotically timelike Killing vector field $T$). In addition, there is submanifold $\mathcal{S}$ such that 
		\begin{enumerate}[label=(\arabic*)]
			\item $\mathcal{S}$ is spatially compact (i.e.\ $\mathcal{S} \cap \Sigma_0$ is conpact for some spacelike Cauchy surface $\Sigma_0$) and codimension at least one
			\item $g(V,V)$ vanishes to at least second order on $\mathcal{S}$
			\item for all $L_A \in L_{\mathcal{G}}$, we have
				\begin{equation*}
					g(V, L_A)\big|_{\mathcal{S}} = 0
				\end{equation*}
			and, for all $x \in \mathcal{M}\setminus\mathcal{S}$ there exists some $L_A \in L_{\mathcal{G}}$ such that $g(V, L_A) \neq 0$
			\item for every open set $\mathcal{U}$ such that $\mathcal{S} \subset \mathcal{U}$, there exists some constant $c_{(V,\mathcal{U})} > 0$ such that 
				\begin{equation*}
					\inf_{x \in \mathcal{M}\setminus\mathcal{U}} \ \sup_{L_A \in L_{\mathcal{G}}} \frac{ |g(V, L_A)| }{\sqrt{g(L_A, L_A)}} \geq c_{(V,\mathcal{U})}
				\end{equation*}
			\item either $\mathcal{M}\setminus\mathcal{S}$ consists of a single connected component, \emph{or} $\mathcal{M}\setminus\mathcal{S}$ consists of at least two components, one of which (which we call $\mathcal{M}_{(\text{ext})}$) includes the asymptotic region, and where $\mathcal{M} \setminus \mathcal{M}_{(\text{ext})}$ is precompact, and either
		\begin{itemize}
			\item there is some other Killing field $\Phi$ such that $[T, \Phi] = 0$ and the span of $T$ and $\Phi$ is timelike on $\mathcal{S}$, \emph{or}
			\item the manifold is real analytic in a neighbourhood of $\mathcal{S}$, \emph{or}
			\item $\mathcal{M} \setminus \mathcal{S}$ consists of two\footnote{Again, if there are multiple asymptotic regions, then this condition can be modified to allow for more components.} connected components: one which includes the asymptotic region, which we label $\mathcal{U}_{\text{asy}}$, and an ``enclosed'' region, which we label $\mathcal{U}_{\text{enc}}$. Then, we need the following unique continuation criteria: for every solution $\phi$ to the wave equation $\Box_g \phi = 0$, if $\phi \equiv 0$ on $\mathcal{U}_{\text{asy}}$, then $\phi = 0$ on all of $\mathcal{M}$.
		\end{itemize}
	\end{enumerate} 
\end{enumerate}

\begin{remark}
\label{remark unique continuation}
	In each case, the first few points give a fairly straightforward definition of an evanescent ergosurface. However, the final point in each definition (that is, point $(4)$ in the asymptotically flat case and point $(5)$ in the asymptotically Kaluza-Klein case) requires some additional discussion. We include this point in the definition in analogy to assumption $(A1)$ of \cite{Moschidis2016}: it is made so that, in the case of an evanescent ergosurface which divides the manifold into an ``inside'' and an ``outside'', \emph{if} a wave decays on the outside, then it is also guaranteed to decay on the inside. If we do not make this assumption, then there is a third type of behaviour which is possible for waves on manifolds with evanescent ergosurfaces: the energy of the wave is neither amplified by an arbitrary amount, nor do  we have an Aretakis-like instability, but instead the wave approaches a solution to the wave equation which vanishes outside the ergosurface but is nonzero inside of it. This is only a viable alternative if there is a solution of the wave equation which is compactly supported within the ergosurface for all time, and it is ruled out by any of the conditions in the fourth point of our definitions of an evanescent ergosurface.
	
	Note that in many of the explicit cases of spacetimes with evanescent ergosurfacese - such as the supersymmetric microstates - the evanescent ergosurface is at least codimension two, and so there is no issue here.
\end{remark}

\subsection{The evanescent ergosurface on the example geometry}

Recall the metric \eqref{equation microstate geometry}, which we are using as an example to illustrate our definitions. In this case, we have already seen that the metric is asymptotically Kaluza-Klein. We have also seen that the vector field $V$ (given in terms of coordinate derivatives by $\partial_t + \partial_z$) is globally null.

If we define $L_i = \frac{\partial}{\partial x_i}$, then we find that $g(V, L_i) \equiv 0$. On the other hand, with $L_z = \partial_z$, we can compute
\begin{equation*}
	g(V, L_z) = \frac{1}{h}
\end{equation*}
Thus, this spacetime has an evanescent ergosurface if $h$ diverges somewhere, i.e.\ when $f$ = 0. Such a region does in fact exist and it has (spatial) topology $\mathbb{S} \times \mathbb{S}^3 \times T^4$ \cite{Gibbons2013}. In particular, it is a smooth, compact, co-dimension two submanifold.

We can also compute
\begin{equation*}
	\frac{ |g(V, L_z)| }{ \sqrt{ g(L_z, L_z) } } = \frac{1}{\sqrt{h} \sqrt{1 + \frac{Q_p}{f}}} 
\end{equation*}
so, if $f > 0$ (and hence $h < \infty$) this quantity is positive, as required.

Finally, we note that since the ergosurface is co-dimension two, $\mathcal{M} \setminus \mathcal{S}$ consists of a single connected component, so the final required property holds. In fact, the manifold is also analytic, so this gives another reason that the required ``unique continuation'' result holds.

\subsection{Null geodesics on the ergosurface}

In both cases, we have the following (see also \cite{Eperon2016}) :
\begin{proposition}
 There exists a null geodesic $\gamma$ lying entirely withing the evanescent ergosurface $\mathcal{S}$. In the asymptotically flat case $T$ is tangent to an affinely parameterised null geodesic on $\mathcal{S}$, while in the Kaluza-Klein case, $V$ is tangent to an affinely parameterised null geodesic on $\mathcal{S}$.
\end{proposition}
\begin{proof}
 Using the fact that $T$ is a Killing vector field, we have
\begin{equation*}
 \left(\nabla_T T \right)^\mu = -T^\nu \nabla^\mu T_\nu = -\frac{1}{2}\nabla^\mu (T^\nu T_\nu)
\end{equation*}
 Now, $g(T,T) = 0$ on $\mathcal{S}$ so any derivatives of $g(T,T)$ tangent to $\mathcal{S}$ vanish. Additionaly, since $T$ is a smooth vector field and the function $g(T,T)$ attains its maximum on $\mathcal{S}$, transverse derivatives of $g(T,T)$ also vanish on $\mathcal{S}$. Hence $\nabla_T T = 0$ on $\mathcal{S}$.

In the second case, since $V$ is a Killing vector field, we also have
\begin{equation*}
 \left(\nabla_V V \right)^\mu = \partial^\mu \left( g(V,V) \right)
\end{equation*}
since $g(V,V)$ vanishes to at least second order on $\mathcal{S}$, the right hand side vanishes on $\mathcal{S}$, so $V$ is tangent to affinely parameterised null geodesic. Now, let $\gamma$ be such a geodesic, with tangent $V$, passing through $\mathcal{S}$ at some point. For any $L_A \in L_{\mathcal{G}}$ we compute
\begin{equation*}
 \begin{split}
  V(g(V, L_A))\big|_{\mathcal{S}} &= g(\mathcal{L}_V V, L_A)\big|_{\mathcal{S}} + g(V, \mathcal{L}_V L_A)\big|_{\mathcal{S}} \\
  &= -g(V, \mathcal{L}_{L_A} V)\big|_{\mathcal{S}} \\
  &= -\frac{1}{2} L_A \left(g(V,V)\right) = 0
 \end{split}
\end{equation*}
where in the first and last lines we have used the fact that $V$ and $L_A$ are Killing fields for the metric $g$, respectively. Hence, if $g(V, L_A) = 0$ at some point of $\gamma$, then $g(V, L_A) = 0$ everywhere on $\gamma$, i.e.\ $\gamma$ remains within $\mathcal{S}$.
\end{proof}

\section{The discrete isometry}
\label{section the discrete isometry}

We require our manifolds to possess a discrete isometry $\mathscr{T}$ with the following properties:

\begin{itemize}
	\item $\mathscr{T}$ is an isometry, i.e.\
	\begin{equation*}
		\begin{split}
			\mathscr{T}: \mathcal{M} &\rightarrow \mathcal{M} \\
			\mathscr{T}^*(g) &= g
		\end{split}
	\end{equation*}
	\item $\mathscr{T}$ reverses the direction of Killing vector field $T$, i.e.\
	\begin{equation*}
		\mathscr{T}_*(T) = -T
	\end{equation*}
	\item There exists some \emph{spacelike} hypersurface $\Sigma_0$, which is a Cauchy surface for $\mathcal{M}$ and which is fixed by the discrete isometry $\mathscr{T}$
\end{itemize}

Note that we do not require that the \emph{points} of the hypersurface $\Sigma_0$ are invariant under the discrete isometry, which would lead to the requirement that the spacetime is static. For example, ``$(t-\phi)$-symmetric'' spacetimes are acceptable, in which there is a discrete isometry which maps $T$ to $-T$ and also $\Phi$ to $-\Phi$, where $\Phi$ is an axisymmetric Killing vector field. In this case, in terms of standard coordinates, the hypersurface $\{ t = 0 \}$ can be chosen as $\Sigma_0$: it is invariant under the discrete isometry, although a point with coordinates $\phi = \phi_1$ is mapped to a point with coordinates $\phi = -\phi_1$.

Note also that, in the case we have an additional symmetry of the right kind (see section \ref{section additional symmetry}) we actually do not need this additional symmetry. 

\subsection{The discrete isometry in the example geometry}

The example geometry with metric \eqref{equation microstate geometry} possesses a discrete isometry of the required kind. The map $\mathscr{T}$ can be specified using the coordinates: we have
\begin{equation*}
	\mathscr{T} (t, r, \theta, \phi, \psi, z, x_i)
	=
	(-t, r, \theta, \phi, \psi, 2\pi R_z - z, x_i)
\end{equation*}
In other words, we replace $t \mapsto -t$ and $z \mapsto 2\pi R_z - z$. This metric could be said to be ``$t-z$ symmetric''. We clearly have $\mathscr{T}_* V = -V$, since $V$ is given in coordinates by $\partial_t + \partial_z$.

Examining the form of the metric \eqref{equation microstate geometry} reveals that this is indeed an isometry. Moreover, this isometry clearly fixes the hypersurface $\Sigma_0 = \{t = 0\}$. This hypersurface is spacelike: we can compute
\begin{equation*}
	g^{-1}(\upd t, \upd t) = 
	-\frac{1}{hf} \left(
		f
		+ Q_1
		+ Q_2
		+ Q_p
		+ \frac{ Q_1 Q_2 + Q_1 Q_p + Q_2 Q_p }{ r^2 + (\tilde{\gamma}_1 + \tilde{\gamma}_2)^2 \eta }
	\right)
\end{equation*}
which turns out to be uniformly bounded and negative. Hence the isometry $\mathscr{T}$ given above fixes a spacelike hypersurface as required. We also find that the hypersurfaces $\Sigma_t$ are given by the level sets of $t$.

The discrete isometry descends to a discrete isometry on the hypersurface $t = 0$, which can also be given in coordinates:
\begin{equation*}
	\overline{\mathscr{T}} (r, \theta, \phi, \psi, z, x_i)
	=
	(r, \theta, \phi, \psi, 2\pi R_z - z, x_i)
\end{equation*}
It is easy to see that this is an isometry of the submanifold $t = 0$ equipped with the metric induced by $g$. This isometry can be viewed as a reflection in the hypersurface $z = 0$.

\section{The energy momentum tensor and energy currents}
\label{section energy momentum}

For any smooth, compactly supported function $\phi$ on a manifold $\mathcal{M}$ with metric $g$, we define the associated \emph{energy momentum tensor} $Q$:
\begin{equation*}
	Q_{\mu\nu}[\phi]:= (\partial_\mu \phi)(\partial_\nu \phi) - \frac{1}{2} g_{\mu\nu} g^{\alpha \beta} (\partial_\alpha \phi)(\partial_\beta \phi)
\end{equation*}
Given a smooth vector field $X$ we define the associated deformation tensor
\begin{equation*}
	^{(X)}\pi_{\mu\nu} := (\mathcal{L}_X g)_{\mu\nu} = \nabla_\mu X_\nu + \nabla_\nu X_\mu 
\end{equation*}
Using these we define the following two energy currents:
\begin{equation*}
\begin{split}
	\left( {^{(X)}J[\phi]} \right)_\mu &:= X^\nu Q_{\mu\nu}[\phi] \\
	^{(X)}K[\phi] &:= \frac{1}{2} {^{(X)}\pi^{\mu\nu}} Q_{\mu\nu}[\phi]
\end{split}
\end{equation*}

Now, for some time function $t$ we define the hypersurface
\begin{equation*}
	\Sigma_t := \left\{ x\in \mathcal{M} \big| t(x) = t \right\}
\end{equation*}
as well as the (open) spacetime region
\begin{equation*}
	\mathcal{M}_{t_1}^{t_2} := \left\{ x \in \mathcal{M} \big| t_1 < t(x) < t_2 \right\}
\end{equation*}
We choose the time function $t$ to agree with the coordinate function $t$ in the asymptotic region.

In general, given an initial surface $\Sigma_0$, we will choose $\Sigma_t$ to be the $T$-translate of $\Sigma_0$, where $t$ is a parameter such that $T(t) = 1$. Moreover, we choose the initial hypersurface $\Sigma_0$ to be a hypersurface which is fixed by the discrete isometry, as detailed in section \ref{section the discrete isometry}. 

Then, we have the energy identity: for $t_2 \geq t_1$, using $\imath$ to denote the interior product, then
\begin{equation}
\label{equation energy estimate}
	\int_{\Sigma_{t_2}} \imath_{ ({^{(X)}J[\phi]}) } \dVol_g = \int_{\Sigma_{t_1}} \imath_{ ({^{(X)}J[\phi]}) } \dVol_g + \int_{\mathcal{M}_{t_1}^{t_2}} {^{(X)}K[\phi]} \dVol_g
\end{equation}
In particular, if $X$ is a Killing vector field, then $^{(X)}K[\phi]$ vanishes identically, and the associated energy is conserved.

Now, suppose that $X$ is \emph{uniformly} timelike and future directed. Then the energy current ${^{(X)}J[\phi]}$ is ``non-degenerate'' in a sense we shall make precise below. As above, let $t$ be a time function for $\mathcal{M}$, so the hypersurfaces $\Sigma_t$ are uniformly timelike (i.e.\ the one-form $\upd t$ is uniformly timelike). The unit, future-directed normal to $\Sigma_t$ is defined by
\begin{equation*}
 n := -\frac{(\upd t)^\sharp}{\sqrt{ -g^{-1}(\upd t , \upd t) }}
\end{equation*}

Now, for any sufficiently small open set $\mathcal{U}\subset\mathcal{M}$, we can find an orthonormal basis for the tangent space $\mathcal{T}(\mathcal{U})$, consisting of the vector fields $\{ e_0, \ldots, e_d\}$, satisfying $g(e_a, e_b) = \eta_{ab}$, where $\eta_{00} = -1$, $\eta_{aa} = 1$ for $a \in \{1, 2, \ldots, d\}$ and $\eta_{ab} = 0$ for other values of $a$ and $b$. Furthermore, we can choose $e_0 = n$, and we can choose $e_1$ so that
\begin{equation*}
 X = X^0 n + X^1 e_1
\end{equation*}
for some smooth functions $X^0$ and $X^1$. Finally, we set
\begin{equation*}
 |\upd t| := \sqrt{-g^{-1}(\upd t, \upd t)}
\end{equation*}
Then, a short computation shows that in $\mathcal{U}$
\begin{equation}
\label{equation energy current}
 \begin{split}
  &\left( \imath_{\left(^{(X)}J[\phi]\right)}\dVol_g \right)\big|_{\Sigma_t} \\
  &= \frac{1}{2}\left( \frac{|\upd t|}{X(t)} \left(X(\phi)\right)^2 + \frac{ |\upd t|^2 g(X,X)}{\left( X(t) \right)^2} \left( e_1(\phi)\right)^2 + \sum_{a=2}^{d} \frac{\left(X(t)\right)}{|\upd t|} \left( e_a(\phi)\right)^2 \right) \imath_n \dVol_g \big|_{\Sigma_t}
 \end{split}
\end{equation}
The condition that $X$ is uniformly timelike and future directed ensures that 
\begin{equation}
\label{equation uniform transversal}
 c|\upd t| \leq X(t) \leq C|\upd t|
\end{equation}
and so we find that
\begin{equation*}
 \begin{split}
  &\frac{|\upd t|}{X(t)} \left(X(\phi)\right)^2 + \frac{ |\upd t|^2 g(X,X)}{\left( X(t) \right)^2} \left( e_1(\phi)\right)^2 + \sum_{a=2}^{d} \frac{\left(X(t)\right)}{|\upd t|} \left( e_a(\phi)\right)^2 \\
  &\sim \max\left\{c, \frac{1}{C}, \frac{c}{C^2} \right\} \left( (X(\phi))^2 + \sum_{a = 1}^{d} (e_a(\phi))^2 \right)
 \end{split}
\end{equation*}
so the energy associated with the vector field $X$ is equivalent to the $L^2$ norm of the derivatives, i.e.\ 
\begin{equation*}
 \int_{\Sigma_t}\left( \imath_{\left(^{(X)}J[\phi]\right)}\dVol_g \right) \sim || \partial \phi ||^2_{L^2(\Sigma_t)}
\end{equation*}
where we recall that $\partial \phi$ mean the collection of first derivatives of $\phi$, defined using an orthonormal frame.

Now, suppose that $X$ is future directed and ``uniformly transverse'' to the leaves of the foliation $\Sigma_t$, in such a way that the bounds in equation (\eqref{equation uniform transversal}) hold, but this time we do not require $X$ to be uniformly timelike. Then we can see from equation (\eqref{equation energy current}) that the energy current $^{(X)}J$ is \emph{not} necessarily comparible to the $L^2$ norm of the derivatives -- instead, the energy current will be ``missing'' a derivative wherever $X$ is null, and will not be positive definite where $X$ is spacelike.

\section{Energy currents in a spacetime with an evanescent ergosurface}

Let $\mathcal{M}$ be a stationary spacetime with an evanescent ergosurface, as defined in section \ref{section evanescent ergosurfaces}. In case the asymptotically flat case \ref{condition asymp flat} we define the \emph{conserved energy} at time $t$:
\begin{equation}
 \mathcal{E}^{(T)}[\phi](t) := \int_{\Sigma_t} \imath_{{^{(T)}J[\phi]}} \dVol_g
\end{equation}
while in the Kaluza-Klein case \ref{condition asymp KK} we define
\begin{equation}
 \mathcal{E}^{(V)}[\phi](t) := \int_{\Sigma_t} \imath_{{^{(V)}J[\phi]}} \dVol_g
\end{equation}
Note that these ``energies'' are conserved in the sense that, for any $t \in \mathbb{R}$,
\begin{equation*}
 \mathcal{E}^{(A)}[\phi](t) = \mathcal{E}^{(A)}[\phi](0) 
\end{equation*}
where $A = T$ or $V$, depending on whether we are in the asymptotically flat case or the asymptotically Kaluza-Klein case. This follows from the fact that both $T$ and $V$ are Killing vector fields for their respective geometries. In addition, since both $T$ and $V$ are globally causal, the energy currents $\mathcal{E}^{(A)}[\phi]$ are non-negative.

We also define the ``non-degenerate energy'', which is the energy associated with some \emph{uniformly timelike} vector field $N$, which is also $T$-invariant (or $V$-invariant, in the Kaluza-Klein case), that is, $[T, N] = 0$ (or $[V, N] = 0$).
\begin{equation*}
 \mathcal{E}^{(N)}[\phi](t) := \int_{\Sigma_t} \imath_{{^{(N)}J[\phi]}} \dVol_g
\end{equation*}
Note that, since $N$ is \emph{not} required to be a Killing vector field, this energy is not in general conserved. It represents the total energy measured by a family of observers moving along the timelike curves given by integral curves of $N$. Since $N$ is $T$-invariant, the \emph{family} of observers (though not each individual observer!) is stationary, that is, it is invariant under time translation.

Both the energy current $\mathcal{E}^{(T)}[\phi]$ and the current $\mathcal{E}^{(V)}[\phi]$ are ``degenerate'': the associated energies are \emph{not} equivalent to the $L^2$ norm of the derivatives of $\phi$. In case \ref{condition asymp flat}, the vector field $T$ becomes null on $\mathcal{S}$, so the energy current $\mathcal{E}^{(T)}[\phi]$ is ``missing'' a derivative on the evanescent ergosurface. On the other hand, in case \ref{condition asymp KK} the vector field $V$ is only \emph{globally causal}, and so the energy current $\mathcal{E}^{(V)}[\phi]$ might (in the case that $V$ is null everywhere) be ``missing'' a derivative everywhere.

However, we shall see that, when the function $\phi$ is invariant under the action of $\mathcal{G}$, then the energy current $\mathcal{E}^{(V)}[\phi]$ is, in fact, non-degenerate away from the submanifold $\mathcal{S}$. Indeed, this is the motivation for labelling the surface $\mathcal{S}$ an ``evanescent ergosurface'' in case \ref{condition asymp KK}.

\begin{proposition}
\label{proposition V energy current degenerate}
 Let $\phi$ be a function on $\mathcal{M}$ invariant under the action of the group $\mathcal{G}$, i.e.\ $L_A(\phi) = 0$ for all $L_A \in L_{\mathcal{G}}$. Suppose also that $t$ is invariant under the action of $\mathcal{G}$. Then, for any open set $\mathcal{U}$ such that $\mathcal{S}\subset\mathcal{U}$, there exists some $c = c(\mathcal{U}) > 0$ such that
\begin{equation}
 \int_{\Sigma_t \setminus \mathcal{U}} \imath_{{^{(V)}J[\phi]}}\dVol \geq c ||\partial \phi||^2_{L^2(\Sigma_t \setminus \mathcal{U})}
\end{equation}
\end{proposition}
\begin{proof}
 The calculations in section \ref{section energy momentum} show that the energy current $\mathcal{E}^{(V)}[\phi]$ is almost equivalent to the $L^2$ norm of the derivatives, except that it is ``missing'' a derivative in the direction of the orthogonal projection of $V$ onto the surfact $\Sigma_t$ (that is, in the $e_1$ direction). Note that $V^1 e_1 = V + g(V, n)n$. In order to prove the proposition, we need to show that, away from $\mathcal{S}$, we can express this derivative in terms of the vector fields $L_A$ as well as the other vector fields on $\Sigma_t$ that are orthogonal to the projection of $V$ onto $\Sigma_t$.

 At all points away from $\mathcal{S}$, there exists some $L_{\hat{A}} \in L_{\mathcal{G}}$ such that $g(L_{\hat{A}}, V) \neq 0$. We claim that
\begin{equation*}
 V + g(V,n)n = \frac{(g(V,n))^2}{g(L_{\hat{A}}, V)} \left( L_{\hat{A}} + g(L_{\hat{A}}, n)n - \sum_{A=2}^d g(L_{\hat{A}}, e_A)e_A \right)
\end{equation*}
 To prove this we simply take contractions with the orthonormal base $\{n, e_1, \ldots, e_d\}$. Additionally, using the conditions in the definition of the evanescent ergosurface, as well as the condition that $c|\upd t| \leq V(t) \leq C|\upd t|$, we also find that for any function $\phi$, in the set $\mathcal{M}\setminus\mathcal{U}$ there exists some $c > 0$ such that
\begin{equation*}
 \left| V(u) + g(V,n)n(\phi) \right| \leq c\left( \frac{|L_{\hat{A}}(u)|}{\sqrt{g(L_{\hat{A}},L_{\hat{A}})}} + |n(\phi)| + \sum_{A=2}^d |e_A \phi| \right)
\end{equation*}
In particular, if $\phi$ is invariant under the action of $\mathcal{G}$, then the ``missing'' derivative in the energy can be \emph{uniformly} bounded in terms of the derivatives which \emph{do} appear in the energy, proving the proposition.
\end{proof}

\section{Estimates in asymptotically Kaluza-Klein spaces}
\label{section estimates in KK spaces}

In the following two sections we will need to make use of some of the results proved by Moschidis in both \cite{Moschidis2016} and \cite{Moschidis2015}. In section \ref{section general case} we appeal to the results of \cite{Moschidis2016}, which apply to asymptotically flat manifolds with an ergoregion but with no horizon; we can therefore apply all estimates which do not rely on the existence of a nonempty ergoregion to our manifolds in the asymptotically flat case (in the sense of section \ref{section asymptotic structure}). However, for manifolds with Kaluza-Klein asymptotics of the appropriate form (also defined in section \ref{section asymptotic structure}), some modification is needed.

In section \ref{section additional symmetry} we will also make some use of the results of \cite{Moschidis2015}, which establishes logarithmic decay of the local energy on a very general class of geometries. These results are proved for asymptotically flat manifolds with a \emph{globally} timelike Killing vector $T$, or else on manifolds with an ergoregion which overlaps with a ``red shift region''. Neither of these conditions holds in the cases we are considering, and the additional symmetry assumed in section \ref{section additional symmetry} plays a crucial role in recovering some of the results. These issues concern the local geometry near the evanescent ergosurface, and will be adressed in section \ref{section additional symmetry}. In this section, we shall sketch an adaptation of some of the methods of \cite{Moschidis2016, Moschidis2015} to the asymptotically Kaluza-Klein case.

The key inequalities which we need to extend to the asymptotically Kaluza-Klein case are the ``Carleman inequalities'' (see section 6 of \cite{Moschidis2016}), as well as the ``$p$-weighted energy estimate'' (see section 5 of \cite{Moschidis2015}\footnote{Note that this powerful technique, first developed in \cite{Dafermos2010b}, was also extensively developed for use on the linear wave equation in \cite{Moschidis2015a} (among many others), and was also extended for use with nonlinear equations in \cite{Yang2013, Yang2014, Yang2013a, Keir2018}.}).  The Carleman inequality is used to establish decay in \cite{Moschidis2016}, while \cite{Moschidis2015} used the $p$-weighted estimates for this purpose, so there is actually some redundancy here (this redundancy is noted explicitly in \cite{Moschidis2016}), however, since we will quote the results of these two papers rather than re-derive them, it is important that both types of estimates can be applied in our setting, i.e.\ on asymptotically Kaluza-Klein spacetimes.

Some of the estimates in \cite{Moschidis2016} and \cite{Moschidis2015} will actually \emph{not} hold for general waves on asymptotically Kaluza-Klein spacetimes. However, they will still hold for $\mathcal{G}$-invariant waves on these spacetimes. Before discussing this further, we need the following standard result:

\begin{proposition}
 Let $\phi$ be a solution to the wave equation $\Box_g \phi = 0$ on an asymptotically Kaluza-Klein manifold, and suppose that the initial data for $\phi$, that is $\phi\big|_{\Sigma_0}$ and $n(\phi)\big|_{\Sigma_0}$ is invariant under the action of $\mathcal{G}$, i.e.\ $L_A \phi\big|_{\Sigma_0} = n L_A \phi\big|_{\Sigma_0} = 0$ for all $L_A \in L_{\mathcal{G}}$. Then $L_A \phi = 0$ everywhere in $\mathcal{M}$.
\end{proposition}
\begin{proof}
 Since the vector fields $L_A$ are Killing vector fields for the metric $g$, they commute with the geometric wave operator $\Box_g$, so $L_A \phi$ satisfies
\begin{equation*}
 \Box_g (L_A \phi) = 0
\end{equation*}
 However, the initial data for $L_A \phi$ vanishes identically, and so by standard uniqueness results $L_A \phi = 0$ everywhere.
\end{proof}

Now, suppose $\phi$ is a solution to the wave equation which is invariant under the action of $\mathcal{G}$. Then we claim that $\phi$ satisfies a wave equation (with a right hand side) on a related manifold in the ``asymptotic region''. In fact, away from the fixed points of the action of $\mathcal{G}$ we can use the action of $\mathcal{G}$ to define a smooth projection. We have
\begin{equation*}
 \pi_{\mathcal{G}} : \mathcal{U}^{\text{as}} \rightarrow \mathcal{U}^{\text{as}} / \mathcal{G}
\end{equation*}
In the asymptotic region we have a diffeomorphism $\varphi: \mathcal{U}^{\text{as}} \rightarrow \mathbb{R}\times(\mathbb{R}^d\setminus B)\times X$ and the group action of $\mathcal{G}$ descends to a transitive action on the space $X$, so the projection $\pi_\mathcal{G}(\mathcal{U}^{\text{as}})$ of the asymptotic region may be identified with $\mathbb{R}\times(\mathbb{R}^d\setminus B)$. That is we can define the projection
\begin{equation*}
 \begin{split}
  \bar{\pi}_\mathcal{G} : \mathbb{R}\times(\mathbb{R}^d\setminus B)\times X &\rightarrow \mathbb{R}\times(\mathbb{R}^d\setminus B) \\
  (t,x,gH) &\mapsto (t,x)
 \end{split}
\end{equation*}
This projection allows us to define \emph{vertical vector fields} in $T\mathcal{M}$. To be precise, we say that a vector field $X$ is \emph{vertical} if its pushforward by the projection map $\pi_\mathcal{G}$ (or, equivalently in the asymptotic region, its pushforward by the map $\varphi \circ \bar{\pi}_\mathcal{G}$) vanishes.

Similarly, \emph{in the asymptotic region} (although not necessarily elsewhere) we can make sense of the notion of \emph{horizontal} vector fields as follows: we can define the map
\begin{equation*}
 \begin{split}
  \bar{\pi}_{\text{hor}} : \mathbb{R}\times(\mathbb{R}^d\setminus B)\times X &\rightarrow X \\
  (t, x, gH) &\mapsto (gH)
 \end{split}
\end{equation*}
Then a vector field $X \in T\mathcal{U}^{\text{as}}$ is called \emph{horizontal} if its pushforward under the projection $\varphi_{\text{as}} \circ \bar{\pi}_{\text{hor}}$  vanishes. Using this, we can split the tangent space in the asymptotic region as
\begin{equation*}
 T_p(\mathcal{M}) = \left( T_p(\mathcal{M}) \right)^\text{ver} \oplus \left( T_p(\mathcal{M}) \right)^\text{hor}
\end{equation*}
for $p \in \mathcal{U}^{\text{as}}$. Note that the spaces $\left( T_p(\mathcal{M}) \right)^\text{ver}$ and $\left( T_p(\mathcal{M}) \right)^\text{hor}$ are then linear subspaces of the tangent space at the point $p$, but they are not, in general, orthogonal subspaces.

We now take a non-vanishing section of the frame bundle of some sufficiently small set $\mathcal{U} \subset \mathcal{U}^{\text{as}}$, such that the vectors in the frame each lie within either the vertical or horizontal subspace at each point. We shall label the frame vector fields
\begin{equation*}
 \begin{split}
  (e_a)\big|_p &\in \left( T_p(\mathcal{M}) \right)^\text{hor} \text{ } \forall p \in \mathcal{U}^{\text{as}} \text{, for } a \in \{ 0, \ldots, d\} \\
  (e_A)\big|_p &\in \left( T_p(\mathcal{M}) \right)^\text{ver} \text{ } \forall p \in \mathcal{U}^{\text{as}} \text{, for } A \in \{ d+1, \ldots, d+n\} \\
 \end{split}
\end{equation*}
We may moreover take the frame vectors $e_A$ to satisfy $e_A \in L_{\mathcal{G}}$, since every vector $L_g \in L_{\mathcal{G}}$ is evidently in the vertical subspace. Hence the vectors $e_A$ are Killing vectors for the metric $g$. Note, however, that we do not assume that the vector fields in the frame commute with one another, nor are they necessarily orthogonal. Note also that we will use lower case latin letters to label vectors in the horizontal subspace, and upper case latin letters to label vectors in the vertical subspace. When we wish to sum over the entire (co)tangent space we shall use (lower case) greek letters. Finally, note that, for vector fields $L_A$, $L_B \in L_{\mathcal{G}}$, a standard computation reveals that their commutator is given by
\begin{equation*}
 [L_A, L_B] = L_{[B,A]}
\end{equation*}
where on the right hand side we use the Lie bracket associated with the Lie algebra $\mathfrak{g}$ of $\mathcal{G}$. In particular, we have that $[L_A, L_B] \in L_{\mathcal{G}}$.

Given local coordinates $\vartheta^1 ,\ldots, \vartheta^{d-1}$ for a coordinate patch on the sphere $\mathbb{S}^{d-1}$, we may use the coordinate vector fields $\partial_t$, $\partial_r$, $\partial_{\vartheta^1} , \ldots, \partial_{\vartheta^{d-1}}$ to give a bases for the horizontal subspace. Since these vector fields are coordinate-induced, their commutators vanish. Additionally, the fact that the coordinates $t, r, \vartheta^1 \ldots$ are invariant under the action of $\mathcal{G}$ (e.g.\ we have $L_A (r) = 0$ for all $A \in \mathfrak{g}$) ensures that $[e_A, e_b] \in L_{\mathcal{G}}$.

Given a vector field $X$, we can write
\begin{equation*}
 X = X^a e_a + X^A e_A
\end{equation*}
where we use the usual summation convention to sum over both the indices labelling the horizontal subspace and those labelling the vertical subspace. Similarly, given a covector field $\omega$, we write
\begin{equation*}
 \begin{split}
  \omega_a &:= \omega(e_a) \\
  \omega_A &:= \omega(e_A)
 \end{split}
\end{equation*}

Now, let $\phi$ be a function which is invariant under the action of $\mathcal{G}$, so $e_A(\phi) = 0$. Then we can compute
\begin{equation}
 \label{equation boxg phi}
 \Box_g \phi = g^{ab} e_a \left(e_b(\phi)\right) - \left( g^{ab}\Gamma^c_{ab} - g^{AB}\Gamma^c_{AB} - 2g^{aB}\Gamma^c_{aB} \right)e_c(\phi)
\end{equation}
where
\begin{equation*}
 \Gamma^\alpha_{\beta\gamma} := \frac{1}{2}(g^{-1})^{\alpha\delta} \left( e_\beta g_{\delta \gamma} + e_\gamma g_{\delta \beta} - e_\delta g_{\gamma \beta} + g([e_\delta, e_\gamma], e_\beta) + g([e_\delta, e_\beta], e_\gamma) - g([e_\gamma, e_\beta], e_\delta) \right)
\end{equation*}

In the asymptotic region, we can use the map $\varphi_{\mathcal{G}}:\mathcal{M}\rightarrow \mathcal{M}/\mathcal{G}$ to define the pushforward of the inverse metric
\begin{equation*}
 (\varphi_{\mathcal{G}})_*(g^{-1}):= \bar{g}^{-1}
\end{equation*}
 Moreover, in the asymptotic region we can see from the form of the metric that, at least for sufficiently large $r$, $\bar{g}^{-1}$ is non-degenerate. Thus, we can define a symmetric, non-degenerate inverse tensor $\bar{g} \in \Gamma\left( T_*(\mathcal{M})\otimes T_*(\mathcal{M}) \right)$. We can equip the orbit space $\mathcal{M}/\mathcal{G}$ with the metric $\bar{g}$, and for any $\mathcal{G}$-invariant function $\phi$ on $\mathcal{M}$ we can identify $\phi$ with a function $\bar{\phi}$ on $\mathcal{M}/\mathcal{G}$. Then, for sufficiently large $R$, we find that if $\phi$ satisfies equation \eqref{equation boxg phi}, then $\bar{\phi}$ satisfies a wave equation on the Lorentzian manifold $\left(\mathcal{M}/\mathcal{G} \cap \{r > R\} \ , \ \bar{g}\right)$
\begin{equation*}
 \begin{split}
  \Box_{\bar{g}} \bar{\phi} + F^a (\varphi_{\mathcal{G}})^*e_a (\bar{\phi}) &= 0 \\
  F^a &:= \frac{1}{2}(\bar{g}^{-1})^{ab}(\bar{g}^{-1})^{cd}\left( e_c\left(\bar{g}_{be} g^{eA} g_{Ad} \right) + e_d\left(\bar{g}_{be} g^{eA} g_{Ac} \right) - e_b\left(\bar{g}_{ce} g^{eA} g_{Ad} \right) \right) \\
  &\phantom{:=} - (g^{-1})^{bc}(g^{-1})^{aD} e_b(g_{cD}) -\frac{1}{2}(g^{-1})^{AB}(g^{-1})^{a\mu} e_\mu (g_{AB}) \\
  &\phantom{:=} - (g^{-1})^{Ab}(g^{-1})^{ac}e_b (g_{cA}) + (g^{-1})^{Ab}(g^{-1})^{ac}e_c (g_{Ab}) \\
  &\phantom{:=} - (g^{-1})^{Ab}(g^{-1})^{aD}\left( e_b(g_{AD}) + g([e_A, e_D], e_b) \right)
 \end{split}
\end{equation*}
Note that the terms defining the functions $F^a$, $a\in\{0,\ldots d\}$ are $\mathcal{G}$-invariant, which follows from the form of the metric in the asymptotic region, and so their pushforward to the orbit space is well-defined. Furthermore, note that the derivatives of the metric in the directions $e_a$ are $\mathcal{O}(r^{-2-\alpha})$, whereas the derivatives in the directions $e_A$ generally have worse decay. For example, since the components of $[e_A, e_B]$ will generally be $\mathcal{O}(1)$, we have $(g^{-1})^{a\mu} e_\mu (g_{AB}) = \mathcal{O}(r^{-1-\alpha})$.

Importantly, we find that, overall,
\begin{equation*}
 |F^a| = \mathcal{O}(r^{-1-\alpha})
\end{equation*}
And so, although $\phi$ does not satisfy a wave equation on $\mathcal{M}/\mathcal{G}$ (even in the asymptotic region), we can consider the terms arising from $F$ as error terms, and moreover, error terms of this kind were dealt with previously in \cite{Yang2013a, Moschidis2015a, Keir2018}. Note, for example, that \cite{Yang2013a} dealt\footnote{In fact, using the results of \cite{Keir2018}, error terms of the form $F\cdot \partial \phi$ with $F \sim r^{-1}$ can be handled, but this requires some subtle modification of the estimates.} with nonlinear equations using a similar method to that used in \cite{Dafermos2010b, Moschidis2015a}, and in the process established estimates for equations of the form $\Box_g u = F$ in the asymptotic region, for $F = \mathcal{O}(r^{-2-\alpha})$. Note also that we can view the derivatives in the compact directions $e_A$ as ``bad'' derivatives: taking a derivative of the metric in these directions does not in general improve the decay rate of the metric component, in contrast to the $e_a$ derivatives. Hence we can expect $(e_A \phi) \sim r^{-1}$.

As mentioned above, we intend to make direct use of the results of \cite{Moschidis2016} and \cite{Moschidis2015}. Note, however, that a function $\phi$ satisfying $\Box_{\bar{g}}\phi = F(\partial \phi)$, for $F = \mathcal{O}(r^{-1-\alpha})$ does not quite satisfy all of the estimates in \cite{Moschidis2015}, which are concerned with homogeneous equations of the form $\Box_{\bar{g}} \phi = 0$. One can check that an inhomogeneous term of this sort does not pose any problems for the Carleman estimates of \cite{Moschidis2016}. On the other hand, certain estimates in \cite{Moschidis2015} need a little bit of additional consideration.

To be precise, the $p$-weighted energy estimates of \cite{Dafermos2010b} cannot be established for the entire range $0 < p < 2$ as in the homogeneous case $F = 0$ studied in \cite{Moschidis2015}. To see why, we can sketch the case where $g = m$, the Minkowski metric on $\mathbb{R}^{d+1}$, and $F^a e_a(\phi) = F^{\Lbar}\Lbar(\phi)$, where $\Lbar = \partial_t - \partial_r$ in the standard Minkowski coordinates, and where we assume $F^{\Lbar} = \mathcal{O}(r^{-1-\alpha})$.

First, set $t^*$ to be the function
\begin{equation*}
	t^* :=
	\begin{cases}
		t &\text{ if } r \leq R \\
		t - r + R &\text{ if } r \geq R
	\end{cases}
\end{equation*}
and we set 
\begin{equation*}
	\Sigma_{t^*} := \{x\in \mathcal{M} \big| t^*(x) = t^* \} 
\end{equation*}
For the purposes of this example, we will assume that an \emph{integrated local energy decay} statement of the following type holds:
\begin{equation*}
 \int_{\tau = \tau_1}^{\tau_2} \left( \int_{\Sigma_\tau} \alpha(1+r)^{-1-\frac{\alpha}{2}} \left( \sum_{a = 0}^{d} |e_a \phi|^2 \right) r^{d-1} \upd r   \dVol_{\mathbb{S}^{d-1}} \right) \upd \tau \lesssim \mathcal{E}^{(N)}[\phi](\tau_1)
\end{equation*}
where $\dVol_{\mathbb{S}^{d-1}}$ is the standard volume form on the \emph{unit} $(d-1)$-sphere and $\alpha>0$ is some positive constant. Note that such an estimate can indeed be proven using the multiplier $\left(1 - (1+r)^{-\frac{\alpha}{2}} \right) \partial_r$ and an appropriately modified energy current in the asymptotic region. However, our purpose here is to sketch the argument so we will skip this computation.

Setting $L = \partial_t + \partial_r$, and defining $\psi := r^{\frac{d - 1}{2}}\phi$, we find that $\psi$ satisfies
\begin{equation*}
 -L\Lbar\psi + \slashed{\Delta}\psi + r^{\frac{d-1}{2}}F^{\Lbar} (\Lbar \phi) - \frac{(d-1)(d-3)}{4} r^{-2} \psi = 0
\end{equation*}
Multiplying by $-r^p L\psi$ and integrating by parts in the region $r\geq R$, $\tau_0 \leq \tau \leq \tau_1$ we find that
\begin{equation}
\label{equation p weighted v1}
\begin{split}
	&\int_{\Sigma_{t^*_1} \cap \{r \geq R\}} r^p (L\psi)^2 \upd r   \dVol_{\mathbb{S}^2}
	\\
	&
	+ \frac{1}{2} \int_{t^*_0}^{t^*_1} \bigg(
		\int_{\Sigma_\tau\cap\{r\geq R\}} \bigg(
			pr^{p-1}(L\psi)^2 
			+ (2-p)r^{p-1}|\slashed{\nabla}\psi|^2
			\\
			&\phantom{ + \frac{1}{2} \int_{t^*_0}^{t^*_1} \bigg( \int_{\Sigma_\tau\cap\{r\geq R\}} \bigg( }
			+ (2-p) \frac{(d-1)(d-3)}{4}r^{p-3}|\psi|^2
		\bigg) \upd r \dVol_{\mathbb{S}^{d-1}} \bigg) \upd \tau
		\\
	&\leq \int_{\Sigma_{\tau_0}\cap\{ r \geq R\}} r^p (L\psi)^2 \upd r   \dVol_{\mathbb{S}^{d-1}}
	+ \int_{t = R}^{R + \tau} \left( \int_{\mathbb{S}^{d-1}} r^p\left( (L\psi)^2 + |\slashed{\nabla}\psi|^2 \right) \dVol_{\mathbb{S}^{d-1}} \right)\upd t \\
	&\phantom{\leq} + \int_{\tau_0}^{\tau_1}\left( \int_{\Sigma_\tau}  |F^{\Lbar}|r^{\frac{d-1}{2} + p}|\Lbar \phi||L\psi| \upd r   \dVol{\mathbb{S}^{d-1}} \right) \upd \tau
 \end{split}
\end{equation}
The terms on the right hand side which are evaluated at $\tau = \tau_0$ are to be considered part of the initial data, while the term evaluated on the surface $r = R$ can be dealt with using the integrated local energy decay estimate in a standard way. The new term which must be estimated arises from the presence of the inhomogeneity $F^{\Lbar}$.

Since $F^{\Lbar}$ satisfies $|F^{\Lbar}| = \mathcal{O}(r^{-1-\alpha})$, we can estimate this term as follows: for $r \geq R$, and for any $\delta > 0$ we have
\begin{equation*}
 |F|r^{\frac{d-1}{2} + p}|\Lbar \phi||L\psi| \lesssim \delta r^{p-1}|L\psi|^2 + \frac{1}{\delta}C(R) (1+r)^{-1 + p - 2\alpha}|\Lbar \phi|^2 r^{d-1}
\end{equation*}
If we now take $p = \frac{3}{2}\alpha$ and substitute back into equation (\eqref{equation p weighted v1}), then we can estimate
\begin{equation}
\label{equation p weighted v2}
\begin{split}
	& \int_{\Sigma_{\tau_1} \cap \{r \geq R\}} r^{\frac{3}{2}\alpha} (L\psi)^2 \upd r   \dVol_{\mathbb{S}^2} \\
	& + \int_{\tau_0}^{\tau_1} \bigg(
		\int_{\Sigma_\tau\cap\{r\geq R\}} \bigg(
			\alpha r^{\frac{3}{2}\alpha-1}(L\psi)^2
			+ \left(4-3\alpha\right)r^{\frac{3}{2}\alpha-1}|\slashed{\nabla}\psi|^2 
			\\
			&\phantom{ + \int_{\tau_0}^{\tau_1} \bigg( \int_{\Sigma_\tau\cap\{r\geq R\}} \bigg( }
			+ (d-1)(d-3)(4-3\alpha) r^{\frac{3}{2}\alpha-3}|\psi|^2
		\bigg) \upd r \dVol_{\mathbb{S}^{d-1}}\bigg)\upd \tau \\
	& \lesssim \mathcal{E}^{(N)}[\phi](\tau_0) + \int_{\Sigma_{\tau_0}\cap\{ r \geq R\}} r^{\frac{3}{2}\alpha} (L\psi)^2 \upd r   \dVol_{\mathbb{S}^{d-1}} \\
	& \phantom{\lesssim} + \int_{\tau_0}^{\tau_1}\left( \int_{\Sigma_\tau} \left( \delta r^{\frac{3}{2}\alpha-1}|L\psi|^2 + C(R)\frac{1}{\delta}(1+r)^{-1 - \frac{1}{2}\alpha}|\Lbar \phi|^2 r^{d-1} \right)  \upd r   \dVol_{\mathbb{S}^{d-1}} \right) \upd \tau
 \end{split}
\end{equation}
If we fix $\delta$ as some sufficiently small constant, then the first term in the second integral on the right hand side can be absorbed by the corresponding term on the left hand side. Moreover, the second term in the second integral on the right hand side can be estimated in terms of the initial energy by using the integrated local energy decay estimate, resulting in the bound
\begin{equation}
\label{equation p weighted v3}
\begin{split}
	&\int_{\Sigma_{\tau_1} \cap \{r \geq R\}} r^{\frac{3}{2}\alpha} (L\psi)^2 \upd r   \dVol_{\mathbb{S}^2}
	+ \int_{\tau_0}^{\tau_1} \bigg( 
		\int_{\Sigma_\tau\cap\{r\geq R\}} \bigg(
			\alpha r^{\frac{3}{2}\alpha-1}(L\psi)^2
			+ \left(4-3\alpha\right)r^{\frac{3}{2}\alpha-1}|\slashed{\nabla}\psi|^2
			\\
			&\phantom{\int_{\Sigma_{\tau_1} \cap \{r \geq R\}} r^{\frac{3}{2}\alpha} (L\psi)^2 \upd r   \dVol_{\mathbb{S}^2} + \int_{\tau_0}^{\tau_1} \bigg( }
			+ (d-1)(d-3)(4-3\alpha) r^{\frac{3}{2}\alpha-3}|\psi|^2
		\bigg) \upd r \dVol_{\mathbb{S}^{d-1}} \bigg)\upd \tau \\
	&\lesssim \left(1+\alpha^{-1}C(R)\right)\mathcal{E}^{(N)}[\phi](\tau_0) + \int_{\Sigma_{\tau_0}\cap\{ r \geq R\}} r^{\frac{3}{2}\alpha} (L\psi)^2 \upd r   \dVol_{\mathbb{S}^{d-1}}
 \end{split}
\end{equation}

Next, we return to equation \eqref{equation p weighted v1}, and this time we estimate the term involving the inhomogeneity $F^{\Lbar}$ as follows:
\begin{equation*}
 |F|r^{\frac{d-1}{2} + p}|\Lbar \phi||L\psi| \lesssim (1+\tau)^{-1-\beta} r^p|L\psi|^2 + C(R) (1+\tau)^{1+\beta}(1+r)^{-2 + p - 2\alpha}|\Lbar \phi|^2 r^{d-1}
\end{equation*}
where $\beta > 0$ is some constant. Substituting this into equation \eqref{equation p weighted v1} and setting $p = 1 + \frac{3}{2}\alpha$ we can obtain the bound
\begin{equation*}
\begin{split}
	&\int_{\Sigma_{\tau_1} \cap \{r \geq R\}} r^{1 + \frac{3}{2}\alpha} (L\psi)^2 \upd r \dVol_{\mathbb{S}^2}
	+ \int_{\tau_0}^{\tau_1} \bigg(
		\int_{\Sigma_\tau\cap\{r\geq R\}} \bigg(
			(1+\alpha) r^{\frac{3}{2}\alpha}(L\psi)^2
			+ \left(2-3\alpha\right)r^{\frac{3}{2}\alpha}|\slashed{\nabla}\psi|^2
			\\
			&\phantom{\int_{\Sigma_{\tau_1} \cap \{r \geq R\}} r^{1 + \frac{3}{2}\alpha} (L\psi)^2 \upd r \upd\mu_{\mathbb{S}^2} + \int_{\tau_0}^{\tau_1} \bigg( }
			+ (d-1)(d-3)(2-3\alpha) r^{\frac{3}{2}\alpha-2}|\psi|^2
		\bigg) \upd r \dVol_{\mathbb{S}^{d-1}} \bigg)\upd \tau \\
  &\lesssim \mathcal{E}^{(N)}[\phi](\tau_0) + \int_{\Sigma_{\tau_0}\cap\{ r \geq R\}} r^{1+\frac{3}{2}\alpha} (L\psi)^2 \upd r   \dVol_{\mathbb{S}^{d-1}} \\
  &\phantom{} + \int_{\tau_0}^{\tau_1}\left( \int_{\Sigma_\tau} \left( (1+\tau)^{-1-\beta}r^{1+\frac{3}{2}\alpha}|L\psi|^2 + C(R)(1+\tau)^{1+\beta}(1+r)^{-1 - \frac{1}{2}\alpha}|\Lbar \phi|^2 r^{d-1} \right)  \upd r   \dVol_{\mathbb{S}^{d-1}} \right) \upd \tau
 \end{split}
\end{equation*}

The third term on the right hand side can be estimated by making use of the integrated local energy decay statement, while the second can be controlled using Gronwall's inequality. Note, however, that this third term grows as $(1+\tau_1)^{1+\beta}$. We are led to the inequality
\begin{equation}
\label{equation p weighted v4}
\begin{split}
	&\int_{\Sigma_{\tau_1} \cap \{r \geq R\}} r^{1 + \frac{3}{2}\alpha} (L\psi)^2 \upd r   \dVol_{\mathbb{S}^2}
	+ \int_{\tau_0}^{\tau_1} \bigg(
		\int_{\Sigma_\tau\cap\{r\geq R\}} \bigg(
			(1+\alpha) r^{\frac{3}{2}\alpha}(L\psi)^2
			+ \left(2-3\alpha\right)r^{\frac{3}{2}\alpha}|\slashed{\nabla}\psi|^2
			\\
			&\phantom{ \int_{\Sigma_{\tau_1} \cap \{r \geq R\}} r^{1 + \frac{3}{2}\alpha} (L\psi)^2 \upd r   \dVol_{\mathbb{S}^2} + \int_{\tau_0}^{\tau_1} \bigg( }
			+ (d-1)(d-3)(2-3\alpha) r^{\frac{3}{2}\alpha-2}|\psi|^2
		\bigg) \upd r \upd \mu_{\mathbb{S}^{d-1}} \bigg)\upd \tau \\
	&\lesssim (1+\alpha^{-1}C(R))(1+\tau_1)^{1+\beta}\mathcal{E}^{(N)}(\tau_0) + \int_{\Sigma_{\tau_0}\cap\{ r \geq R\}} r^{1+\frac{3}{2}\alpha} (L\psi)^2 \upd r   \dVol_{\mathbb{S}^{d-1}} \\
 \end{split}
\end{equation}

Interpolating between equations \eqref{equation p weighted v4} and \eqref{equation p weighted v3}, and making use of H\"older's inequality we can show
\begin{equation}
\label{equation p weighted v5}
\begin{split}
	&\int_{\Sigma_{\tau_1} \cap \{r \geq R\}} r (L\psi)^2 \upd r   \dVol_{\mathbb{S}^2}
	\\
	&
	+ \int_{\tau_0}^{\tau_1} \bigg(
		\int_{\Sigma_\tau\cap\{r\geq R\}} \bigg(
			(L\psi)^2
			+ |\slashed{\nabla}\psi|^2
			+ (d-1)(d-3) r^{-2}|\psi|^2
		\bigg) \upd r \dVol_{\mathbb{S}^{d-1}} \bigg) \upd \tau \\
	& \lesssim (1+\alpha^{-1}C(R))(1+\tau_1)^{(1+\beta)\left(1-\frac{3}{2}\alpha\right)}\tilde{\mathcal{E}}^{(N)}_{(1+\frac{3}{2}\alpha)}[\phi](\tau_0) \\
 \end{split}
\end{equation}
where we have defined the weighted energy
\begin{equation*}
 \tilde{\mathcal{E}}^{(N)}_{(p)}[\phi](\tau) := \tilde{\mathcal{E}}^{(N)}[\phi](\tau) + \int_{\Sigma_{\tau} \cap \{r \geq R\}} r^{p} (L\psi)^2 \upd r   \dVol_{\mathbb{S}^2} 
\end{equation*}
where, as usual, $\psi = r^{\frac{d-1}{2}}\phi$.

Now, if $\tilde{\mathcal{E}}^{(N)}_{(1+\frac{3}{2}\alpha)}[\phi](\tau_0) < \infty$ and if we define
\begin{equation*}
 \sigma := 1 - (1+\beta)\left(1-\frac{3}{2}\alpha\right)
\end{equation*}
then, if $\sigma > 0$ (e.g.\ if $\beta > \frac{3}{2}\alpha$) then we can use equation (\eqref{equation p weighted v5}) to find a diadic sequence of times $\tau_i$ such that the energy satisfies
\begin{equation*}
 \tilde{\mathcal{E}}^{(N)}[\phi](\tau_i) \lesssim (1+\tau_i)^{-\sigma}\tilde{\mathcal{E}}^{(N)}_{(1+\frac{3}{2}\alpha)}[\phi](\tau_0)
\end{equation*}
The energy boundedness assumption \ref{assumption boundedness} then allows us to drop the restriction to the diadic sequence, and we find that for all times $\tau \geq \tau_0$ we have polynomial decay of the energy:
\begin{equation*}
 \mathcal{E}^{(N)}(\tau)[\phi] \lesssim (1+\tau)^{-\sigma}\mathcal{E}^{(N)}_{(1+\frac{3}{2}\alpha)}[\phi](\tau_0)
\end{equation*}

Note that we also have the important corollary: for $\tau \geq \tau_0$, $\tilde{\mathcal{E}}^{(N)}_{(1+\frac{3}{2}\alpha)}(\tau) \lesssim \tilde{\mathcal{E}}^{(N)}_{(1+\frac{3}{2}\alpha)}[\phi](\tau_0)$, i.e.\ the weighted energy at future times is bounded by the weighted energy initially.

In the above sketch, we have shown that the additional, inhomogeneous term arising in the wave equation satisfied by $\bar{\phi}$ on the orbit space $\mathcal{M}/\mathcal{G}$ does not prevent polynomial decay, assuming that an appropriately weighted initial energy is bounded, and also that an integrated local energy decay estimate holds. In the cases we will consider in this paper (as in the general cases studied in \cite{Moschidis2015}) this integrated local energy decay estimate will only hold for \emph{bounded} frequencies; nevertheless, by combining this with an energy boundedness statement such as assumption \ref{assumption boundedness}, this is sufficient to conclude decay at a logarithmic rate. In other words, for this argument it is only important to establish polynomial decay for bounded frequency waves: the precise exponent (which depends on the maximum value of $p$) is not important.

\section{Adapted coordinates near an evanescent ergosurface}
\label{section adapted coordinates}

In this section we will construct coordinates near a point on the evanescent ergosurface. These local coordinates will be used later in order to construct initial data for the wave. We will first do this in the case of an asymptotically flat manifold admitting an evanescent ergosurface of the first kind \ref{condition asymp flat}, and then show how to adapt this construction to an asymptotically Kaluza-Klein manifold and to an ergosurface of the second kind \ref{condition asymp KK}. The coordinates we will construct are very similar to ``null Fermi coordinates'' \cite{Manasse1963, Fermi1922}.

As above, let $\Sigma_t$ be a spacelike hypersurface and let $p \in \mathcal{S} \cap \Sigma_t$ be a point on the intersection of the evanescent ergosurface with the hypersurface $\Sigma_t$. Let $Y$ be the orthogonal projection of $T$ onto the hypersurface $\Sigma_t$. That is, if $n$ is the unit, future directed normal to $\Sigma_t$ then we define
\begin{equation*}
 Y := T + g(T, n) n
\end{equation*}
Note that, at $p$, $Y$ is non-vanishing, since $T$ is null here and $\Sigma_t$ is spacelike. We can define a normalised version of $Y$:
\begin{equation*}
 \hat{Y} := \frac{1}{\sqrt{g(Y,Y)}} Y
\end{equation*}

We complete $\{T, \hat{Y}\}$ to a form a basis for the tangent space of $\mathcal{M}$ at the point $p$. We can choose these other vectors to be mutually orthogonal, normalised, and also orthogonal to both $Y$ and $T$. In other words, we have some set of vectors $e_a \in \mathcal{T}_p(\mathcal{M})$ such that
\begin{equation*}
 \begin{split}
  g(e_a, e_b) &= \delta_{ab} \\
  g(e_a, \hat{Y}) &= 0 \\
  g(e_a, T) &= 0
 \end{split}
\end{equation*}
Note that these vectors are orthogonal to both $T$ and $Y$, and the normal to the hypersurface $\Sigma_t$ at the point $p$ is proportional to $(T - Y)$. Hence these vectors are tangent to the hypsersurface $\Sigma_t$ at the point $p$.

Now, we consider an affinely parameterised geodesic $\beta_{(y, x^a)}(s)$ originating at the point $p$ and with initial tangent vector $y\hat{Y} + x^a e_a$. That is, the geodesic $\beta_{(y, x^a)}(s)$ satisfies
\begin{equation*}
 \begin{split}
  \beta_{(y, x^a)}(0) &= p \\
  \left.\frac{\partial}{\partial s} \beta_{(y, x^a)}(s) \right|_{s = 0} &= y\hat{Y} + x^a e_a
 \end{split}
\end{equation*}

\begin{figure}[htbp]
	\centering
	\includegraphics[width = 0.95\linewidth, keepaspectratio]{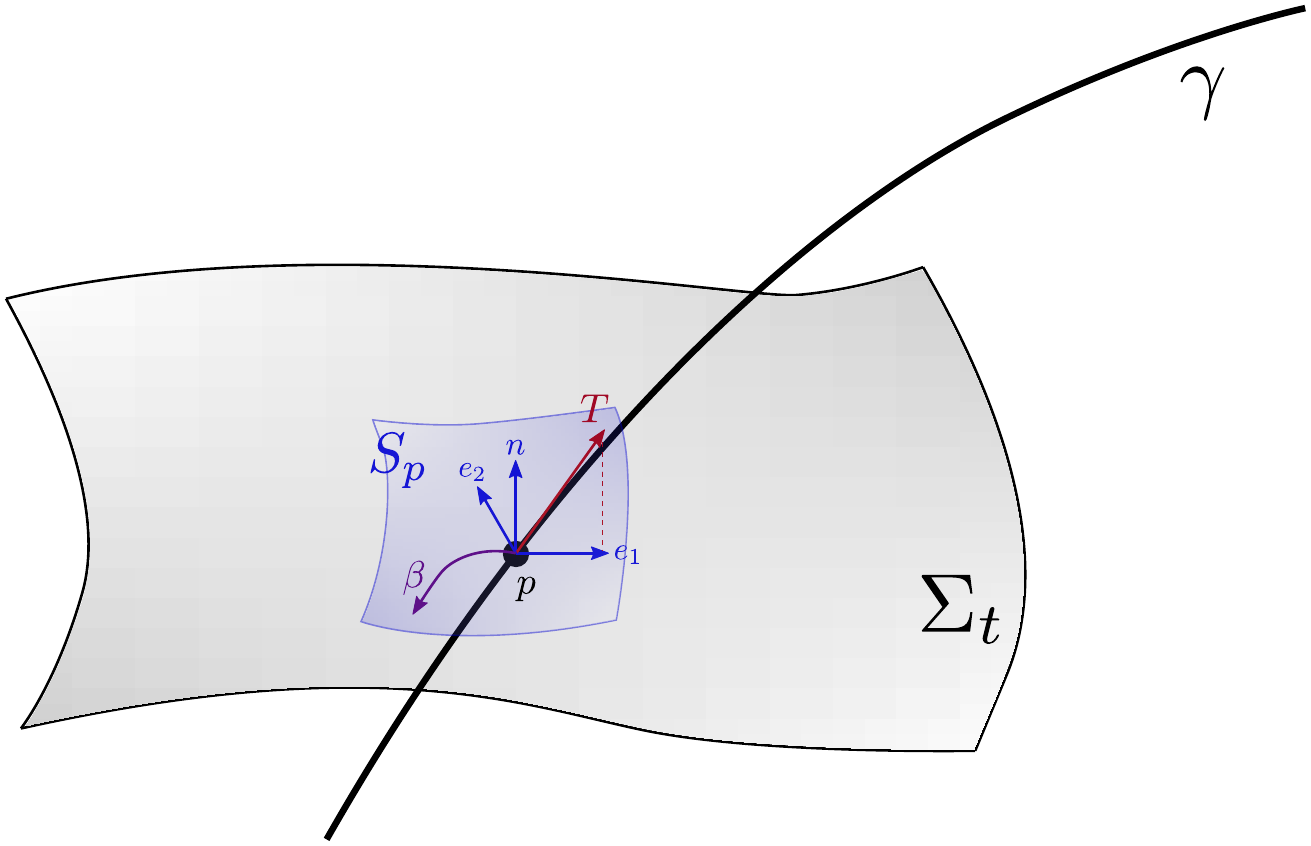}
	\caption{
		A sketch showing our construction of coordinates adapted to the ergosurface. Here, $p$ is a point on the ergosurface, and on the hypersurface $\Sigma_t$. The geodesic $\gamma$ is an integral curve of the Killing vector field $T$ (which should be replaced by $V$ in the asymptotically Kaluza-Klein case). $n$ is the future directed normal to $\Sigma_t$, and $e_1$ is a unit spacelike vector at $p$ in the direction of the orthogonal projection of $T$ onto $\Sigma_t$. $e_2$ is another unit vector field at $p$, tangent to $\Sigma_t$ and orthogonal to $e_1$. The purple curve $\beta$ is a geodesic, with tangent at $p$ which is in the span of $\{e_1, e_2\}$. The blue surface $S_p$ is the surface consisting of the union of all such curves in a neighbourhood of $p$. Note that the two hypersurfaces $S_p$ and $\Sigma_t$ are tangent at $p$ but do not, in general, agree away from $p$.
	}
	\label{figure coordinates}
\end{figure}
We can use the coordinates $(sy, sx^a)$ to label the point reached along this geodesic $\beta_{(y, x^a)}(s)$ after an affine distance $s$. Note that, since $\{Y, e_a\}$ do \emph{not} span the tangent space of $\mathcal{M}$ at the point $p$, we cannot yet use these coordinates in a neighbourhood of $p$. However, we can define the set
\begin{equation*}
 S_p := \left\{ q \in \mathcal{M} \big| \exists \ s, y, x^a \in \mathbb{R} \text{ s.t.\ } q = \beta_{(y, x^a)}(s) \right\}
\end{equation*}
This set is locally a smooth hypersurface near the point $p$. In other words, for any sufficiently small neighbourhood of $p$, the restriction of $S_p$ to this neighbourhood defines a smooth hypersurface in the neighbourhood. Moreover, if this neighbourhood is sufficiently small, then the vector field $T$ is transverse to $S_p$.

We recall that, since $p \in \mathcal{S}$, there is an affinely parameterised geodesic $\gamma$ through $p$ with tangent $T$. We will now define coordinates in a neighbourhood of $\gamma$. First, we define a function $\tilde{t}$ by the condition
\begin{equation*}
\begin{split}
	\tilde{t}\big|_{S_p} &= 0 \\
	T(\tilde{t}) &= 1
\end{split}
\end{equation*}
Next, we extend the coordinates $(y, x^a)$ off the hypsersurface $S_p$ by imposing the condition
\begin{equation*}
 T(y) = T(x^a) = 0
\end{equation*}
Note that the isometry generated by $T$ preserves distances, so by this process we are able to obtain coordinates for a local neighbourhood of the entire geodesic $\gamma$. See figure \ref{figure coordinates} for a sketch of this construction.

We note that, in these coordinates, we have
\begin{equation*}
 T = \partial_{\tilde{t}}
\end{equation*}
In addition, \emph{at the point} $p$ we have
\begin{equation*}
 \begin{split}
  \partial_y &= \hat{Y} \\
  \partial_a &= e_a 
 \end{split}
\end{equation*}

Hence, if we define
\begin{equation*}
 A := g(T, \hat{Y})
\end{equation*}
then \emph{at the point} $p$ the metric and its inverse are given by
\begin{equation}
\label{equation metric and inverse local coords p}
 \begin{split}
  g\big|_p &= 2A \upd \tilde{t}\upd y + \upd y^2 + \delta_{ab} \upd x^a \upd x^b \\
  g^{-1}\big|_p &= -A^{-2}\partial^2_{\tilde{t}} + 2A^{-1} \partial_{\tilde{t}}\partial_y + \delta^{ab}\partial_a \partial_b
 \end{split}
\end{equation}

Now, by construction, the curve with coordinates $(\tilde{t}, y, x^a) = (0, sy_0, sx^a_0)$ for \emph{arbitrary constants} $y_0$ and $x_0^a$ is a geodesic with affine parameter $s$. Since $T$ is a Killing vector field, the curve with coordinates $(t_0, sy_0, sx^a_0)$ for constants $t_0, sy_0, sx^a_0$ is also a geodesic. Hence, from the geodesic equation we obtain
\begin{equation*}
 \begin{split}
  \Gamma^{\tilde{t}}_{yy}\big|_\gamma &= \Gamma^{\tilde{t}}_{ya}\big|_\gamma = \Gamma^{\tilde{t}}_{ab}\big|_\gamma = 0 \\
  \Gamma^{y}_{yy}\big|_\gamma &= \Gamma^{y}_{ya}\big|_\gamma = \Gamma^{y}_{ab}\big|_\gamma = 0 \\
  \Gamma^{a}_{yy}\big|_\gamma &= \Gamma^{a}_{yb}\big|_\gamma = \Gamma^{a}_{bc}\big|_\gamma = 0 \\
 \end{split}
\end{equation*}

Similarly, since the curve with coordinates $(s, 0, 0)$ is an affinely parameterised geodesic we find
\begin{equation*}
 \Gamma^{\tilde{t}}_{\tilde{t}\tilde{t}}\big|_\gamma = \Gamma^{y}_{\tilde{t}\tilde{t}}\big|_\gamma = \Gamma^{a}_{\tilde{t}\tilde{t}}\big|_\gamma = 0
\end{equation*}

If it were the case that all the Christoffel symbols of $g$ vanish along $\gamma$, then we would be working in ``null Fermi coordinates'' adapted to the geodesic $\gamma$ \cite{Manasse1963}. However, this is not the case for the coordinates defined above: specifically, we cannot guarantee that all Christoffel symbols with mixed spatial and time indices vanish along $\gamma$. However note that at the point $p$ we have
\begin{equation*}
 \nabla_y (\partial_{\tilde{t}})\big|_p = \nabla_y T \big|_p = \Gamma^{\tilde{t}}_{y\tilde{t}}\big|_p T + \Gamma^{y}_{y\tilde{t}}\big|_p Y + \Gamma^{a}_{y\tilde{t}}\big|_p e_a
\end{equation*}
where we note that, although the vectors $Y$ and $e_a$ are only defined at the point $p$, $T$ is defined globally.

Taking an inner product with $T$ at the point $p$ and using the expression for the metric \eqref{equation metric and inverse local coords p} we find
\begin{equation*}
 g(\nabla_y T, T)\big|_p = A\Gamma^y_{y\tilde{t}}\big|_p
\end{equation*}
However, the left hand side is given by
\begin{equation*}
 g(\nabla_y T, T) = \frac{1}{2} \partial_y \left( g(T,T) \right)
\end{equation*}
and $g(T,T)$ vanishes at least quadratically on the evanescent ergosurface (and so, in particular, at $p$). We conclude that
\begin{equation*}
 \Gamma^y_{y\tilde{t}}\big|_p = 0
\end{equation*}

We can also compute
\begin{equation*}
 g(\nabla_T \partial_y, \partial_y)\big|_p = A \Gamma^{\tilde{t}}_{y\tilde{t}}\big|_p + \Gamma^{y}_{y\tilde{t}}\big|_p
\end{equation*}
The second term on the right hand side vanishes, as we have already seen. On the other hand, the left hand side is given by
\begin{equation*}
 \begin{split}
  g(\nabla_T \partial_y, \partial_y) &= \frac{1}{2}Tg(\partial_y, \partial_y) \\
  &= g([T, \partial_y], \partial_y) = 0
 \end{split}
\end{equation*}
where we have used the fact that $\mathcal{L}_T g = 0$ since $T$ is a Killing vector field, as well as the fact that $[T, \partial_y] = [\partial_{\tilde{t}}, \partial_y] = 0$. Hence we find that
\begin{equation*}
 \Gamma^{\tilde{t}}_{yt}\big|_p = 0
\end{equation*}

Similarly, by using the vector fields $\partial_a$ in place of $\partial_y$ we can also find that
\begin{equation*}
 \begin{split}
  \Gamma^{y}_{a\tilde{t}}\big|_p &= 0 \\
  \Gamma^{a}_{a \tilde{t}}\big|_p &= 0
 \end{split}
\end{equation*}
where in the second line there is \emph{no summation} over the index $a$.

To summarize the above calculations, we find that, with our choice of coordinates, the only Christoffel symbols which can be nonzero at $p$ (and hence along $\gamma$) are 
\begin{equation*}
 \Gamma^{\tilde{t}}_{\tilde{t}a}\ ,\quad \Gamma^a_{\tilde{t}y}\ ,\quad \Gamma^a_{by}
\end{equation*}
where $a \neq b$.

We now define
\begin{equation*}
 |\tilde{x}| := \sqrt{ y^2 + \sum_a (x^a)^2}
\end{equation*}
and we also give labels to certain metric components: we define
\begin{equation*}
 \begin{split}
  g(T, T) &= g_{\tilde{t} \tilde{t}} := a \\
  g(T, \partial_a) &= g_{\tilde{t}a} := b_a
 \end{split}
\end{equation*}
and we note that
\begin{equation*}
 \begin{split}
  |a| &= \mathcal{O}(|\tilde{x}|^2) \\
  |b| &= \mathcal{O}(|\tilde{x}|)
 \end{split}
\end{equation*}
where we have defined
\begin{equation*}
 |b| := \sqrt{ \sum_a (b_a)^2 }
\end{equation*}
Finally, we note that, since $T$ is timelike away from the evanescent ergosurface, we have
\begin{equation*}
 a \leq 0
\end{equation*}

Putting together the calculations above, we conclude that the metric near the point $p$ is given by
\begin{equation*}
 \begin{split}
  g &= 2A \upd \tilde{t} \upd y + \upd y^2 + \delta_{ab}\upd x^a \upd x^b + a \upd \tilde{t}^2 + b_a \upd \tilde{t} \upd x^a 
    + \mathcal{O}(|\tilde{x}|)\upd \tilde{t} \upd y
    + \mathcal{O}(|\tilde{x}|^2)\upd y^2 \\
    &\phantom{=} + \mathcal{O}(|\tilde{x}|^2)\upd y \upd x^a + \mathcal{O}(|\tilde{x}|^2)\upd x^a \upd x^b
 \end{split}
\end{equation*}
Consequently, the inverse metric is given by
\begin{equation*}
 \begin{split}
  g^{-1} &= A^{-2} \partial_{\tilde{t}}^2 + 2A^{-1}\partial_{\tilde{t}}\partial_y + \delta^{ab}\partial_a \partial_b  + \mathcal{O}(|\tilde{x}|) \partial_{\tilde{t}}^2 + \mathcal{O}(|\tilde{x}|)\partial_{\tilde{t}}\partial_y + \mathcal{O}(|\tilde{x}|)\partial_{\tilde{t}}\partial_a + \mathcal{O}(|\tilde{x}|)\partial_y \partial_a \\
  &\phantom{=} + \left( A^{-2}|b|^2 - A^{-2}a + \mathcal{O}(|\tilde{x}|^3) \right)\partial_y^2 + \mathcal{O}(|\tilde{x}|^2)\partial_a \partial_b
 \end{split}
\end{equation*}

Finally, again making use of the expressions for the Christoffel symbols at the point $p$, we find that the wave operator can be expressed as
\begin{equation}
\label{equation wave local coords AF}
 \begin{split}
  \Box_g u &= -A^{-2}\partial_{\tilde{t}}^2 u + 2A^{-1} \partial_{\tilde{t}}\partial_y u + \delta^{ab}\partial_a \partial_b u + A^{-2}\left( |b|^2 - a \right)\partial^2_y u + \mathcal{O}(|\tilde{x}| |\partial u|) + \mathcal{O}(|\tilde{x}||\partial T u|) \\
  &\phantom{=} + \mathcal{O}(|\tilde{x}|^2 |\bar{\partial}\partial u|) + \mathcal{O}(|\tilde{x}|^3 |\partial^2 u|)
 \end{split}
\end{equation}
where we have defined
\begin{equation*}
 \begin{split}
  |\partial u| &:= \sqrt{ |\partial_{\tilde{t}}u|^2 + |\partial_y u|^2 + \sum_a |\partial_a u|^2 } \\
  |\bar{\partial} u| &:= \sqrt{ |\partial_{\tilde{t}}u|^2 + \sum_a |\partial_a u|^2 } \\
 \end{split}
\end{equation*}

We now explain how to adapt this construction to the asymptotically Kaluza-Klein case, with an evanescent ergosurface of the second kind \ref{condition asymp KK}. In this case, we once again take a point $p \in \mathcal{S}\cap\Sigma_t$. We now define a linear subspace of the tangent space of $p$
\begin{equation*}
 G_p \subset \mathcal{T}_p(\mathcal{M}) := \left\{ X \in \mathcal{T}_p(\mathcal{M}) \big| X = L_A \text{\quad for some } A \in \mathfrak{g} \right\}
\end{equation*}
Note that, since the group $\mathcal{G}$ does not necessarily act freely, this subspace is not necessarily isomorphic to $\mathfrak{g}$. Indeed, if $p$ is a fixed point of $\mathcal{G}$, then $G_p$ is trivial.

Note that $V$ is orthogonal to $G_p$, in the sense that
\begin{equation*}
 g(V, X) = 0 \text{ \quad for all } X \in G_p
\end{equation*}
which follows from the definition of an evanescent ergosurface of the second kind. Consequently, the vectors $X$ must either be spacelike or proportional to $V$. However, if a vector $X \in G_p$ were proportional to $V$ then, since $V$ generates an action of $\mathbb{R}$ on $\mathcal{M}$ by isometries, $\mathcal{G}$ would have a subgroup isomorphic to $\mathbb{R}$. But this is impossible, since $\mathcal{G}$ is a compact Lie group. Hence the vectors $X$ must all be spacelike.

Analogously to the previous case, we now define
\begin{equation*}
 Y := V + g(n, V) n
\end{equation*}
as the orthogonal projection of $V$ onto the hypersurface $\Sigma_t$ at the point $p$. Note that $Y$ is nonvanishing, since $V$ is null at $p$ and $\Sigma_t$ is spacelike. Again, we can define the normalised version of $Y$ as
\begin{equation*}
 \hat{Y} := \frac{1}{\sqrt{g(Y,Y)}} Y
\end{equation*}
Note that $Y \notin G_p$ (and hence $\hat{Y} \notin G_p$) because $Y$ is \emph{not} orthogonal to $V$. Again, this follows from the fact that $V$ is null at $p$ and $\Sigma_t$ is spacelike.

We now take an orthonormal basis for $G_p$, which we shall label as $e_A$. Note that this is possible since $G_p$ is spacelike. We now complete the set $\{V, \hat{Y}, e_A\}$ to form a basis for $\mathcal{T}_p(\mathcal{M})$ by adding some vectors $e_a$, which are chosen to satisfy
\begin{equation*}
 \begin{split}
  g(e_a, e_b) &= \delta_{ab} \\
  g(e_a, V) &= 0 \\
  g(e_a, \hat{Y}) &= 0 \\
  g(e_a, e_A) &= 0
 \end{split}
\end{equation*}
We note here that we do not necessarily have $g(e_A, \hat{Y})\big|_p = 0$.

Now we repeat the previous construction to define coordinates. We first consider an affinely parameterised geodesic $\beta_{(y, x^a, z^A)}(s)$ originating at the point $p$ and with initial tangent vector $y\hat{Y} + x^a e_a + z^A e_A$. So, the geodesic $\beta_{(y, x^a, z^A)}(s)$ satisfies
\begin{equation*}
 \begin{split}
  \beta_{(y, x^a, z^A)}(0) &= p \\
  \left.\frac{\partial}{\partial s} \beta_{(y, x^a, z^A)}(s) \right|_{s = 0} &= y\hat{Y} + x^a e_a + z^A e_A
 \end{split}
\end{equation*}

We can use the coordinates $(sy, sx^a, sz^A)$ to label the point reached along this geodesic $\beta_{(y, x^a, z^A)}(s)$ after an affine distance $s$. We define the set
\begin{equation*}
 S_p := \left\{ q \in \mathcal{M} \big| \exists \ s, y, x^a, z^A \in \mathbb{R} \text{ s.t.\ } q = \beta_{(y, x^a, z^A)}(s) \right\}
\end{equation*}
As before, this set is locally a smooth hypersurface near the point $p$. We now extend these coordinates off this hypersurface by defining
\begin{equation*}
 \begin{split}
  \tilde{t}\big|_{S_p} &= 0 \\
  V(\tilde{t}) &= 1 \\
  V(y) = V(x^a) = V(z^A) &= 0
 \end{split}
\end{equation*}
This gives us local coordinates near the geodesic $\gamma$ through $p$ with tangent $V$.

Note that, although the vector fields $\partial_A$ are tangent to the generators of the group action $\mathcal{G}$ \emph{at p}, they do not necessarily remain so. Note also that, in these coordinates, we have
\begin{equation*}
 V = \partial_{\tilde{t}}
\end{equation*}

In analogy to the previous case we define
\begin{equation*}
 A := g(V, \hat{Y})\big|_p
\end{equation*}
We also define
\begin{equation*}
 B_A := g(\hat{Y}, e_A)\big|_p
\end{equation*}
Now, the metric at the point $p$ can be expressed as
\begin{equation*}
 g\big|_p = 2A \upd \tilde{t}\upd y + 2B_A \upd \tilde{t} \upd z^A + \upd y^2 + \delta_{ab}\upd x^a \upd x^b + \delta_{AB}\upd z^A \upd z^B
\end{equation*}
and the inverse metric is
\begin{equation*}
 g^{-1}\big|_p = -A^{-2}( 1 -  |B|^2 )\partial_{\tilde{t}}^2 + 2A^{-1}\partial_{\tilde{t}}\partial_y  - 2A^{-1}B^A \partial_{\tilde{t}}\partial_{A} + \delta^{AB}\partial_A \partial_B + \delta^{ab}\partial_a \partial_b
\end{equation*}
Where we lower and raise the indices $A, B, \ldots$ using the Euclidean metric $\delta_{AB}$ and its inverse, and we have defined $|B|^2 = B^A B_A$. Note that, since $\hat{Y}$ and $e_A$ are spacelike, unit vectors, we have $|B|^2 \leq 1$. Moreover, we cannot have $|B|^2 = 1$ because this would imply that $Y \in G_p$, which, as we have seen above, is impossible. Hence $|B|^2 < 1$, and so the coefficient of $\partial_{\tilde{t}}^2$ in the inverse metric is strictly negative.

As before, using the fact that the curve with coordinates $(\tilde{t}, y, x^a, z^A) = (\tilde{t}_0, sy_0, sx_0^a, sz_0^A)$ is a geodesic with affine parameter $s$, that many of the Christoffel symbols vanish. Specifically,
\begin{equation*}
 \begin{split}
  \Gamma^{\tilde{t}}_{yy}\big|_\gamma &= \Gamma^{\tilde{t}}_{ya}\big|_\gamma = \Gamma^{\tilde{t}}_{yA}\big|_\gamma = \Gamma^{\tilde{t}}_{ab}\big|_\gamma = \Gamma^{\tilde{t}}_{aA}\big|_\gamma = \Gamma^{\tilde{t}}_{AB}\big|_\gamma = 0 \\
  \Gamma^{y}_{yy}\big|_\gamma &= \Gamma^{y}_{ya}\big|_\gamma = \Gamma^{y}_{yA}\big|_\gamma = \Gamma^{y}_{ab}\big|_\gamma = \Gamma^{y}_{aA}\big|_\gamma = \Gamma^{y}_{AB}\big|_\gamma = 0 \\
  \Gamma^{a}_{yy}\big|_\gamma &= \Gamma^{a}_{yb}\big|_\gamma = \Gamma^{a}_{yA}\big|_\gamma = \Gamma^{a}_{bc}\big|_\gamma = \Gamma^{a}_{bA}\big|_\gamma = \Gamma^{a}_{AB}\big|_\gamma = 0 \\
  \Gamma^{A}_{yy}\big|_\gamma &= \Gamma^{A}_{ya}\big|_\gamma = \Gamma^{A}_{yB}\big|_\gamma = \Gamma^{A}_{ab}\big|_\gamma = \Gamma^{A}_{aB}\big|_\gamma = \Gamma^{A}_{BC}\big|_\gamma = 0 \\
 \end{split}
\end{equation*}

In addition, the curve with coordinates $(s, 0, 0, 0)$ is an affinely parameterized geodesic with affine parameter $s$. Hence
\begin{equation*}
 \Gamma^{\tilde{t}}_{\tilde{t}\tilde{t}} = \Gamma^{y}_{\tilde{t}\tilde{t}} = \Gamma^{a}_{\tilde{t}\tilde{t}} = \Gamma^{A}_{\tilde{t}\tilde{t}} = 0
\end{equation*}

Finally, we note that we have the expression
\begin{equation*}
 \nabla_{\partial_y}\partial_{\tilde{t}} = \Gamma^{\tilde{t}}_{y\tilde{t}} \partial_{\tilde{t}} + \Gamma^y_{y\tilde{t}}\partial_y + \Gamma^a_{y\tilde{t}}\partial_a + \Gamma^A_{y\tilde{t}}\partial_A
\end{equation*}
Noting that, at the point $p$ we have $\partial_t = V$, $\partial_y = \hat{Y}$, $\partial_a = e_a$ and $\partial_A = e_A$, we can evaluate the expression above at the point $p$ and then take the inner product with $V$. We find
\begin{equation*}
 g(\nabla_{\partial_y}V, V) = \frac{1}{2}\partial_y g(V,V) = A\Gamma^y_{y\tilde{t}}
\end{equation*}
Now, since $A \neq 0$ and $g(V,V)$ vanishes to (at least) second order on the evanescent ergosurface, we also conclude that
\begin{equation*}
 \Gamma^y_{y\tilde{t}}\big|_p = 0
\end{equation*}

Similarly, by considering $g(\nabla_a V, V)$ and $g(\nabla_A V, V)$ we conclude that
\begin{equation*}
 \Gamma^y_{a\tilde{t}}\big|_p = \Gamma^y_{A\tilde{t}}\big|_p = 0
\end{equation*}

Defining now
\begin{equation*}
 \begin{split}
  |\tilde{x}| &= \sqrt{y^2 + \sum_a (x^a)^2 + \sum_A (z^A)^2} \\
  a &:= g(V,V) = g_{\tilde{t}\tilde{t}} = \mathcal{O}(|\tilde{x}|^2) \\
  b_a &:= g(V, \partial_a) = g_{\tilde{t}a} = \mathcal{O}(|\tilde{x}|) \\
  c_A &:= g(V, \partial_A) = g_{\tilde{t}A} = \mathcal{O}(|\tilde{x}|)
 \end{split}
\end{equation*}
 where, since $V$ is globally causal, we have $a \leq 0$.
 
 We find that the metric near the point $p$ can be expressed as
\begin{equation*}
 \begin{split}
  g &= 2A \upd \tilde{t}\upd y + 2B_A \upd \tilde{t} \upd z^A + \upd y^2 + \delta_{ab}\upd x^a \upd x^b + \delta_{AB}\upd z^A \upd z^B + a \upd\tilde{t}^2 + b_a \upd \tilde{t} \upd x^a + c_A \upd \tilde{t} \upd z^A \\
  &\phantom{=} + \mathcal{O}(|\tilde{x}|)\upd\tilde{t} \upd y + \mathcal{O}(|\tilde{x}|^2)\upd y^2 + \mathcal{O}(|\tilde{x}|^2)\upd y \upd x^a + \mathcal{O}(|\tilde{x}|^2) \upd y \upd z^A + \mathcal{O}(|\tilde{x}|^2) \upd x^a \upd x^b \\
  &\phantom{=} + \mathcal{O}(|\tilde{x}|^2) \upd x^a \upd z^A + \mathcal{O}(|\tilde{x}|^2) \upd z^A \upd z^B
  \end{split}
\end{equation*}
and the inverse metric can be expressed as
\begin{equation}
\label{equation expression for inverse metric}
\begin{split}
	g^{-1}
	&= 
	- A^{-2}(1 - |B|^2)\partial_{\tilde{t}}^2 + 2A^{-1}\partial_{\tilde{t}}\partial_y
	- 2A^{-1}B^A \partial_{\tilde{t}}\partial_{A}
	+ \delta^{AB}\partial_A \partial_B
	+ \delta^{ab}\partial_a \partial_b
	\\
	&\phantom{=}
	+ A^{-2}(|b|^2 + |c|^2 - a)\partial^2_y 
	+ \mathcal{O}(|\tilde{x}|)\partial^2_{\tilde{t}}
	+ \mathcal{O}(|\tilde{x}|)\partial_{\tilde{t}}\partial_y
	+ \mathcal{O}(|\tilde{x}|)\partial_{\tilde{t}}\partial_a
	+ \mathcal{O}(|\tilde{x}|)\partial_{\tilde{t}}\partial_A
	+ \mathcal{O}(|\tilde{x}|^3)\partial^2_y
	\\
	&\phantom{=}
	+ \mathcal{O}(|\tilde{x}|)\partial_{y}\partial_a
	+ \mathcal{O}(|\tilde{x}|)\partial_{y}\partial_A
	+ \mathcal{O}(|\tilde{x}|^2)\partial_a \partial_b
	+ \mathcal{O}(|\tilde{x}|)\partial_{a}\partial_A
	+ \mathcal{O}(|\tilde{x}|)\partial_{A}\partial_B
 \end{split}
\end{equation}

Finally, we note that the wave equation can be expressed as
\begin{equation}
\label{equation wave local coords KK}
 \begin{split}
   \Box_g u &= -A^{-2}(1 - |B|^2)\partial_{\tilde{t}}^2 u + 2A^{-1}\partial_{\tilde{t}}\partial_y u  - 2A^{-1}B^A \partial_{\tilde{t}}\partial_{A} u + \delta^{AB}\partial_A \partial_B u + \delta^{ab}\partial_a \partial_b u \\
   &\phantom{=} + A^{-2}(|b|^2 + |c|^2 - a)\partial^2_y u + \mathcal{O}(|\tilde{x}|\partial V u|) + \mathcal{O}(|\tilde{x}^2| |\bar{\partial}\partial u|) + \mathcal{O}(|\tilde{x}| |\partial u|) + \mathcal{O}(|\tilde{x}|^3 |\partial^2 u|) 
  \end{split}
\end{equation}

\section{Instability in the general case, without additional symmetry assumptions}
\label{section general case}

The purpose of this section is to establish the existence of some kind of linear ``instability'' in the general case of a spacetime with an evanescent ergosurface. As above, we shall only consider spacetimes which are asymptotically flat or asymptotically Kaluza-Klein, and which do not have event horizons.

As explained in the introduction, we allow for two different types of instability, which we refer to as case \ref{case (A)} and case \ref{case (B)}. So far, we have only given a ``rough'' version of the statement referring to these instabilities. Here, we state the primary theorem of this paper, and in doing so we make precise the two kinds of instabilities which may be present.

\begin{theorem}[Evanescent ergosurface instability, general case]
\label{theorem main theorem general case}
	Let $(\mathcal{M}, g)$ be a smooth, Lorentzian manifold which is stationary and either asymptotically flat in the sense of subsection \ref{subsection asymptotically flat}, or asymptotically Kaluza-Klein in the sense of subsection \ref{subsection asymptotically KK}. Suppose that the manifold possesses an evanescent ergosurface in the sense of \ref{condition asymp flat} (in the asymptotically flat case) or \ref{condition asymp KK} (in the asymptotically Kaluza-Klein case). Finally, suppose that the manifold possesses a discrete symmetry $\mathcal{I}$ as in section \ref{section the discrete isometry}.
	
	Then at least one of the following applies:
	\begin{enumerate}[label=(\Alph*)]
		\item For any $C > 0$, and any open set $\mathcal{U}_0 \subset \Sigma_0$ such that $(\mathcal{S} \cap \Sigma_0) \subset \mathcal{U}_0$, there exists a solution $\phi$ to the wave equation $\Box_g \phi = 0$ arising from smooth, compactly supported (and $\mathcal{G}$-invariant, in the Kaluza-Klein case) data, and a time $\tau$ such that
			\begin{equation*}
				\frac{\mathcal{E}^{(N)}_{\mathcal{U}}[\phi_C](\tau)}{\mathcal{E}^{(N)}[\phi_C](0)} \geq C
			\end{equation*}
		where we recall that $\mathcal{E}^{(N)}_{\mathcal{U}}[\phi_C](\tau)$ measures the non-degenerate energy of the wave $\phi_C$ in the set $\mathcal{U}_{\tau}$, which is the time translate of the set $\mathcal{U}_0$ onto the surface $\Sigma_\tau$.
		\item For any open set $\mathcal{U}_0 \subset \Sigma_0$ with $(\mathcal{S} \cap \Sigma_0) \subset \mathcal{U}_0$, there exists some constant $\mathring{C}>0$ and a solution $\phi$ to the wave equation $\Box_g \phi = 0$, arising from smooth, compactly supported data, such that for \emph{all} times $\tau > 0$, we have 
			\begin{equation*}
				\frac{\mathcal{E}^{(N)}_{\mathcal{U}}[\phi_C](\tau)}{\mathcal{E}^{(N)}[\phi_C](0)} \geq \mathring{C}
			\end{equation*}
		Furthermore we have the following pointwise blowup behaviour: there exists a constant $c > 0$ such that, given any set $\mathcal{U}_\epsilon \subset \Sigma_0$ with $\mathcal{S} \cap \mathcal{U}_\epsilon \neq \emptyset$ and $\text{Volume}(\mathcal{U}_\epsilon) = \epsilon$, there exists a solution $\phi$ to the wave equation $\Box_g \phi = 0$ and a time $\tau$ such that
			\begin{equation*}
				|| Y \phi ||_{L^\infty [\mathcal{U}_\epsilon](\tau)} \geq \frac{c}{\epsilon} 
			\end{equation*}			
	\end{enumerate}
\end{theorem}

In this general setting we cannot establish many details of the instability. Note that the two possible behaviours are not mutually exclusive. Note also that, since we will argue by contradiction, we will not obtain any details of the initial data which gives rise to these instabilities, other than smoothness and compact support: in particular, our proof is not constructive. Finally, we  remark that although we have several examples of manifolds giving rise to the behaviour of case \ref{case (A)} (for example, supersymmetric microstate geometries), we do not have an example of a manifold giving rise to the behaviour of case \ref{case (B)}.

This is in marked contrast to the situation in which an additional symmetry is present, discussed in section \ref{section additional symmetry} below. In that case, we can rule out case \ref{case (B)}, and explicitly construct data giving rise to the behaviour in case \ref{case (A)}. Moreover, we can also establish some bounds on the time at which the local energy becomes large, and the required support of the initial data. Finally, when this additional symmetry is present we can also construct (possibly non-smooth, and non-compactly supported) initial data such that the local energy of the resulting solution is actually \emph{unbounded} in time.

The structure of the proof of theorem \ref{theorem main theorem general case} is a little convoluted, so for clarity we outline it below.

\begin{enumerate}
 \item \label{step 1} We begin by assuming that case \ref{case (A)} does \emph{not} hold, that is, we assume that the local energy is bounded by some multiple of its initial value. This will be referred to as a \emph{nondegenerate energy boundedness} statement.
 \item There are now two possibilities: either the local energy of all suitable waves \emph{decays} over time, or it does not.
  \begin{enumerate}
  \item If the local energy does not decay, then we are led to an Aretakis-type instability and case \ref{case (B)}.
  \item On the other hand, if the local energy \emph{does} decay, then we can construct initial data for a wave whose local energy is amplified by an arbitrarily large factor, i.e.\ case \ref{case (A)}. This contradicts the boundedness assumption made in step \ref{step 1} above. See figure \ref{figure instability} for an overview of the construction of the instability in this case.
  \end{enumerate}
 \item From the above argument we see that, if \ref{case (A)} does not hold, then we must have \ref{case (B)}.
 \item Hence we must have at least one of \ref{case (A)} or \ref{case (B)}.
\end{enumerate}

\subsection{The nondegenerate energy boundedness assumption}

In order to make progress we will need to assume a suitable \emph{nondegenerate energy boundedness} inequality holds, which is essentially\footnote{Technically, the negation of the statement of case \ref{case (A)} only entails that there exists some open set $\mathcal{U}_0$, which includes the ergosurface, such that a nondegenerate energy boundedness statement holds in that region. However, away from the ergoregion, the conservation of the $T$ energy already gives the required bound.} the negation of the statement of case \ref{case (A)}. Note, however, that this is \emph{not} an assumption which limits the scope of the theorem: if this assumption does not hold, then we can show that we have an instability of type \ref{case (A)}. On the other hand, after making this assumption, we will be able to show that a consequence of this assumption is an Aretakis-type instability in the general case. Note that, in section \ref{section additional symmetry} we will show that the nondegenerate energy boundedness inequality assumption leads directly to a contradiction in the case where an additional symmetry is present.

To make the statement precise, we make the following assumption:

\begin{assumption}
\label{assumption boundedness}
 There exists some constant $C^{(N)} > 0$ such that, for all ($\mathcal{G}$-invariant) solutions $\phi$ to the linear wave equation $\Box_g \phi = 0$ and all $t \in \mathbb{R}$, 
\begin{equation*}
\label{equation boundedness}
 \mathcal{E}^{(N)}[\phi](t) \leq C^{(N)} \mathcal{E}^{(N)}[\phi](0)
\end{equation*}
\end{assumption}

Note that, if $N$ were a Killing field, then it would be easy to verify \ref{assumption boundedness} using the energy estimate \eqref{equation energy estimate} associated with $N$. Likewise, if there were to exist a uniformly timelike Killing vector field, then it would be easy to verify assumption \ref{assumption boundedness}, even if $N$ is not chosen to be this Killing vector field\footnote{In this case, the energy associated with $N$ will not generally be conserved, however, it will still remain bounded: the energy associated with $N$ and the energy associated with the timelike Killing vector field provide equivalent norms.}. However, the geometries we are studying only possess a globally causal (and not globally timelike!) Killing field. Thus we cannot straightforwardly verify assumption \ref{assumption boundedness}, and indeed, in some cases it can lead to a contradiction. For now, we shall proceed, making the assumption \ref{assumption boundedness}.

\subsection{Local energy decay away from the evanescent ergosurface}
\label{subsection local decay away from ergosurface}

The results of \cite{Moschidis2016} imply the decay of the local energy of waves away from the evanescent ergosurface, assuming a boundedness statement of the form \ref{assumption boundedness} holds. To be precise, the following proposition is a very slight adaptation of proposition ($4.1$) of \cite{Moschidis2016}, making use of the comments above regarding asymptotically Kaluza-Klein manifolds.

We first need to define (in analogy to the ``extended ergoregions'' of \cite{Moschidis2016}) the \emph{extended ergosurface}:
\begin{definition}[The extended ergosurface]
	Suppose that $(\mathcal{M}\setminus\mathcal{S})$ consists of a number of connected components, which we can seperate into two types: those that include an asymptotic region, and those that do not. For simplicity, suppose that there is only one component which includes an asymptotic region. We label this region $\mathcal{M}_{(\text{outer})}$. Then we define the extended ergosurface:
	\begin{equation*}
		\mathcal{S}_{(\text{ext})} := \mathcal{M} \setminus \mathcal{M}_{(\text{outer})}
	\end{equation*}
\end{definition}

\begin{proposition}
\label{proposition decay away from extended ergosurface}
 Let $\mathcal{M}$ be an asymptotically flat or asymptotically Kaluza-Klein manifold with an evanescent ergosurface in the sense of \ref{condition asymp flat}, or an asymptotically Kaluza-Klein manifold with an evanescent ergosurface in the sense of \ref{condition asymp KK}. Let $\phi$ be a smooth ($\mathcal{G}$-invariant) solution to $\Box_g \phi = 0$ arising from compactly supported initial data. Suppose in addition that the initial energy of $\phi$ and its first three $T$ derivatives is finite, that is,
 \begin{equation*}
  \sum_{j = 0}^{3} \mathcal{E}^{(N)}[T^j \phi](0) < \infty
 \end{equation*}
 Finally, suppose that the boundedness statement \ref{assumption boundedness} holds.

 Then, for any $\delta > 0$ let $\mathcal{U}_0 \subset \Sigma_0$ be any compact set such that the distance\footnote{The distance can be measured using the induced Riemannian metric on $\Sigma_0$.} from $\mathcal{U}_0$ to $\mathcal{S}_{(\text{ext})} \cap \Sigma_0$ is at least $\delta$. Then, for any $\epsilon > 0$ there is a time $\tau > 0$ such that
\begin{equation}
 \mathcal{E}^{(N)}_{\mathcal{U}}[T\phi](\tau) + \mathcal{E}^{(N)}_{\mathcal{U}}[T^2 \phi](\tau) < \epsilon
\end{equation}

\end{proposition}

Note that this proposition actually follows from an application of the mean value theorem to the proposition given in \cite{Moschidis2016}, which establishes a very similar inequality for an \emph{integrated} energy quantity. However, we have chosen to present the proposition in the form which will be most useful for our purposes.

Note also that the distance to the extended evanescent ergosurface, $\delta$, can be chosen to be as small as we like, although the time $\tau$ taken to decay will depend on $\delta$. Hence, this proposition establishes decay of the local energy everywhere \emph{away from the ergosurface}. This will play an important role in our argument for instability. Indeed, if this decay can be extended to cover the ergoregion as well, then we will find a contradiction with the boundedness assumption \ref{assumption boundedness}. On the other hand, if this decay cannot be extended to the ergosurface, then we are faced with a situation in which the energy decays everywhere except for on the ergosurface. In this case, an instability of a very similar kind to that encountered in extremal black holes is present.

\subsection{Local energy decay ``inside'' the evanescent ergosurface}

In situations where the evanescent ergosurface divides the manifold $\mathcal{M}$ into an ``inside'' and an ``outside'', we also need to establish energy decay ``inside'' the ergosurface. In other words, we need to establish decay in the set $\mathcal{S}_{(\text{ext})} \setminus \mathcal{S}$, if this is nonempty.

Here, we can use lemma $4.2$ of \cite{Moschidis2016}, which we quote\footnote{with a very slight modification to account for the lack of a horizon.} here for convenience:

\begin{lemma}[Lemma $4.2$ of \cite{Moschidis2016}]

	Let $\phi$ be a smooth solution to $\Box_g \phi = 0$ arising from compactly supported initial data, such that
	\begin{equation*}
		\sum_{j = 0}^3 \mathcal{E}^{(N)}[T^j \phi](0) < \infty
	\end{equation*}
	Define the function $\psi_\tau : \mathcal{M} \rightarrow \mathbb{R}$ as
	\begin{equation*}
		\psi_\tau (t, x) := \begin{cases}
			T\phi(t+\tau, x) \quad &, \quad t \geq -\tau,
			\\
			0 \quad &, \quad t < \tau
			\end{cases}
	\end{equation*}
	Then, there exists an increasing sequence $\{\tau_n\}_{n \in \mathbb{N}}$ of non-negative integers and a function $\tilde{\psi}$, with $\tilde{\psi}$, $T\tilde{\psi} \in H^1_{\text{loc}}(\mathcal{M})$ such that $\Box_g \tilde{\psi} = 0$ and
	\begin{equation}
		\int_{-\tau_*}^{\tau_*} \left( \mathcal{E}^{(N)}[\tilde{\psi}](\tau) + \mathcal{E}^{(N)}[T\tilde{\psi}] \right) \upd \tau < \infty \quad \text{ for any } \tau_* > 0
	\end{equation}
	and also
	\begin{equation}
		\tilde{\psi} \equiv 0 \quad \text{ on } \mathcal{M}\setminus \mathcal{S}_{(\text{ext})}
	\end{equation}
	
	Moreover, $(T\phi_{\tau_n}, T^2 \phi_{\tau_n}) \rightarrow (\tilde{\psi}, T\tilde{\psi})$ weakly in $H^1_{(\text{loc})}(\mathcal{M}) \times H^1_{(\text{loc})}(\mathcal{M})$, strongly in $H^1_{(\text{loc})}(\mathcal{M}\setminus \mathcal{S}_{\text{ext}}) \times H^1_{(\text{loc})}(\mathcal{M}\setminus \mathcal{S}_{\text{ext}})$ and strongly in $L^2_{(\text{loc})}(\mathcal{M}) \times L^2_{(\text{loc})}(\mathcal{M})$ in the following sense:
	\begin{itemize}
		\item For any compactly supported test functions $\{\zeta_j \}_{j=0,1} \in L^2(\mathcal{M})$ and compactly supported vector fields $\{X_j \}_{j=0,1}$ on $\mathcal{M}$ such that $|X_j|_{g_{(\text{ref})}} \in L^2(\mathcal{M})$:
		\begin{equation*}
			\lim_{n\rightarrow \infty} \sum_{j=0}^1 \int_{\mathcal{M}} \bigg(
				g_{(\text{ref})} \left( \nabla (T^j \psi_{\tau_n} - T^j \tilde{\psi} \ , \ X_j \right)
				+ (T^j \psi_{\tau_n} - T^j \tilde{\psi} ) \zeta_j
			\bigg) \dVol = 0
		\end{equation*}
		\item For any compact subset $\mathcal{K} \subset \mathcal{M}$ and any $\delta > 0$ 
		\begin{equation*}
			\lim_{n\rightarrow \infty} \bigg(
				\sum_{j=0}^1 \int_{\mathcal{K}}
					|T^j \psi_{\tau_n} - T^j \tilde{\psi}|^2 \dVol
				+ \sum_{j=0}^1 \int_{\mathcal{K} \setminus (\mathcal{S}_{(\text{ext})})_{(\delta)}}
					| \nabla (T^j \psi_{\tau_n}) - \nabla (T^j \tilde{\psi}) |^2 _{g_{(\text{ref})}}
					\dVol
				\bigg) = 0
		\end{equation*}
	\end{itemize}
	
	where $g_{(\text{ref})}$ is an arbitrarily chosen smooth, $T$-invariant Riemannian metric on $\mathcal{M}$.
	
\end{lemma}

Now, by our assumption on the spacetime, since $\tilde{\psi}$ solves the wave equation and vanishes outside the ergoregion, we actually have $\tilde{\psi} \equiv 0$ everywhere on $\mathcal{M}$ (see remark \ref{remark unique continuation})

If we apply this lemma also to the field $T\phi$ (i.e.\ we take one more $T$ derivative), then we can obtain, in particular, that there is a sequence of times $\tau_n$ such that
\begin{equation*}
	\lim_{n\rightarrow \infty}
		\sum_{j=0}^2 \int_{\tau_n}^{\tau_n + 1} \left(
			\int_{\mathcal{S}_{(\text{ext})} \cap \Sigma_\tau }
				|T^j \psi_{\tau_n}|^2 \dVol \right) \upd \tau
		= 0
\end{equation*}

Now, in $\mathcal{S}_{(\text{ext})} \setminus \mathcal{S} $, $T$ is a timelike vector field. Therefore, by using elliptic estimates and the mean value theorem, we have that, for any $\delta > 0$, there is a time $\tilde{\tau}_n$ with $\tau_n \leq \tilde{\tau}_n \leq \tau_n + 1$ such that
\begin{equation}
\label{equation decay on interior}
	\sum_{|j| \leq 2} \mathcal{E}^{(N)}_{( \mathcal{S}_{(\text{ext})} \setminus \mathcal{S}_{(\delta)} )}[\partial^{(j)}\psi_{\tau_n}] ( \tilde{\tau}_n) \rightarrow 0
\end{equation}

\subsection{Local energy decay on the evanescent ergosurface and energy amplification}

The purpose of this section is to show that, if the local energy decay of subsection \ref{subsection local decay away from ergosurface} can be extended to cover the evanescent ergosurface, then this leads to a contradiction with the boundedness assumption \ref{assumption boundedness}. This lies at the heart of our argument for instability. Put another way, we show that under assumption \ref{assumption boundedness} we \emph{cannot} extend the local energy decay to cover the evanescent ergosurface. This is the form taken by the following proposition.

\begin{proposition}
\label{proposition non decay on ergosurface}
 Let $\mathcal{M}$ be an asymptotically flat or an asymptotically Kaluza-Klein manifold with an evanescent ergosurface in the sense of \ref{condition asymp flat}, or an asymptotically Kaluza-Klein manifold with an evanescent ergosurface in the sense of \ref{condition asymp KK}. Suppose also that the manifold admits a discrete isometry as in section \ref{section the discrete isometry}. Additionally, suppose that the boundedness statement \ref{assumption boundedness} holds.

 Let $\mathcal{U}_0 \subset \Sigma_0$ be any compact set such that $(\mathcal{S}\cap\Sigma_0 ) \subset \mathcal{U}_0$. Then there exists some positive constant $\mathring{C}$ and a solution to the wave equation $\Box_g \phi = 0$ such that, for all times $t > 0$ the local energy of its $T$ derivatives in the set $\mathcal{U}$ is at least $\mathring{C}$, i.e.\
 \begin{equation}
  \mathcal{E}^{(N)}_{\mathcal{U}}[T\phi](t) + \mathcal{E}^{(N)}_{\mathcal{U}}[T^2 \phi](t) \geq \mathring{C}
 \end{equation}

 Moreover, the solution $\phi$ can be chosen to be smooth and to arise from compactly supported initial data satisfying
 \begin{equation*}
  \sum_{j = 0}^{3} \mathcal{E}^{(N)}[T^j u](0) < \infty
 \end{equation*}

\end{proposition}

We can compare the conclusion of proposition \ref{proposition non decay on ergosurface} runs counter to the conclusion of proposition \ref{proposition decay away from extended ergosurface}. We are showing that, if the set $\mathcal{U}_0$ is allowed to contain the ergosurface $\mathcal{S}$, then assumption \ref{assumption boundedness} leads to the exact opposite behaviour to the case where the set $\mathcal{U}_0$ is disjoint from the ergosurface. We shall see further consequences of this conclusion in the next subsection.

\begin{proof}
 The proof of proposition \ref{proposition non decay on ergosurface} will proceed by contradiction. That is, we shall suppose that, for all solutions of the wave equation arising from suitable initial data, the local energy in the set $\mathcal{U}_0$ eventually becomes arbitrarily small. We shall then derive a contradiction with assumption \ref{assumption boundedness}.

 To be precise, suppose the following: for all smooth solutions $\phi$ to $\Box_g \phi = 0$ such that $\phi$ arises from compactly supported initial data satisfying
 \begin{equation*}
  \sum_{j = 0}^{3} \mathcal{E}^{(N)}[T^j \phi](0) < \infty
 \end{equation*}
 Then, for all sets $\mathcal{U}_t$ as defined in proposition \ref{proposition non decay on ergosurface} and for all $\epsilon > 0$ there is some time $\tau$ such that
 \begin{equation*}
  \mathcal{E}^{(N)}_{\mathcal{U}}[T\phi](\tau) + \mathcal{E}^{(N)}_{\mathcal{U}}[T^2 \phi](\tau) < \epsilon
 \end{equation*}
 Suppose in addition that assumption \ref{assumption boundedness} holds.

Using the discrete isometry $\mathscr{T}$ we can see that the same results would hold in the \emph{time reversed} manifold, which is the manifold with the other choice of time orientation. To see this, we note the following: we can combine the discrete isometry $\mathscr{T}$ with the one-parameter family of isometries associated to the Killing vector field $T$ (or $V$), which we label $F_t$, to form the isometry
\begin{equation*}
	\mathscr{T}_t := F_{t} \circ \mathscr{T} \circ F_{-t}
\end{equation*}
Since $\mathscr{T}$ fixes the hypersurface $\Sigma_0$, we find that $\mathscr{T}_t$ is a discrete isometry fixing the hypersurface $\Sigma_t$. In particular, $\mathscr{T}_t$ descends to a discrete isometry of $\Sigma_t$ together with its induced metric, which we denote by $\overline{\mathscr{T}}_t$.

Now, suppose that $\phi$ is a solution to the wave equation, inducing the following data on the hypersurface $\Sigma_t$:
\begin{equation*}
\begin{split}
	\phi\big|_{\Sigma_t} &= \phi_0 \\
	T\phi\big|_{\Sigma_t} &= \phi_1
\end{split}
\end{equation*}
Then $(\mathscr{T}_t^{-1})^*(\phi)$ will be a solution to the wave equation on the \emph{time reversed} manifold: that is, on the manifold $(\mathscr{T}_t^{-1}(\mathcal{M}), (\mathscr{T}_t^{-1})^*(g)) = (\mathcal{M}, g)$. Moreover, this solution will induce data on the hypersurface $\Sigma_t$ given by
\begin{equation*}
\begin{split}
	(\mathscr{T}_t^{-1})^*(\phi)\big|_{\Sigma_t} &= (\overline{\mathscr{T}}_t^{-1})^* \phi_0 \\
	T \left( (\mathscr{T}_t^{-1})^*(\phi)\right)\big|_{\Sigma_t} &= -(\overline{\mathscr{T}}_t^{-1})^*\phi_1
\end{split}
\end{equation*}
Since $\phi_0$ and $\phi_1$ are smooth and compactly supported, this initial data is also smooth and compactly supported. Moreover, $(\mathscr{T}_t^{-1})^*(\phi)$ solves $\Box_g (\mathscr{T}_t^{-1})^*(\phi) = 0$. Hence this solution will disperse in the future: for any $\epsilon$ and for any compact set $\mathcal{U}_0$ there is some time $\tau > 0$ such that
 \begin{equation*}
  \mathcal{E}^{(N)}_{\mathcal{U}}[T(\mathscr{T}_t^{-1})^*(\phi)](\tau) + \mathcal{E}^{(N)}_{\mathcal{U}}[T^2 (\mathscr{T}_t^{-1})^*(\phi)](\tau) < \epsilon
 \end{equation*}
 If we now apply the discrete isometry $\mathscr{T}_t$ to this solution, we find\footnote{Note that the energy, on a surface which is fixed by the discrete isometry, is invariant under this isometry.} that
 \begin{equation*}
  \mathcal{E}^{(N)}_{\mathcal{U}}[T\phi](t-\tau) + \mathcal{E}^{(N)}_{\mathcal{U}}[T^2 \phi](t-\tau) < \epsilon
 \end{equation*}

Note that we have made use of the fact that the dispersion result holds for \emph{all} initial data. It is also important that the discrete isometry fixes a Cauchy hypersurface, since this allows us to pick ``time reversed'' initial data.

 We now make the following claim: 
 \begin{claim}
 \label{claim initial data 1}
  In the asympotitcally flat case, for all $\delta > 0$ and for any $\tau_0$ there exists data (on $\Sigma_{\tau_0}$) for the wave equation such that
  \begin{equation}
   \begin{split}
    \mathcal{E}^{(N)}[Tu](\tau_0) & \geq \delta^{-1} \\
    \mathcal{E}^{(T)}[Tu](\tau_0) & = 1 \\
    \sum_{j=0}^3 \mathcal{E}^{(N)}[T^j u](\tau_0) &< \infty
   \end{split}
  \end{equation}
  Likewise, in the asymptotically Kaluza-Klein case, for all $\delta > 0$ and for any $\tau_0$ there exists $\mathcal{G}$-invariant data such that
  \begin{equation}
   \begin{split}
    \mathcal{E}^{(N)}[Tu](\tau_0) & \geq \delta^{-1} \\
    \mathcal{E}^{(V)}[Tu](\tau_0) & = 1 \\
    \sum_{j=0}^3 \mathcal{E}^{(N)}[T^j u](\tau_0) &< \infty
   \end{split}
  \end{equation}
 \end{claim}
 
 We postpone the proof of this claim, and first show that, combined with the assumptions above, this leads to a contradiction. To exhibit this contradiction, we begin with initial data as in claim \ref{claim initial data 1} at some time $\tau_0$, and evolve it \emph{backwards} in time (equivalently, we evolve it forward in time on the time-reversed manifold). We then find that, using the dispersion property derived above, this data will then disperse in the following sense: given any set $\mathcal{U}_0$, given any $\epsilon > 0$ we can find some time $\tau_1$ such that
 \begin{equation*}
  \mathcal{E}^{(N)}_{\mathcal{U}}[T\phi](\tau_0 - \tau_1) < \epsilon
 \end{equation*}
 In particular, we can pick $\epsilon = 1$ and also use the isometry generated by the Killing vector $T$ to translate the solution in time, so that $\tau_0 = \tau_1$. Additionally, we can pick the set $\mathcal{U}_0$ to include the evanescent ergosurface $\mathcal{S}$. Moreover, we can apply the boundedness assumption \ref{assumption boundedness} to the waves $\phi$, $T\phi$, $T^2 \phi$ and $T^3 \phi$ (using the fact that, since $T$ is a Killing vector field, $T^j \phi$ also solves the wave equation). Finally, we can use the fact that the $T$-energy is conserved to deduce that the $T$-energy at the initial time is $\mathcal{E}^{(T)}[T\phi](0) = 1$.
 
We arrive at initial data on $t = 0$ such that, in the case of the first kind of evanescent ergosurface \ref{condition asymp flat}
\begin{equation*}
 \begin{split}
  \mathcal{E}^{(N)}_{\mathcal{U}}[T\phi](0) &< 1 \\
  \mathcal{E}^{(T)}(0) &= 1 \\
  \sum_{j=0}^3 \mathcal{E}^{(N)}[T^j \phi](0) &< \infty  
 \end{split}
\end{equation*}
Furthermore, this data is such that, at the time $\tau_0$, we have
\begin{equation*}
 \mathcal{E}^{(N)}[T\phi](\tau_0) \geq \delta^{-1}
\end{equation*}

Similarly, given a manifold with the second kind of evanescent ergosurface \ref{condition asymp KK} we arrive at $\mathcal{G}$-invariant initial data at $t = 0$ such that
\begin{equation*}
 \begin{split}
  \mathcal{E}^{(N)}_{\mathcal{U}}[T\phi](0) &< 1 \\
  \mathcal{E}^{(V)}(0) &= 1 \\
  \sum_{j=0}^3 \mathcal{E}^{(N)}[T^j \phi](0) &< \infty  
 \end{split}
\end{equation*}

Now we only need to show that the \emph{global} $N$-energy is of order $1$ initially, and not just the local $N$-energy, as stated above. Since $N$ is \emph{uniformly} timelike, we have
\begin{equation*}
 \mathcal{E}^{(N)}[T\phi](t) \sim || \partial T\phi ||_{L^2[\Sigma_t]}^2
\end{equation*}
Combining this estimate with the estimates in section \ref{section energy momentum} and the definition of the evanescent ergosurface of the first kind (\ref{condition asymp flat}) we see that 
\begin{equation*}
 \mathcal{E}^{(T)}_{\mathcal{M}\setminus\mathcal{U}}[T\phi](t) \sim || \partial T\phi ||^2_{L^2[\Sigma_t\setminus \mathcal{U}_t]} \sim \mathcal{E}^{(N)}_{\mathcal{M}\setminus\mathcal{U}}[T\phi](t)
\end{equation*}

In other words, outside of the set $\mathcal{U}_t$, the $T$-energy $\mathcal{E}^{(T)}$ \emph{is} comparable to the $N$-energy $\mathcal{E}^{(N)}$. Similarly, we saw in proposition \ref{proposition V energy current degenerate} that, for $\mathcal{G}$-invariant waves, the $V$-energy $\mathcal{E}^{(V)}$ is comparable to the $N$-energy outside of the set $\mathcal{U}_t$. However, in the asymptotically flat case, the \emph{global} $T$-energy is bounded by 1 (at all times, since this energy is conserved). Similarly, in the asymptotically Kaluza-Klein case, the global $V$-energy is bounded by 1 at all times.

Hence, we find that, for the data defined above, in the asymptotically flat case the data is such that
\begin{equation*}
 \begin{split}
  \mathcal{E}^{(N)}[T\phi](0) &\lesssim 1 \\
  \mathcal{E}^{(T)}[T\phi](0) &= 1 \\
  \sum_{j=0}^3 \mathcal{E}^{(N)}[T^j \phi](0) &< \infty  
 \end{split}
\end{equation*}
Note, importantly, that the bound on the initial $N$-energy is \emph{independent} of $\delta$. Similarly, in the asymptotically Kaluza-Klein case we find initial data such that
\begin{equation*}
 \begin{split}
  \mathcal{E}^{(N)}[T\phi](0) &\lesssim 1 \\
  \mathcal{E}^{(V)}[T\phi](0) &= 1 \\
  \sum_{j=0}^3 \mathcal{E}^{(N)}[T^j \phi](0) &< \infty  
 \end{split}
\end{equation*}
while, in both cases, at time $\tau$ we have
\begin{equation*}
 \mathcal{E}^{(N)}[T\phi](\tau) \geq \delta^{-1}
\end{equation*}

Since we can pick $\delta$ arbitrarily small, and since $T\phi$ obeys the wave equation, we arrive at a contradiction with assumption \ref{assumption boundedness}. In fact, we have shown that the claim that the local energy decay statement of proposition \ref{proposition decay away from extended ergosurface} can be extended to cover also the evanescent ergosurface leads to a contradiction with assumption \ref{assumption boundedness}.

Subject to proving claim \ref{claim initial data 1}, we have finished the proof.
\end{proof}

\subsection{Constructing the initial data}
\label{subsection constructing the initial data}

To finish the proof of proposition \ref{proposition non decay on ergosurface}, we need only to construct initial data to prove claim \ref{claim initial data 1}. This turns out to be fairly difficult, and the construction will be the subject of this subsection. We shall need to make detailed use of the adapted coordinate system we considered in section \ref{section adapted coordinates}.

In the case of an asymptotically flat manifold with an evanescent ergosurface of the first kind \ref{condition asymp flat}, suppose the support of the wave $(T\phi)\big|_{\Sigma_t}$ is contained in a region $\mathcal{U}_t$, which is sufficiently small that the coordinates in section \ref{section adapted coordinates} are defined in this region. Then the $N$-energy of the wave $T\phi$ is given by
\begin{equation}
\label{equation N energy}
 \mathcal{E}^{(N)}[T\phi](\tau) \sim \int_{\Sigma_\tau} \left( (TT\phi)^2 + (\partial_y T\phi)^2 + \sum_a (\partial_{x^a} T\phi)^2 \right) \upd y \, \upd x^1 \ldots \upd x^{D-1}
\end{equation}
whereas the conserved $T$-energy of the wave is
\begin{equation}
\label{equation T energy}
 \mathcal{E}^{(T)}[T\phi](\tau) \sim \int_{\Sigma_\tau} \left( (T T\phi)^2 + \sum_a (\partial_{x^a} T\phi)^2 + \mathcal{O}(|\tilde{x}|^2)(\partial T\phi) \right) \upd y \, \upd x^1 \ldots \upd x^{D-1}
\end{equation}

Recall that in these coordinates we have $T = \partial_{\tilde{t}}$. Note, however, that the hypersurface $\Sigma_t$ is not necessarily locally a surface of constant $\tilde{t}$, and so we must bear in mind that, on $\Sigma_t$, the coordinate $\tilde{t}$ should be considered a function of the other coordinates $(y, x^a)$.

If we could freely prescribe both $T\phi$ and $TT\phi$ on $\Sigma_t$ then it would be very easy to prescribe initial data satisfying claim \ref{claim initial data 1}. However, we can only prescribe $\phi$ and $T\phi$ on $\Sigma_t$. Higher order \emph{spatial} derivatives of these quantities can then be obtained by taking the spatial derivatives of this data, however, the quantity $TT\phi$ is constrained by the wave equation to take on values which depend on the other derivatives.

Specifically, we have that, if $\phi$ solves the wave equation $\Box_g \phi = 0$ then in the adapted coordinates the expression \eqref{equation wave local coords AF} gives
\begin{equation}
\label{equation TTphi}
 \begin{split}
  TT \phi &=  2A\partial_y T\phi + A^2\delta^{ab}\partial_a \partial_b \phi + \left( |b|^2 - a \right)\partial^2_y \phi + \mathcal{O}(|\tilde{x}| |\partial \phi|) + \mathcal{O}(|\tilde{x}| |\partial T \phi|) \\
  &\phantom{=} + \mathcal{O}(|\tilde{x}^2||\bar{\partial}\partial \phi|) + \mathcal{O}(|\tilde{x}|^3 |\partial^2 \phi|)
 \end{split}
\end{equation}

Let $\chi$ be a smooth cut-off function such that
\begin{equation*}
 \begin{split}
  \chi(x) = 0 \text{\quad for } x \geq 1 \\
  \chi(x) = 1 \text{\quad for } x \leq \frac{1}{2}
 \end{split}
\end{equation*}

As mentioned above, we are free to prescribe both $u\big|_{\Sigma_t}$ and $T\phi\big|_{\Sigma_t}$. As a preliminary step, we first make the following choices for the same quantities associated to a function $u_0$:
\begin{equation*}
 \begin{split}
  u_0\big|_{\Sigma_t} &= \mathfrak{Re}\left( e^{i\delta_0^{-1}y } F_1(x) \right) \chi \left (\delta_0^{-\frac{3}{4}}|y| \right) \\
  T\phi_0\big|_{\Sigma_t} &= \mathfrak{Re}\left( -i \omega e^{i\delta_0^{-1}y } F_1(x) \right) \chi \left (\delta_0^{-\frac{3}{4}}|y| \right) \\
 \end{split}
\end{equation*}
where $\delta_0$ and $\omega$ are constants (to be fixed below), and $F_1(x)$ is functions (to be defined below) which depends only on the coordinates $x^a$ (and not on $y$ or $\tilde{t}$).

Recall that we have
\begin{equation*}
 0 \leq \left( |b|^2 - a \right) = \mathcal{O}(|\tilde{x}|^2)
\end{equation*}
Evaluating this at $y = 0$, in terms of the coordinates $x^a$ we can write
\begin{equation*}
 \left( |b|^2 - a \right)\big|_{y=0} = M_{ab}x^a x^b + \mathcal{O}(|x|^3)
\end{equation*}
for some symmetric matrix $M_{ab}$. Since $M$ is symmetric, we can diagonalise it by making some orthogonal transformation on the coordinates $x^a$, that is, we can define new variables
\begin{equation*}
 x'^a := R_b^{\phantom{b}a} x^b
\end{equation*}
where $R$ is an orthogonal matrix, and where the matrix $M$ is diagonal in this basis. Note that the coordinate derivatives $\partial_{x'^a}$ are still orthonormal at $p$ and satisfy the same conditions as the original coordinate derivatives, and so the form of the metric (and hence the wave equation) is unchanged by this change of variables. From now on, we will assume that this change of basis has been made, and drop the prime on the coordinates $x^a$.

Since $M$ is a positive matrix its eigenvalues are non-negative. Associated to each eigenvalue is a coordinate $x^a$. We now split the coordinates $x^a$ into two sets: those associated with a nonzero eigenvalue for $M$, and those associated with a zero eigenvalue of $M$. That is, we define
\begin{equation*}
 \begin{split}
  X_1 = \{a \in \{1, 2,\ldots \} \big| M_a^{\phantom{a}b} x^a = 0 \} \\
  X_2 = \{a \in \{1, 2,\ldots \} \big| M_a^{\phantom{a}b} x^a \neq 0 \}
 \end{split}
\end{equation*}
Note that either one of these sets might be empty. We also define the notation
\begin{equation*}
\begin{split}
	|x| &:= \sqrt{\sum_{a \in X_1 \cup X_2} (x^a)^2 } \\
	|x|_1 &:= \sqrt{\sum_{a \in X_1} (x^a)^2 } \\
	|x|_2 &:= \sqrt{\sum_{a \in X_2} (x^a)^2 } \\
\end{split}
\end{equation*}

We now define
\begin{equation*}
 F_1(x) := \chi\left( \delta_0^{-\frac{2}{5}}|x|_1 \right) F_2(x)
\end{equation*}
where $F_2$ is a function \emph{only} of the coordinates $x^a$ for $a \in X_2$.

We now construct the function $F_2(x)$ by requiring it to satisfy
\begin{equation}
\label{equation poisson F_2}
 \begin{split}
  -\Delta F_2 + A^{-2} \delta_0^{-2} (M_{ab} x^a x^b) F_2 &= 2A^{-1} \omega \delta_0^{-1} F_2 \\
  F_2(x) = 0 \text{\quad when } |x|_2 = \delta_0^{\frac{1}{2} - \delta_1}
 \end{split}
\end{equation}
where $\Delta$ is the Laplacian type operator $\sum_{a \in X_2} \partial_a^2$ and $\delta_1$ is some small constant that we fix below.

We wish to view equation \eqref{equation poisson F_2} as an elliptic eigenvalue problem for the function $F_2$ and its associated eigenvalue $\omega$. We first rescale the coordinates by a factor of $\delta_0^{-\frac{1}{2}}$:
\begin{equation*}
 \bar{x}^a := \delta_0^{-\frac{1}{2}}x^a \text{\quad for } a \in X_2
\end{equation*}
Then, using an overbar to refer to quantities defined by replacing the coordinate $x^a$ with $\bar{x}^a$ throughout, we arrive at the eigenvalue problem
\begin{equation}
\label{equation poisson F_2 v2}
 \begin{split}
  -\bar{\Delta} F_2 + A^{-2} (M_{ab} \bar{x}^a \bar{x}^b) F_2 &= 2A^{-1} \omega F_2 \\
  F_2(\bar{x}) &= 0 \text{\quad when } |\bar{x}|_2 = \delta_0^{-\delta_1}
 \end{split}
\end{equation}
where we consider $\omega$ as the eigenvalue and $M_{ab}\bar{x}^a\bar{x}^b$ as the potential for this eigenvalue problem.

By considering the variational formulation of this problem, we can place a lower bound on the number of eigenvalues $\omega$ below some threshold $\omega_{\text{max}}$. Specifically, let $N(\omega_{\text{max}}; \delta_0)$ be the number of positive eigenvalues $\omega$ for the problem \eqref{equation poisson F_2 v2} satisfying $\omega \leq \omega_{\text{max}}$. Let $N_+(\mathcal{U}; \delta_0)$ be the number of positive eigenvalues $\omega_+$ satisfying the same bound, where $\omega_+$ is an eigenvalue for the related problem
\begin{equation*}
 \begin{split}
  -\bar{\Delta} F_2 + A^{-2} \sup_{\bar{x} \in \mathcal{U}} \left( M_{ab} \bar{x}^a \bar{x}^b \right ) F_2 &= 2A^{-1} \omega_+ F_2 \\
  F_2(\bar{x})\big|_{\partial \mathcal{U}} &= 0 
 \end{split}
\end{equation*}
where $\mathcal{U} \subset \{|\bar{x}|_2 \leq \delta_0^{-\delta_1} \}$. Then we have
\begin{equation*}
 N(\omega_{\text{max}}; \delta_0) \geq N_+(\mathcal{U}; \delta_0) 
\end{equation*}
In particular, we can take the $\mathcal{U}$ to be the cubic region with unit volume
\begin{equation*}
 \mathcal{U} := \{ |\bar{x}^a| < 1 \text{ \quad for } a \in X_2 \}
\end{equation*}
Note that this set is indeed a subset of $\{|\bar{x}|_2 \leq \delta_0^{-\delta_1} \}$ for all sufficiently small $\delta_1$.

Then we can explicitly calculate the positive eigenvalues for \eqref{equation poisson F_2 v2}: they are given by
\begin{equation*}
 \omega_+ = \frac{1}{2}\left( A^{-1}\lambda_{\text{max}} + A \pi^2 \sum_{i = 1}^{|X_2|} n_i^2 \right) \text{ , \quad } n_i \in \mathbb{N}
\end{equation*}
where $\lambda_{\text{max}}$ is the largest eigenvalue of the matrix $M$, and where at least one of the $n_i$ is non-zero.

In particular, this proves the following proposition:

\begin{proposition}
 For all sufficiently small $\delta_0$, there exists a function $F_2$ and an eigenvalue $\omega$ solving the problem \eqref{equation poisson F_2 v2}. Moreover, there exists some $\omega_{\text{max}} > 0$, which is \emph{independent} of $\delta_0$, such that for all sufficiently small $\delta_0$ we can find an eigenvalue $\omega$ satisfying
\begin{equation}
 \omega \leq \omega_{\text{max}}
\end{equation}
 \end{proposition}

We can also establish some basic properties of the eigenfunction $F_2$. By linearity, we can rescale $F_2$ and so we can assume that
\begin{equation*}
 \int_{|x|_2 < \delta_0^{\frac{1}{2} - \delta_1}}(F_2)^2 \prod_{a\in X_2} \upd x^a = 1
\end{equation*}
Now, multiplying the eigenvalue equation \eqref{equation poisson F_2 v2} by $F_2$ and integrating by parts (using the boundary conditions satisfied by $F_2$) we find that
\begin{equation*}
 \int_{|\bar{x}|_2 < \delta_0^{-\delta_1}} \left( \sum_{a \in X_2} (\partial_{\bar{x}^a} F_2)^2 + A^{-2} (M_{ab} \bar{x}^a \bar{x}^b)(F_2)^2 \right) \prod_{a\in X_2} \upd x^a = 2A^{-1} \omega 
\end{equation*}
In particular, transforming back to the non-rescaled coordinates, we have
\begin{equation*}
 \int_{|x|_2 < \delta_0^{\frac{1}{2} - \delta_1}} \left( \sum_{a \in X_2} (\partial_a F_2)^2 \right) \prod_{a \in X_2} \upd x^a \leq \delta_0^{-1} 2A^{-1} \omega
\end{equation*}
Recall that we are choosing $\omega$ to be bounded independently of $\delta_0$. Hence, the $L_2$ norm of the derivatives of $F_2$ scales as $\delta_0^{-1}$.

 \subsubsection{Agmon estimates}
 
 We can obtain more detailed information regarding the behaviour of the eigenfunction through the use of ``Agmon estimates'' (see, for example, \cite{Holzegel2014}). These quantify the size of the solution in the so-called ``forbidden region''. Specifically, we can define a kind of forbidden region:
 \begin{equation*}
  \mathcal{U}_{\text{forbidden}}(\delta_1) := \left\{ \bar{x} \quad \big| \quad |\bar{x}|_2 \leq \delta_0^{-\delta_1} \ , \ A^{-2} M_{ab} \bar{x}^a \bar{x}^b - 2A^{-1} \omega > \delta_2 \right\}
 \end{equation*}
 for some positive constant $\delta_2$.

 We also define the ``classical region'':
 \begin{equation*}
  \mathcal{U}_{\text{classical}} := \left\{ \bar{x} \quad \big| \quad |\bar{x}|_2 \leq C_0 \ , \ A^{-2} M_{ab} \bar{x}^a \bar{x}^b - 2A^{-1} \omega < 0 \right\}
 \end{equation*}
 
 We define the ``Agmon distance'' between points $x$ and $y$ with coordinates $\bar{x} = (\bar{x}^a)$, $a \in X_2$ and $\bar{y} = (\bar{y}^a)$, $a \in X_2$. This distance is defined as
 \begin{equation*}
  d_\omega(\bar{x},\bar{y}) := \inf_{\substack{ \gamma:[0,1] \rightarrow (\bar{x}^a), \ a \in X_2 \\ \gamma \text{ smooth} \\ \gamma(0) = \bar{x}, \ \gamma(1) = \bar{y}}} \int_0^1 \sqrt{ \left( \sum_{a \in X_2} \left( \frac{\upd \gamma}{\upd s} (\bar{x}^a) \right)^2 \sup\left( A^{-2}M_{bc} \bar{x}^b \bar{x}^c - 2A^{-1} \omega \ , \ 0 \right) \right) } \upd s 
 \end{equation*}
 In other words, $d_\omega$ is the distance function defined with respect to the metric
 \begin{equation*}
  (g_{(\text{Agmon})})_{ab} = \sup\left( A^{-2} M_{bc} \bar{x}^b \bar{x}^c - 2A^{-1} \omega \ , \ 0 \right) \delta_{ab}
 \end{equation*}
 
 If we define, for some function $u(\bar{x}^a \big| a\in X_2)$, 
 \begin{equation*}
  |\bar{\nabla} u| = \sqrt{ \sum_a (\partial_{\bar{x}^a} u)^2 }
 \end{equation*}
 then we find that the distance function with some fixed point $\bar{x}_0$ satisfies
 \begin{equation*}
  | \bar{\nabla} d_\omega (\bar{x}, \bar{x}_0) | \leq \sqrt{ \sup\left( A^{-2} M_{bc} \bar{x}^b \bar{x}^c - 2A^{-1} \omega \ , \ 0 \right) }
 \end{equation*}
 We can also define the distance to the classical region:
 \begin{equation*}
  d_\omega^{\text{classic}}(\bar{x}) := \inf_{\substack{ |\bar{y}|_2 \leq \delta_0^{-\delta_1} \\ y \in \mathcal{U}_{\text{classical}} }} d_\omega (\bar{x}, \bar{y})
 \end{equation*}
 
 Now, we can prove the following proposition, which is a kind of exponentially weighted energy estimate, and which follows from integrating by parts.
 
 \begin{proposition}[An exponentially weighted energy estimate]
  Let $\phi$, $W$ and $D$ be smooth, real valued functions on $\{ |\bar{x}|_2 \leq C_0 \}$, such that $\phi(\bar{x}) = 0$ when $|\bar{x}|_2 = \delta_0^{-\delta_1}$. Then
  \begin{equation}
   \int_{|\bar{x}|_2 \leq \delta_0^{-\delta_1}} \left( |\bar{\nabla} (e^D \phi)|^2 + \left( W - |\bar{\nabla} D|^2 \right)e^{2D}|\phi|^2 \right) \prod_{a \in X_2} \upd \bar{x}^a
   = \int_{|x|_2 \leq \delta_0^{-\delta_1} } \left( - \bar{\Delta} \phi + W u \right) u e^{2D} \prod_{a \in X_2} \upd \bar{x}^a
  \end{equation}
 \end{proposition}
 
 Now, we apply this proposition with the following choices:
 \begin{equation*}
  \begin{split}
   \phi &= F_2 \\
   W &= A^{-2} (M_{ab} x^a x^b) - 2A^{-1} \omega \\
   D &= (1-\delta_2) d_\omega^{\text{classic}}
  \end{split}
 \end{equation*}
 This leads to the estimate
 \begin{equation*}
 \begin{split}
 	& \int_{\mathcal{U}_{\text{forbidden}}} \left| \bar{\nabla} \left( e^{(1-\delta_2) d_\omega^{\text{classic}}} F_2 \right) \right|^2 \prod_{a \in X_2} \upd \bar{x}^a
 	\\
 	& + \int_{\mathcal{U}_{\text{forbidden}}} \left( A^{-2} (M_{bc} \bar{x}^b \bar{x}^c) - 2A^{-1} \omega - (1-\delta_2)^2 \left| \bar{\nabla} d_\omega^{\text{classic}} \right|^2 \right) e^{2(1-\delta_2) d_\omega^{\text{classic}}} |F_2|^2 \prod_{a \in X_2} \upd \bar{x}^a
 	\\
 	& = -\int_{ \left\{|\bar{x}|_2 \leq \delta_0^{-\delta^1} \right\}\setminus \mathcal{U}_{\text{forbidden}}} \left| \bar{\nabla} \left( e^{(1-\delta_2) d_\omega^{\text{classic}}} F_2 \right) \right|^2 \prod_{a \in X_2} \upd \bar{x}^a
 	\\
 	& \phantom{=}
 	+ \int_{ \left\{|\bar{x}|_2 \leq \delta_0^{-\delta_1} \right\}\setminus \mathcal{U}_{\text{forbidden}}} \Big(
 	2A^{-1} \omega - A^{-2} (M_{bc} \bar{x}^b \bar{x}^c)
 	\\
 	&\phantom{= + \int_{ \left\{|\bar{x}|_2 \leq \delta_0^{-\delta_1} \right\}\setminus \mathcal{U}_{\text{forbidden}}} \Big(}
 	+ (1-\delta_2)^2 \left| \bar{\nabla} d_\omega^{\text{classic}} \right|^2 \Big) e^{2(1-\delta_2) d_\omega^{\text{classic}}} |F_2|^2 \prod_{a \in X_2} \upd \bar{x}^a
  \end{split}
 \end{equation*}
 
 Now, in the forbidden region we can calculate
 \begin{equation*}
  \left| \bar{\nabla} d_{\omega}^{\text{classic}} \right|^2 = A^{-2} (M_{bc} \bar{x}^b \bar{x}^c) - 2A^{-1} \omega > \delta_2
 \end{equation*}
 
 Some straightforward calculations now lead to the estimate
 \begin{equation*}
  \begin{split}
   &\int_{\mathcal{U}_{\text{forbidden}}} \left| \bar{\nabla} \left( e^{(1-\delta_2) d_\omega^{\text{classic}}} F_2 \right) \right|^2 \prod_{a \in X_2} \upd \bar{x}^a 
   + \delta_2^2 \int_{\mathcal{U}_{\text{forbidden}}} |F_2|^2 e^{2(1-\delta_2) d_\omega^{\text{classic}}} \prod_{a \in X_2} \upd \bar{x}^a \\
   &\lesssim A^{-1} \omega e^{2(1-\delta_2) a_\omega(\delta_2)} \int_{ \left\{|\bar{x}|_2 \leq \delta_0^{-\delta_1} \right\}\setminus \mathcal{U}_{\text{forbidden}}} |F_2|^2 \prod_{a \in X_2} \upd \bar{x}^a
  \end{split}
 \end{equation*}
 for sufficiently small $\delta_2$, and where $a_\omega(\delta_2)$ is defined as
 \begin{equation*}
 a_\omega(\delta_2) := \sup_{\bar{x} \in \left\{|\bar{x}|_2 \leq \delta_0^{-\delta_1} \right\}\setminus \mathcal{U}_{\text{forbidden}}} d_\omega^{\text{classic}}(\bar{x})
 \end{equation*}
 i.e.\ $a_\omega(\delta_2)$ is the largest Agmon distance from the complement of the forbidden region to the classical region.
 
 As $\delta_2 \rightarrow 0$, we have $a_\omega(\delta_2) \rightarrow 0$. More precisely, if we pick any $\delta_3 > 0$, then there is some choice of $\delta_2 > 0$ such that $a_\omega(\delta_2) \leq \frac{1}{2}\delta_3$. Moreover, this bound holds \emph{independently} of the choice of $\delta_0$, at least for all sufficiently small $\delta_0$. We now fix this choice of $\delta_2$.
 
 Recall that we can choose the eigenfunction to satisfy the bound $\omega \leq \omega_{\text{max}}$ independently of $\delta_0$. From now on, we make this choice for the eigenvalue $\omega$. For any subset of the forbidden region, $\mathcal{U}_{\text{f}} \subset \mathcal{U}_{\text{forbidden}}$, we have
 \begin{equation}
 \label{equation Agmon 1}
  \begin{split}
   &\int_{\mathcal{U}_{\text{f}}} \left| \bar{\nabla} \left( e^{(1-\delta_2) d_\omega^{\text{classic}}} F_2 \right) \right|^2 \prod_{a \in X_2} \upd \bar{x}^a 
   + \int_{\mathcal{U}_{\text{f}}} |F_2|^2 e^{2(1-\delta_2) d_\omega^{\text{classic}}} \prod_{a \in X_2} \upd \bar{x}^a \\
   &\lesssim e^{\delta_3} \int_{ \left\{|\bar{x}|_2 \leq \delta_0^{-\delta_1} \right\}\setminus \mathcal{U}_{\text{forbidden}}} |F_2|^2 \prod_{a \in X_2} \upd \bar{x}^a
  \end{split}
 \end{equation}
 Note that the constants depend on $\delta_3$, but by picking $\delta_2$ suitably small we are able to make $\delta_3 > 0$ as small as we like. Moreover, we can make such a choice for the constant $\delta_2$ \emph{independent} of the value of $\delta_0$.
 
 We now choose the subset of the forbidden region to be defined by
 \begin{equation*}
  \mathcal{U}_\text{f} := \left\{ \bar{x} \ \big| \ \frac{1}{2} \delta_0^{-\delta_1} < |\bar{x}|_2 \leq \delta_0^{-\delta_1} \right\}
 \end{equation*}
 Note that this is indeed a subset of the forbidden region for all sufficiently small $\delta_0$. In addition, within this subset we have the lower bound
 \begin{equation*}
  d_\omega^{\text{classic}} \gtrsim \delta_0^{-2\delta_1}
 \end{equation*}
 which follows from the fact that the potential grows quadratically in the forbidden region.
 
 With this choice for the region $\mathcal{U}_\text{f}$, we return to equation \eqref{equation Agmon 1}. Dropping the first term on the left hand side, and making use of the lower bound on $d_\omega^{\text{classic}}$ in the region $\mathcal{U}_\text{f}$, we find that there is some $c_1 > 0$ such that
 \begin{equation*}
  \int_{U_\text{f}} |F_2|^2 \prod_{a \in X_2} \upd \bar{x}^a \lesssim e^{-c_1 \delta_0^{-2\delta_1}} \int_{|\bar{x}|_2 \leq \delta_0^{-\delta_1}} |F_2|^2 \prod_{a \in X_2} \upd \bar{x}^a 
 \end{equation*}
 or, returning to the non-rescaled coordinates $x^a$ and remembering our normalisation condition on $F_2$, we have
 \begin{equation*}
  \int_{\frac{1}{2}\delta_0^{-\delta_1} \leq |x|_2 \leq \delta_0^{\delta_1}}  |F_2|^2 \prod_{a \in X_2} \upd x^a \lesssim e^{-c_1 \delta_0^{-2\delta_1}}
 \end{equation*}
 In other words, the $L^2$ norm of $F_2$ is exponentially small in the region $U_f$, for small $\delta_0$.
 
 Now we return again to equation \eqref{equation Agmon 1}, and this time we drop the second term on the left hand side. We can expand the first term as
 \begin{equation*}
  \int_{\mathcal{U}_{\text{f}}} e^{2(1-\delta_2)d_\omega^{\text{classic}}} \left( |\bar{\nabla} F_2|^2 + 2F_2(1-\delta_2)(\bar{\nabla} d_\omega^{\text{classic}})\cdot (\bar{\nabla}F_2) + (1-\delta_2)^2 |\bar{\nabla} d_\omega^{\text{classic}} |^2 \right) \prod_{a \in X_2} \upd \bar{x}^a
 \end{equation*}
 The third term is positive and we can immediately drop it. For the second term, the Cauchy-Schwarz inequality, allow us to write
 \begin{equation*}
  \left| 2F_2(1-\delta_2)(\bar{\nabla} d_\omega^{\text{classic}})\cdot (\bar{\nabla}F_2) \right| \leq \frac{1}{2}|\bar{\nabla} F_2|^2 + 4(1-\delta_2)^2 |\bar{\nabla} d_\omega^{\text{classic}}|^2 |F_2|^2
 \end{equation*}
 Now, in the region $\mathcal{U}_f$ we have the bound
 \begin{equation*}
  |\bar{\nabla} d_\omega^{\text{classic}}|^2 \lesssim \delta_0^{-2\delta_1}
 \end{equation*}
 Combining this bound with the bound already obtained for the $L^2$ norm of $F_2$ in the region $\mathcal{U}_f$ we find that we can bound
 \begin{equation*}
  \int_{U_\text{f}} |\bar{\nabla}F_2|^2 \prod_{a \in X_2} \upd \bar{x}^a \lesssim \delta_0^{-2\delta_1} e^{-c_1 \delta_0^{-
  2\delta_1}} \int_{|\bar{x}|_2 \leq \delta_0^{-\delta_1}} |F_2|^2 \prod_{a \in X_2} \upd \bar{x}^a
 \end{equation*}
 for all sufficiently small $\delta_0$. Again, we can return to our non-rescaled coordinates, and also redefine $c_1$ to be some slightly smaller constant, and we find that the $L^2$ norm of $\nabla F_2$ is also exponentially supressed for small $\delta$. Putting together the results above, we have proved the following proposition:
 
 \begin{proposition}
 \label{proposition Agmon bounds}
  For all sufficiently small $\delta_0$, there exists a solution $(F_2, \omega)$ to the eigenvalue problem \eqref{equation poisson F_2}. Moreover, the eigenvalue $\omega$ can be chosen to satisfy the bound
  \begin{equation}
   \omega \leq \omega_{\text{max}}
  \end{equation}
  where $\omega_{\text{max}}$ is some sufficiently large constant, but which is \emph{independent} of $\delta_0$. With this choice of eigenvalue, the associated eigenfunction $F_2$ satisfies
  \begin{equation}
   \int_{\frac{1}{2}\delta_0^{\delta_1} \leq |x|_2 \leq \delta_0^{\delta_1} } \left( |F_2|^2 + |\nabla F_2|^2 \right) \prod_{a \in X_2} \upd x^a \lesssim e^{-c_1 \delta_0^{-2\delta_1}} \int_{|x|_2 \leq \delta_0^{\delta_1} } |F_2|^2 \prod_{a\in X_2} \upd x^a 
  \end{equation}
 \end{proposition}

\subsubsection{Estimating the error terms}

We now wish to plug our choice of initial data into the equation for $TT\phi$ (\eqref{equation TTphi}) and obtain bounds on the size of this term, which appears in both the conserved energy and the nondegenerate energy of the wave $T\phi$. Specifically, our choice of the initial data is the following:
\begin{equation*}
 \begin{split}
  \phi_0 \big|_{\Sigma_t} &= \mathfrak{Re} \left( e^{i\delta_0^{-1} y} \right) F_2(x) \chi\left( \delta_0^{-\frac{3}{4}} |y| \right) \chi\left( \delta_0^{-\frac{2}{5}} |x|_1 \right) \chi\left( \delta_0^{-\frac{1}{2} + \delta_1} |x|_2 \right) \\
  T\phi_0 \big|_{\Sigma_t} &= \mathfrak{Re} \left( -i \omega e^{i\delta_0^{-1} y} \right) F_2(x) \chi\left( \delta_0^{-\frac{3}{4}} |y| \right) \chi\left( \delta_0^{-\frac{2}{5}} |x|_1 \right) \chi\left( \delta_0^{-\frac{1}{2} + \delta_1} |x|_2 \right)
 \end{split}
\end{equation*}
where $F_2$, $\omega$ and $\delta_1$ are as above.

For the sake of brevity we define
\begin{equation*}
 \begin{split}
  \chi_0 &:= \chi\left( \delta_0^{-\frac{3}{4}} |y| \right) \\
  \chi_1 &:= \chi\left( \delta_0^{-\frac{2}{5}}|x|_1 \right) \\
  \chi_2 &:= \chi\left( \delta_0^{-\frac{1}{2} + \delta_1} |x|_2 \right)
 \end{split}
\end{equation*}
Similarly, in order to easily keep track of the scaling of each quantity with respect to $\delta_0$, we define
\begin{equation*}
 \begin{split}
  \chi_0' &:= \chi'\left( \delta_0^{-\frac{3}{4}} |y| \right) \\
  \chi_0'' &:= \chi''\left( \delta_0^{-\frac{3}{4}} |y| \right)
 \end{split}
\end{equation*}
and similarly for $\chi_1$ and $\chi_2$. Note that we have, for example, 
\begin{equation*}
\partial_y \chi(\delta_0^{-\frac{3}{4}} |y|) = \delta_0^{-\frac{3}{4}} \frac{y}{|y|} \chi_0'
\end{equation*}

Also note that we have
\begin{equation*}
\begin{split}
	(|b|^2 - a) &= M_{ab} x^a x^b + \mathcal{O}(|x|^3) + \mathcal{O}(|y||\tilde{x}|) \\
	&= \mathcal{O}(|x|_2^2) + \mathcal{O}(|x|^3) + \mathcal{O}(|y||\tilde{x}|) \\
\end{split}
\end{equation*}

Using equation \eqref{equation TTphi} we now calculate
\begin{equation*}
  TT\phi = \mathfrak{Re}\left( e^{i\delta_0^{-1}y} \right)\left( 2A \omega \delta_0^{-1} F_2 + A^2 \Delta F_2 - M_{ab} x^a x^b \delta_0^{-2} F_2 \right) \chi_0 \chi_1 \chi_2 + \text{Err}
\end{equation*}
where the error term is given by
\begin{equation*}
\begin{split}
	\text{Err} &= \mathfrak{Re} (\text{Err}_1) \\
	\text{Err}_1 &=
	- \delta_0^{-\frac{3}{4}} 2Ai\omega e^{i\delta_0^{-1} y} \frac{y}{|y|} F_2 \chi_0' \chi_1 \chi_2
	+ \delta_0^{-\frac{4}{5}} e^{i\delta_0^{-1} y} |X_1| F_2 \chi_0 \chi_1'' \chi_2
	\\
	&\phantom{=}
	+ A^2 e^{i\delta_0^{-1} y} \chi_0 \chi_1 \left( -2C_0\delta_0^{-\frac{1}{2} + \delta_1} \chi_2' \frac{x^a}{|x|}\nabla_a F_2  + C_0^2 \delta_0^{-1 + 2\delta_1} F_2 \chi_2'' \right)
	\\
	&\phantom{=}
	- \left( (|b|^2 - a) - M_{ab} x^a x^b \right) \delta_0^{-2} e^{i\delta_0^{-1} y} F_2 \chi_0 \chi_1 \chi_2
	\\
	&\phantom{=}
	+ (|b|^2 - a) \left( i\delta_0^{-\frac{7}{4}} e^{i\delta_0^{-1}y} \frac{y}{|y|} F_2 \chi_0' \chi_1 \chi_2
		+ \delta_0^{-\frac{3}{2}} e^{i\delta_0^{-1}y} F_2 \chi_0'' \chi_1 \chi_2 \right)
	+ \mathcal{O}(|\tilde{x}||\partial u|)
	\\
	&\phantom{=}
	+ \mathcal{O}(|\tilde{x}||\partial T \phi|)
	+ \mathcal{O}(|\tilde{x}|^2 |\bar{\partial}\partial \phi|)
	+ \mathcal{O}(|\tilde{x}|^3 |\partial^2 \phi|)
\end{split}
\end{equation*}

These first two error terms are easy to estimate: we have
\begin{equation*}
\begin{split}
	\left| - \delta_0^{-\frac{3}{4}} 2Ai\omega e^{i\delta_0^{-1} y} \frac{y}{|y|} F_2 \chi_0' \chi_1 \chi_2 \right|
	&\lesssim \delta_0^{-\frac{3}{4}}
	\\
	\left| \delta_0^{-\frac{4}{5}} e^{i\delta_0^{-1} y} |X_1| F_2 \chi_0 \chi_1'' \chi_2 \right|
	&\lesssim \delta_0^{-\frac{4}{5}}
\end{split}
\end{equation*}
Next, due to the support of $\chi_1 \chi_2 \chi_3$ we have
\begin{equation*}
	\left| \left( (|b|^2 - a) - M_{ab} x^a x^b \right) \delta_0^{-2} e^{i\delta_0^{-1} y} F_2 \chi_0 \chi_1 \chi_2 \right|
	\lesssim \delta_0^{-\frac{17}{20}}
\end{equation*}
Similarly, we have
\begin{equation*}
	\left| (|b|^2 - a) \left( i\delta_0^{-\frac{7}{4}} e^{i\delta_0^{-1}y} \frac{y}{|y|} F_2 \chi_0' \chi_1 \chi_2
		+ \delta_0^{-\frac{3}{2}} e^{i\delta_0^{-1}y} F_2 \chi_0'' \chi_1 \chi_2 \right) \right|
	\lesssim \delta_0^{-\frac{3}{4} + 2\delta_1}
\end{equation*}
and also
\begin{equation*}
	\left| \mathcal{O}(|\tilde{x}||\partial \phi|)
	+ \mathcal{O}(|\tilde{x}||\partial T \phi|)
	+ \mathcal{O}(|\tilde{x}|^2 |\bar{\partial}\partial \phi|)
	+ \mathcal{O}(|\tilde{x}|^3 |\partial^2 \phi|) \right|
	\lesssim
	\delta_0^{-\frac{4}{5}}
\end{equation*}

The only terms remaining are those involving $\chi_2'$ and $\chi_2''$. Specifically, we must bound the terms
\begin{equation*}
	A^2 e^{i\delta_0^{-1} y} \chi_0 \chi_1 \left( -2C_0\delta_0^{-\frac{1}{2} + \delta_1} \chi_2' \frac{x^a}{|x|}\nabla_a F_2  + C_0^2 \delta_0^{-1 + 2\delta_1} F_2 \chi_2'' \right)
\end{equation*}
A na\"ive approach to bounding these terms (in a similar manner to the bounds above) suggests that the first term behaves like $\delta^{-1 + \delta_1}$ and the second term like $\delta^{-1 + 2\delta_1}$, but these bounds are insufficient for our purposes. Instead, we use the Agmon esimtates of the previous section. Since $\chi_2'$ and $\chi_2''$ are both supported only in the region $\frac{1}{2} \delta_0^{-\frac{1}{2} + \delta_0} $, proposition \ref{proposition Agmon bounds} gives
\begin{equation*}
	\bigg|\bigg| A^2 e^{i\delta_0^{-1} y} \chi_0 \chi_1 \left( -2C_0\delta_0^{-\frac{1}{2} + \delta_1} \chi_2' \frac{x^a}{|x|}\nabla_a F_2  + C_0^2 \delta_0^{-1 + 2\delta_1} F_2 \chi_2'' \right) \bigg|\bigg|_{L^2}
	\lesssim
	\delta_0^{-1+2\delta_1} e^{-c_1 \delta_0^{-2\delta_1}}
\end{equation*}
so this actually \emph{decays} as $\delta_0 \rightarrow 0$.

Putting together all of these calculations, we find that
\begin{equation*}
	||\text{Err}||_{L^2} \lesssim \delta_0^{-\frac{17}{20}}
\end{equation*}

Finally, we rescale the data, so that it satisfies $\mathcal{E}^{(T)}[T\phi] = 1$. This means multiplying by a constant that scales as $\delta_0^{\frac{17}{20}}$ By the calculations above, along with the expressions for the $N$-energy \eqref{equation N energy} and the $T$-energy \eqref{equation T energy} we have found data such that
\begin{equation*}
\begin{split}
	\mathcal{E}^{(T)}[T\phi] &= 1 \\
	\mathcal{E}^{(N)}[T\phi] & \gtrsim \delta_0^{-\frac{3}{20}}
\end{split}
\end{equation*}
Now by choosing $\delta_0$ sufficiently small we can prove claim \ref{claim initial data 1}.

The construction in the asymptotically Kaluza-Klein case follows along almost identical lines, beginning from equation \eqref{equation wave local coords KK} instead of \eqref{equation wave local coords AF} and constructing data which is symmetric in the Kaluza-Klein directions.

\subsection{Non-decay of energy and an Aretakis-type instability}

Now that we have proved proposition \ref{proposition non decay on ergosurface}, we know that \emph{if} assumption \ref{assumption boundedness} holds, then there is some constant $\mathring{C} > 0$ and some solution $\phi$ to the wave equation such that 
\begin{equation*}
	\mathcal{E}^{(N)}_{\mathcal{U}} [T\phi] (t) + \mathcal{E}^{(N)}_{\mathcal{U}} [T^2 \phi] (t) \geq \mathring{C}
\end{equation*}
for all times $t$, and for any $T$-invariant open set $\mathcal{U}$ such that $\mathcal{S} \subset \mathcal{U}$. On the other hand, proposition \ref{proposition decay away from extended ergosurface} together with equation \ref{equation decay on interior} show that, if $\delta > 0$ and $\mathcal{U}$ is such that $S \cap \mathcal{U}_{(\delta)} = \emptyset$, then for \emph{all} $\epsilon > 0$ there is some time $\tau_\epsilon$ such that
\begin{equation*}
	\mathcal{E}^{(N)}_{\mathcal{U}} [T\phi] (\tau_\epsilon) + \mathcal{E}^{(N)}_{\mathcal{U}} [T^2 \phi] (\tau_\epsilon) \leq \epsilon
\end{equation*}

From this, it follows that, if we take any precompact, open set $\mathcal{U}_0 \subset \Sigma_0$ such that $S \subset \mathcal{U}_0$, and if $\mathcal{S}_{(\epsilon)}$ is the $\delta$-thickening of $\mathcal{S}$ defined such that
\begin{equation*}
	\text{Volume}(\mathcal{S}_{\epsilon,0}) = \epsilon
\end{equation*}
for some constant $\epsilon > 0$, and where the volume is defined with respect to the induced Riemannian metric on $\Sigma_0$. Then, there is some time $\tau_\epsilon$ such that
\begin{equation*}
\begin{split}
	\mathcal{E}^{(N)}_{\mathcal{U}_0} [T\phi] (\tau_\epsilon) + \mathcal{E}^{(N)}_{\mathcal{U}_0} [T^2 \phi] (\tau_\epsilon) \geq \mathring{C} \\
	\mathcal{E}^{(N)}_{\mathcal{U} \setminus \mathcal{S}_\epsilon} [T\phi] (\tau_\epsilon) + \mathcal{E}^{(N)}_{\mathcal{U} \setminus \mathcal{S}_\epsilon} [T^2 \phi] (\tau_\epsilon) \leq \epsilon \\
\end{split}
\end{equation*}
from which it follows that, if we choose $\epsilon$ sufficiently small
\begin{equation*}
	\mathcal{E}^{(N)}_{\mathcal{S}_\epsilon} [T\phi] (\tau_\epsilon) + \mathcal{E}^{(N)}_{\mathcal{S}_\epsilon} [T^2 \phi] (\tau_\epsilon) \geq \frac{1}{2}\mathring{C}
\end{equation*}
This is a kind of ``energy concentration'' phenomenon. It is easy to see from this result that the wave $\phi$ must blow up pointwise: indeed, we must have
\begin{equation*}
	\sup_{\mathcal{S}_\epsilon \cap \Sigma_{\tau_\epsilon}} \left( |\partial T \phi|^2 + |\partial T^2 \phi|^2 \right) \geq \frac{\mathring{C}}{2\epsilon}
\end{equation*}
since $\epsilon$ can be taken arbitrarily small, this establishes pointwise blow-up (without a rate).

We refer to this as an ``Aretakis-type'' instability because its similarity to the instability found in \cite{Aretakis2011, Aretakis2012a}. In particular, in both cases there is some quantity which is ``conserved'' on one hypersurface, but which decays everywhere else, and this is the cause of some kind of pointwise blowup.

\section{Spacetimes with additional symmetry}
\label{section additional symmetry}


Now we turn to spacetimes with an additional symmetry. Along with the symmetry generated by the Killing vector field $T$, in this section we will assume the existence of another Killing vector field $\Phi$, such that the span of $\{T, \Phi\}$ is \emph{timelike} in a neighbourhood of the ergosurface $S$. For simplicity, we shall assume that the Killing vector field $\Phi$ is an axial Killing field, i.e.\ that its integral curves are closed and spacelike. Moreover, we will assume that $T$ and $\Phi$ commute. Hence, there is a Killing vector field $\tilde{T}$ such that
\begin{itemize}
	\item $\tilde{T} = \alpha T + \beta \Phi$ for some constants $\alpha$ and $\beta$
	\item $\tilde{T}$ is timelike and future directed in a neighbourhood of $S$
\end{itemize}
It turns out that, given the presence of an additional symmetry of this sort, we can give many additional details regarding the instability discussed above. First, we find that we do not require the discrete isometry of section\ref{section the discrete isometry}. In addition, we can show that the energy of waves \emph{is} bounded, not in terms of the initial energy, but in terms of a ``higher order'' energy quantity. Furthermore, we can rule out case \ref{case (B)}, showing that the Aretakis-type instability cannot occur, and instead we will always encounter the unbounded local energy amplification of case \ref{case (A)}. Additionally, we can provide an upper bound on the time for this energy amplification to occur, which may be important for physical applications. We can also provide an example of unit-energy initial data which is \emph{not} compactly supported, but which gives rise to a solution of the wave equation with unbounded local energy. In other words, rather than a family of solutions, each of which exhibits energy amplification by a larger and larger factor, we can provide a \emph{single} solution of the wave equation, for which the local energy tends to infinity along a certain sequence of times. Finally, the additional symmetry allows us to deal with the issue of \emph{higher derivatives}. Consider the situation in which we know that the initial \emph{higher order} energy is small: then, in the case where this extra symmetry is present, we can prove that this same higher order energy can be amplified by an arbitrarily large constant.

Finally, we can use our results to rule out the existence of a manifold with an evenescent ergosurface, an additional symmetry of the kind described above, and a \emph{globally timelike} Killing vector field. Note that this is a result in pure differential geometry - \emph{a priori} this has nothing to do with the wave equation. However, we can use the properties of solutions of the wave equation to prove that such a manifold cannot exist.

\subsection{``Time reversal'' without a discrete isometry}

First, we describe how to deal with the issue of ``time reversal'' when we lack the discrete isometry of section \ref{section the discrete isometry}.

Let $\Sigma_0$ be some spacelike Cauchy surface for $\mathcal{M}$. Then, using the fact that $T$ (or $V$) is causal and transverse to $\Sigma_0$ we can construct a foliation of $\mathcal{M}$ by leaves $\Sigma_t$, where these leaves are the level sets of the function $t$, defined by
\begin{equation*}
	\begin{split}
		t\big|_{\Sigma_0} &= 0
		\\
		T(t) = 1
	\end{split}
\end{equation*}
As usual, we replace $T$ with $V$ in the asymptotically Kaluza-Klein case.

Now, we construct the ``time reversal'' operator as follows:
\begin{equation*}
	\begin{split}
		\mathscr{T} : \mathcal{M} &\rightarrow \mathcal{M} \\
		p \in \Sigma_t &\mapsto q \in \Sigma_{-t} \text{ \, where } q \text{ is such that an integral curve of } T \text{ passes through } p \text{ and } q
	\end{split}
\end{equation*}

Note that $\mathscr{T}$ is \emph{not} necessarily an isometry. Nevertheless, the ``time reversed'' manifold $(\mathcal{M}, \mathscr{T}^*(g))$ will possess the same important properties as the manifold $(\mathcal{M}, g)$: it will have an evanescent ergosurface and an additional symmetry of the correct kind.

\subsection{Boundedness with a loss of derivatives}

Now, we can show that this additional symmetry leads to \emph{energy boundedness with a loss of derivatives}. Specifically, we can prove a statement of the form \ref{assumption boundedness}, but where, on the right hand side, instead of the $N$-energy we see a ``higher order'' $N$-energy. Note that this will also hold for the ``time reversed'' manifold constructed in the subsection above.

Recall that, in a neighbourhood of $\mathcal{S}$, the vector field $\tilde{T}$ is timelike. We can express $\tilde{T}$ in terms of the frame constructed in section \ref{section adapted coordinates} as
\begin{equation*}
	\tilde{T} = \tilde{T}^T T + \tilde{T}^Y \hat{Y} + \tilde{T}^a e_a + \tilde{T}^A e_A
\end{equation*}
where the last term is absent in the case of an asymptotically flat manifold. Then, since $\tilde{T}$ is timelike, on the ergosurface $\mathcal{S}$ we have
\begin{equation*}
	2\tilde{T}^T \tilde{T}^Y \frac{g(T, n)}{\sqrt{g(Y,Y)}} + (\tilde{T}^Y)^2 + \sum_a (\tilde{T}^a)^2 + \sum_A (\tilde{T}^A)^2 < 0
\end{equation*}
Since $\tilde{T}$ is future-directed, it follows that $\tilde{T}^T > 0$. In turn, this implies that $\tilde{T}^Y > 0$ (note that $g(T, n) < 0$) everywhere on $S$. Since $S$ is a compact submanifold, it follows that there is some constant $C_{\tilde{T}} > 0$ such that $\tilde{T}^Y \geq C_{\tilde{T}}$ everywhere on $S$.

Now, if we apply the $T$-energy estimate to the field $\tilde{T}\phi$, we find that
\begin{equation*}
	\mathcal{E}^{(T)}[\tilde{T}\phi](\tau) = \mathcal{E}^{(T)}[\tilde{T}\phi](\tau_0)
\end{equation*}
for all $\tau \geq \tau_0$. In particular, from equation \eqref{equation T energy}, we see that, near $S$, we have
\begin{equation*}
	\mathcal{E}^{(T)}[\tilde{T}\phi](\tau) \sim \int_{\Sigma_\tau} \left( (T \tilde{T}\phi)^2 + \sum_a (\partial_{x^a} \tilde{T} \phi)^2 + \mathcal{O}(|\tilde{x}|^2)(\partial \tilde{T}u) \right) \upd y \, \upd x^1 \ldots \upd x^{D-1}
\end{equation*}
Now, using a Hardy inequality (see, for example, \cite{Moschidis2016}) we find that we have a bound of the form
\begin{equation*}
	\int_{\Sigma_\tau} h(x) (\tilde{T}\phi)^2 \upd y \, \upd x^1 \ldots \upd x^{D-1}
	\lesssim
	\mathcal{E}^{(T)}[\tilde{T}\phi](\tau)
\end{equation*}
where $h(x)$ is some positive, $T$-invariant (and $\mathcal{G}$-invariant, in the Kaluza-Klein case) function which tends to zero in the asymptotic region as $r \rightarrow \infty$, but which is otherwise bounded away from zero\footnote{In fact, we can choose $h(x) \sim r^{-2}$.}. In particular, $h(x)$ is bounded away from zero in a neighbourhood of $\mathcal{S}$.

Combining this bound with the calculation above, we see that, if $\mathcal{U}$ is a sufficiently small, $T$-invariant neighbourhood of $\mathcal{S}$, then we have
\begin{equation*}
\begin{split}
	&\int_{\Sigma_\tau} \left(
		(\hat{Y}\phi)^2
		+ (T \phi)^2
		+ \sum_a (\partial_{x^a} \phi)^2
		+ \sum_A (\partial_{x^A} \phi)^2
	\right) \upd y \, \upd x^1 \ldots \upd x^{D-1}
	\\
	&\lesssim \mathcal{E}^{(T)}[\phi](\tau) + \frac{1}{C_{\tilde{T}}^2} \mathcal{E}^{(T)}[\tilde{T}\phi](\tau)
\end{split}
\end{equation*}
which in turn yields the global bound
\begin{equation}
\label{equation boundedness with loss of derivatives}
\begin{split}
	\mathcal{E}^{(N)}[\phi](\tau)
	&\lesssim
	\mathcal{E}^{(T)}[\phi](\tau) + \frac{1}{C_{\tilde{T}}^2} \mathcal{E}^{(T)}[\tilde{T}\phi](\tau)
	\\
	&\lesssim
	\mathcal{E}^{(T)}[\phi](0) + \frac{1}{C_{\tilde{T}}^2} \mathcal{E}^{(T)}[\tilde{T}\phi](0)
\end{split}
\end{equation}

We refer to the bound in equation \eqref{equation boundedness with loss of derivatives} as a \emph{uniform boundedness estimate with a loss of derivatives}. It provides ``uniform boundedness'' of the \emph{non-degenerate} $N$-energy in terms of conserved quantities, which can therefore be evaluated in terms of the initial data. On the other hand, the quantity that we bound (namely the $N$-energy) involves only \emph{first} derivatives of the solution, and yet in order to control it we find that we need information regarding the \emph{second} derivatives of the initial data. We refer to this as a \emph{loss of derivatives}. We will show, below, that it is actually necessary to lose derivatives in this way: there is no such uniform boundedness estimate in terms of the first derivatives alone.

\subsection{Ruling out case \ref{case (B)}}

In the general case considered in section \ref{section general case}, we were led to a dichotomy: we either had amplification of the local energy by an arbitrarily large factor (case \ref{case (A)}), or else we had an Aretakis-type instability (case \ref{case (B)}). We can now show that, in the situation of enhanced symmetry now under consideration, this latter case cannot occur, and so we must have the kind of behaviour considered in case \ref{case (A)}.

Recall that the behaviour of case \ref{case (B)} can occur \emph{only if} the local energy \emph{does not decay} towards the past, that is, if there is some smooth solution $\phi$ to $\Box_g \phi = 0$ such that $\phi$ arises from compactly supported initial data satisfying\begin{equation*}
	\sum_{j = 0}^{3} \mathcal{E}^{(N)}[T^j \phi](0) < \infty
\end{equation*}
and some positive constant $\mathring{C}$ such that, for all times $\tau < 0$, we have
\begin{equation}
\label{equation non decay for contradiction}
	\mathcal{E}^{(N)}_{\mathcal{U}}[T\phi](\tau) + \mathcal{E}^{(N)}_{\mathcal{U}}[T^2 \phi](\tau) \geq \mathring{C}
\end{equation}
We now show that, in the case of enhanced symmetry, such a solution cannot exist.


Recall that the argument of \cite{Moschidis2015} establishes logarithmic decay of the $N$-energy on manifolds with the same asymptotic structure as those we are considering\footnote{With the required modifications for $\mathcal{G}$-invariant data on asymptotically Kaluza-Klein manifolds as discussed in section \ref{section estimates in KK spaces}.}, under the assumptions that:
\begin{enumerate}
	\item A \emph{uniform boundedness} statement holds. Note that we now have such a uniform boundedness statement, albeit in terms of a higher order energy of the initial data.
	\item Either the spacetime has a horizon, in which case it is allowed to have a suitably ``small'' ergoregion (in a way made precise in \cite{Moschidis2015}), or the asymptotically timelike Killing field $T$ is globally and uniformly timelike.
\end{enumerate}

Note that the spacetimes we are considering do not obey this second condition, since the asymptotically timelike Killing field is null on the evanescent ergosurface. Nevertheless, we can modify the arguments of \cite{Moschidis2015} to show that we still have some decay of solutions to the wave equation.

Most of the estimates of \cite{Moschidis2015} still apply to the kinds of manifolds we are considering - in particular, all of the estimates in the ``asymptotic region'' still apply. However, the low frequency estimates of section 6 and the Carleman estimates in section 7 of \cite{Moschidis2015}, which are used to establish an ``integrated local energy decay'' estimate in the interior region, need to be significantly modified if they are to apply to the kinds of spacetimes we are considering, since it is in the proof of these estimates that the failure of the vector field $T$ to be globally timelike causes an issue.

Fortunately, for the spacetimes under consideration here, such a modification is possible, and we will sketch the details below. The main idea is that, when we have an additional axial Killing field of the kind we are assuming, then we can simultaneously decompose solutions to the wave equation into (time) frequency-localized and angular frequency-localized components. Then, the Carleman estimates performed in \cite{Moschidis2015} can be shown to apply to each of the angular frequency components separately, but with an additional degenerate factor that degenerates at high angular frequency. Finally, we can perform a double interpolation argument, first showing that the individual angular frequency components decay logarithmically, and then showing that the entire solution decays sub-logarithmically.

\subsubsection{The angular frequency decomposition}

We first note that, since $\Phi$ is a Killing vector field with closed orbits, and since $[T, \Phi] = 0$, we can define some coordinate $\varphi$ with period $2\pi$ such that (rescaling $\Phi$ if necessary) $\Phi(\varphi) = 1$, Moreover, the level sets of $\tau$ can be chosen such that the integral curves of $\Phi$ lie within the level sets of $\tau$. Then, we can decompose a solution to the wave equation $\phi$ in terms of ``axial modes''
\begin{equation*}
\begin{split}
	\phi &= \sum_{n = -\infty}^{\infty} \phi_{(n)}
	\\
	\Phi (\phi_{(n)}) &= i n\phi_{(n)}
\end{split}
\end{equation*}
In other words, each of the $\phi_{(n)}$ is of the form $\phi_{(n)} = e^{in\varphi} \tilde{\phi}_{(n)}$ where $\Phi(\tilde{\phi}_{(n)}) = 0$. Note that, by orthogonality, for each $n$ we have
\begin{equation*}
	\Box_g (\phi_{(n)}) = 0
\end{equation*}

The idea is that we will apply the logarithmic decay result of \cite{Moschidis2015} \emph{seperately} for each axial mode. When doing so, we will have to keep track of the dependence of various constants on the mode number $n$. The key estimates are the low frequency estimates of section 6 and the Carleman estimates of section 7 of \cite{Moschidis2015}; it is easy to see that the estimates in the asymptotic region apply to the solutions $\phi_{(n)}$. Note also that, since $[T, \Phi] = 0$, we can simultaneously perform the angular frequency decomposition and also decompose the solutions with respect to time frequencies.

Below we will sketch the required modifications to the arguments given in \cite{Moschidis2015}. Many of the arguments are almost identical, so we will only go into depth in those places which differ significantly from the argument presented in \cite{Moschidis2015}.

\subsubsection{Integrated local energy decay for very low frequencies}

The approach of \cite{Moschidis2015} involves dividing up the (time) frequency range into a low frequency part $\omega \lesssim \omega_0$, intermediate frequency parts $\omega \sim \omega_k$ and a high frequency part $\omega \gtrsim \omega_+$. In our case, this division will itself depend on the value of $n$, the angular quantum number.

First, we define some $\omega_0$, which is sufficiently small compared with various geometric quantities (as in \cite{Moschidis2015}) and which also satisfies
\begin{equation*}
	\omega_0 \lesssim \frac{1}{n}
\end{equation*}
The corresponding frequency localised wave will be denoted by $\uppsi_{0, n}$

Then, we can repeat the calculations of section 6 of \cite{Moschidis2015}, using the current
\begin{equation*}
	J_\mu
	:=
	h_R \uppsi_{0,n} \partial_\mu \uppsi_{0,n}
	- \frac{1}{2} \partial_\mu h_R \uppsi_{0,n}^2
\end{equation*}
The construction of the function $h_R$ is identical to that given in \cite{Moschidis2015}. We can see (using the expression for the inverse metric given in \eqref{equation expression for inverse metric} together with the fact that $\{T, \Phi\}$ span a timelike direction, so that $\Phi$ must have a component in the $\partial_y$ direction) that
\begin{equation*}
	h_R \partial^\mu \uppsi_{0,n} \partial_\mu \uppsi_{0,n}
	\geq
	c \cdot h_R |\nabla_{\Sigma, \text{degen}} \uppsi_{0,n}|^2
	- C \cdot h_R |T \uppsi_{0,n}| \left( |T \uppsi_{0,n}| + |\Phi \uppsi_{0,n}| \right)
\end{equation*}
where $c$ and $C$ are some numerical constants, and $|\nabla_{\Sigma, \text{degen}} \uppsi_{0,n}|^2$ includes all of the spatial derivatives of $\uppsi_{0,n}$ \emph{except} for the $\Phi$ derivative on the surface $\mathcal{S}$. More precisely, the coefficient of $|\Phi \uppsi_{0,n}|^2$ degenerates quadratically at $\mathcal{S}$.

The remainder of the calculations proceed in the same way as those in section 6 of \cite{Moschidis2015}, except that there is no black hole horizon, and so many of the calculations are easier (see the footnotes in \cite{Moschidis2015}). Following these calculations we can show that, in the notation of \cite{Moschidis2015},
\begin{equation*}
\begin{split}
	\int_{ \{r\leq R\} \cap \mathcal{R}(0, t^*) } \left(
		|\nabla_{\Sigma, \text{degen}} \uppsi_{0,n}|^2
		+ |\uppsi_{0,n}|^2
	\right)
	&\lesssim
	(1 + \omega_0^2 + n^2)\mathcal{E}^{(T)}[\psi_{0,n}](0)
	\\
	&\phantom{\lesssim}
	+ \int_{ \{r\leq R\} \cap \mathcal{R}(0, t^*) } \left(
		(\omega_0) (\omega_0 + n) |\uppsi_{0,n}|^2
	\right)
\end{split}
\end{equation*}
Hence, if $\omega_0$ is sufficiently small relative to the geometry \emph{and} to $\frac{1}{n}$, then we can absorb the second term on the right hand side by the left hand side.

This proves a \emph{degenerate} integrated local energy decay statement, since we still do not control all of the derivatives of $\uppsi_{0,n}$ -- we are missing the $\Phi$ derivatives on the surface $\mathcal{S}$. We can fix this by commuting once with $\Phi$. We obtain, in the end,
\begin{equation*}
	\int_{ \{r\leq R\} \cap \mathcal{R}(0, t^*) } \left(
		|\partial \uppsi_{0,n}|^2
		+ |\uppsi_{0,n}|^2
	\right)
	\lesssim
	(1 + \omega_0^2 + n^2)(1+n^2)\mathcal{E}^{(T)}[\psi_{0,n}](0)
\end{equation*}

\subsubsection{Integrated local energy decay for intermediate frequencies}

For the intermediate frequencies $\omega_k$ we can also follow the calculations of section 7 of \cite{Moschidis2015}. The calculations in this section make use of the fact that the frequencies are bounded away from zero. Hence, when repeating these calculations, we must keep in mind the fact that, in our case, $\omega_0 \sim n^{-1}$, and so we must track the dependence of various constants on the value of $\omega_0$.

As in \cite{Moschidis2015}, we extend the function $r$ from the asymptotic region (where it is the pullback of the spherical polar radial coordinate on $\mathbb{R}^d$ or $\mathbb{R}^d \times X$) to the entire hypersurface $\Sigma_t$ by requiring that $r$ is a Morse function. Moreover, we can arrange that, say, $r = r_0$ on $\mathcal{S}$, and $\upd r \neq 0$ on $\mathcal{S}$. Furthermore, we can set $\Phi(r) = 0$ in a neighbourhood of $\mathcal{S}$. Finally, we can also arrange that in the region $r \leq 2r_0$ the vector field $\tilde{T}$ is uniformly timelike.

The construction of the two Morse functions $\omega$, $\omega'$ required for the energy currents in section 7 of \cite{Moschidis2015} proceeds exactly as in \cite{Moschidis2015} - note that we only construct these functions away from $\mathcal{S}$, in the region where $T$ is uniformly timelike. Note that, for $r_0 \leq r \leq 2r_0$, we have $\omega = \omega' = r$.

Note that, as in \cite{Moschidis2015}, $\partial^\mu \omega \partial_\mu \omega \neq 0$ away from the critical points of $\omega$. In the region $r \geq 2r_0$ this follows directly from the arguments of \cite{Moschidis2015}, using the fact that $T$ is timelike in this region. On the other hand, for $r_0 \leq r \leq 2r_0$, both the vector fields $T$ and $\Phi$ are tangent to the level sets of $\omega$, and so in particular the timelike vector field $\tilde{T}$ is tangent to the level sets of $\omega$.

The remainder of the caluclations in section 7 of \cite{Moschidis2015} proceed in almost identical fashion, with the important exception of inequality $(7.22)$. Here, we instead have
\begin{equation*}
	|\partial \uppsi_{k,n}|_h^2
	\leq
	C_1 \partial_\mu \uppsi_{k,n} \partial^\mu \bar{\uppsi}_{k,n}
	+ C_2 |T \uppsi_{k,n}|^2
	+ C_3 |\Phi \uppsi_{k,n}|^2
\end{equation*}
where the third term is new. Effectively, this means that the constants $\omega_k^2$ that appear in the inequalities in section 7 of \cite{Moschidis2015} need to be replaced with $\omega_k^2 + n^2$ in our calculation. In particular, the parameter $s$ needs to be chosen sufficiently large compared to $\omega_k + |n|$ rather than just $\omega_k$.

With this in mind, and noting that the various constants that depend on $\omega_0$ are, at worst, of order $(\omega_0)^{-2}$ (see the comments in Lemma 4.6), we see that the key integrated local energy decay estimate, proposition 7.2 in \cite{Moschidis2015}, is replaced by
\begin{proposition}[ILED for bounded frequencies (proposition 7.2) in \cite{Moschidis2015}]
	For any $R \geq r_0$ and for $\omega_0$ sufficiently small compared with both $1$ and $n^{-1}$, there exists some positive constant $C(R)$ such that, for any smooth axial mode $\uppsi_{k}$ with compactly supported initial data, and any $\omega_+ > 1$ and $1 \leq |k| \leq n$, we have
	\begin{equation*}
		\int_{ \{r\leq R\} \cap \mathcal{R}(0, t^*) } \left(
			|\partial \uppsi_{\leq \omega_+,n}|^2
			+ |\uppsi_{\leq \omega_+,n}|^2
		\right)
		\lesssim
		C(R) e^{C(R) n^2 (\omega_+ + |n|)}\mathcal{E}^{(T)}[\psi_{0,n}](0)
	\end{equation*}
\end{proposition}

In other words, we can repeat all of the calculations of section 7 of \cite{Moschidis2015} for the axial mode $\psi_n$, at the expense of making our estimates degenerate exponentially in $n$.

\subsubsection{A double interpolation argument and sub-logarithmic decay}

Performing the interpolation argument as in \cite{Moschidis2015}, and remembering that we must ``lose derivatives'' in the boundedness statement, we find that the local energy of the $n$-th axial mode decays logarithmically as
\begin{equation*}
\begin{split}
	&\mathcal{E}^{(N)}_{\mathcal{U}}[\phi_{(n)}](\tau)
	\\
	&\leq
	C_{(\mathcal{U}, m)} \left(
		\frac{n^{4m}}{\left(\log(2+\tau)\right)^{2m}}
		+ \frac{e^{C_{(\mathcal{U}, m)} |n|^3}}{(1+\tau)^{\delta_5}}
	\right)
	\left(
		\int_{\Sigma_0} (r^{\delta_4} )\imath_{ {^{(N)}J[\phi]}} \dVol 
		+ \sum_{j+k \leq m+1} \mathcal{E}^{(N)}[T^j \Phi^k \phi] (0)
	\right)
\end{split}
\end{equation*}
for any positive integer $m$ and any real number $\delta_4 > 0$, and for some small $\delta_5 > 0$, and where the $C_{(\mathcal{U}, m)}$ are some (possibly very large) constants depending only on $m$ and the set $\mathcal{U}$. This is a slightly modified version of corollary 2.2 of \cite{Moschidis2015}. Note that the first term defines a kind of ``weighted'' energy - see the discussion of $p$-weighted energy estimates in section \ref{section estimates in KK spaces}.

Note that, although this estimate shows that the $n$-th axial mode decays logarithmically, we cannot simply ``add up'' all such estimates to show that the solution as a whole decays logarithmically. In fact, these estimates degenerate (exponentially) in $n$ at large values of $n$ (Note that the polynomial degeneration in $n$ of the bounds given for very time frequencies is strictly better than this). To obtain decay for the solution $\phi$, rather than just the axial modes $\phi_{(n)}$, we can use the interpolation argument again, using the decay statement for axial modes with $|n| \leq n_+$, and simply using the boundedness statement (and commuting with $\Phi$) for modes with $|n| > n_+$. Choosing $n_+ \sim (\log (2+\tau))^{\frac{1}{3}}$, we find
\begin{equation}
\label{equation decay with extra symmetry 1}
\begin{split}
	&\mathcal{E}^{(N)}_{\mathcal{U}}[\phi](\tau)
	\leq
	\frac{ C_{(\mathcal{U}, m)} }{ (\log (2+\tau))^{\frac{2}{3}m} }
	\left(
		\int_{\Sigma_0} (r^{\delta_4} )\imath_{ {^{(N)}J[\phi]}} \dVol 
		+ \sum_{j+k \leq m+2} \mathcal{E}^{(N)}[T^j \Phi^k \phi] (0)
	\right)
\end{split}
\end{equation}
where $\mathcal{U}_0 \subset \Sigma_0$ is \emph{any} precompact set (including those which intersect the ergosurface $\mathcal{S}$). Note that this decay result holds in both the original manifold and the time-reversed manifold.

In short, the local $N$-energy of $\phi$ decays \emph{sub-logarithmically}. In particular, if $\phi$ arises from smooth, compactly supported initial data then
\begin{equation*}
\label{equation decay with extra symmetry 2}
\begin{split}
	\mathcal{E}^{(N)}_{\mathcal{U}}[T\phi](\tau) + \mathcal{E}^{(N)}_{\mathcal{U}}[T^2 \phi](\tau)
	&\leq
	\frac{C_{\mathcal{U}}}{\left(\log(2+\tau)\right)^{\frac{2}{3}}} \bigg(
		\int_{\Sigma_0} (r^{\delta_4} )\imath_{ {^{(N)}J[T\phi]}} \dVol 
		+ \int_{\Sigma_0} (r^{\delta_4} )\imath_{ {^{(N)}J[T^2 \phi]}} \dVol
		\\
		&\phantom{\lesssim \frac{1}{\left(\log(2+\tau)\right)^{\frac{2}{3}}} \bigg(}
		+ \sum_{j+k \leq 4} \mathcal{E}^{(N)}[T^j \Phi^k \phi] (0)
	\bigg)
	\\
	&\leq
	\frac{C_{(\mathcal{U} , \phi)}}{\left(\log(2+\tau)\right)^{\frac{4}{3}}} 
\end{split}
\end{equation*}
where the first line follows from equation \eqref{equation decay with extra symmetry 1} and the second line follows from the fact that the data is smooth and compactly supported, so that the energy quantities (including the weighted energies) are finite initially. Note that the numerical constants $C_{\mathcal{U}}$ and $C_{(\mathcal{U}, \phi)}$ are generally different.

We call this decay ``sub-logarithmic'' because, in terms of pointwise decay rates, this would lead to decay for the fields $\phi$ at a rate $\phi \sim (log(2+\tau))^{-\frac{1}{3}}$. Note that this is a kind of converse to the lower bound proved on microstate geometries in \cite{Keir2016}, albeit this result shows decay at a slower rate. It is likely, therefore, that this does not represent a sharp decay rate for linear waves on these geometries.

In any case, this rules out the existence of a constant $\mathring{C} > 0$ such that equation \eqref{equation non decay for contradiction} holds, since, if $\tau$ is sufficiently large, then we will always have
\begin{equation*}
	\frac{C_{(\mathcal{U}, \phi)}}{\left(\log(2+\tau)\right)^{\frac{2}{3}}} \leq \mathring{C}
\end{equation*}
even if the numerical constant $C_{(\mathcal{U}, \phi)}$ is very large. In turn, this rules out case \ref{case (B)}, so that, when an extra symmetry of the right kind is present, then we must have case \ref{case (A)}, i.e.\ energy amplification by an arbitrarily large factor. Note also that, since the $T$-energy remains bounded, the energy amplification must occur near the evanescent ergosurface $\mathcal{S}$. In other words, we have proved the following:

\begin{lemma}[Local energy amplification in the presence of an additional symmetry]
\label{lemma energy amplification with symmetry}
	Let $(\mathcal{M}, g)$ be a Lorentzian manifold which is either asymptotically flat and has an evanescent ergosurface of the first kind (see condition \ref{condition asymp flat}), or which is asymptotically Kaluza-Klein and has an evanescent ergosurface of the second kind (see condition \ref{condition asymp KK}). Furthermore, suppose there is an additional Killing vector field $\Phi$, such that the span of $T$ and $\Phi$ includes a Killing vector field that is timelike in a neighbourhood of $\mathcal{S}$.
	
	Then, for any constant $C > 0$ and any open set $\mathcal{U}_0 \subset \Sigma_0$ such that $\mathcal{S} \cap \mathcal{U}_0 \neq \emptyset$, there is a solution $\phi_{(C,\, \mathcal{U})}$ to the wave equation $\Box_g \phi_{(C, \, \mathcal{U})} = 0$ arising from smooth, compactly supported initial data, and a time $\tau_{(C, \, \mathcal{U})}$ such that
	\begin{equation}
	\mathcal{E}^{(N)}_{\mathcal{U}}[\phi_{(C, \, \mathcal{U}})](\tau_{(C, \, \mathcal{U})}) \geq C \mathcal{E}^{(N)}[\phi_{(C, \, \mathcal{U})}](0)
	\end{equation}
	In other words, there is a solution to the wave equation whose local energy (in the set $\mathcal{U}$) is amplified, relative to its \emph{total} initial energy, by a factor of at least $C$.
\end{lemma}

Note that, by construction, the solution $\phi_{(C, \ \mathcal{U})}$ in lemma \ref{lemma energy amplification with symmetry} can be chosen to be the $T$ derivative of a solution to the wave equation.

\subsection{Bounds for the amplification time and the support of the data}

Physically, it is very important to be able to estimate the \emph{timescale} of any proposed instability. For example, if the timescale of an instability of some object is very small compared to the timescale on which those objects form, then we would not expect to find such objects in nature, whereas in the opposite case we might still expect to find these objects, despite the presence of an instability.

Now that we have shown that, in the presence of an additional symmetry, the local (non-degenerate) energy of a solution to the wave equation can grow arbitrarily large relative to its initial (non-degenerate) energy, it turns out that the extra symmetry also allows us to prove an upper bound for the time taken for the local energy to grow. At the same time, we can prove a bound on the size of the support of the initial data which leads to this growing solution, which might also have some physical relevance.

Recall that, in the case of additional symmetry, we can show both that the $N$-energy is bounded (by a higher order initial energy) and that the local $N$-energy decays at least logarithmically (again, this bound necessarily involves higher order initial energies). Thus, we can construct data for the wave equation as in subsection \ref{subsection constructing the initial data}. This leads to initial data for $\phi$ at a time $\tau_1$ such that
\begin{itemize}
	\item $\mathcal{E}^{(T)}[T\phi](\tau_1) = 1$
	\item $\mathcal{E}^{(N)}[T^2\phi](\tau_1) = \mathcal{O}(\delta_0^{-\frac{3}{20}})$
	\item $\mathcal{E}^{(N)}[\partial T\phi](\tau_1) = \mathcal{O}(\delta_0^{-2 - \frac{3}{20}})$
\end{itemize}
moreover, it is not very difficult to see that the initial data constructed in this way satisfies
\begin{equation*}
	\sum_{j+k \leq m} \mathcal{E}^{(N)}[T^j \Phi^k \phi](\tau_1) = \mathcal{O}(\delta_0^{-2m-\frac{3}{20}})
\end{equation*}
Thus, solving the wave equation \emph{backwards} in time, the decay estimate \eqref{equation decay with extra symmetry 1} tells us that the local non-degenerate energy of the wave $T\phi$ at time $0$ is bounded by
\begin{equation*}
	\mathcal{E}^{(N)}_{\mathcal{U}}[T\phi](0) \leq \frac{ C_{(\mathcal{U}, m)} \delta_0^{-2m - 6-\frac{3}{20}} }{ \left( \log(2+\tau_1) \right)^{\frac{2}{3}m}} 
\end{equation*}
for any choice of $m \in \mathbb{N}$. Thus, we can guarantee that the local energy at time $0$ is bounded above by $1$ by choosing
\begin{equation*}
	\tau_1 = \exp \left((C_{(\mathcal{U}, m)})^{\frac{3}{2m}} \delta_0^{-3-\frac{82}{5m}} \right)
\end{equation*}
Moreover, at this time the \emph{total} $N$-energy is bounded by some constant which is \emph{independent of} $\delta_0$, since the $N$-energy away from the ergosurface is bounded by the $T$-energy, which is conserved and takes the value $1$.

Putting this together\footnote{Choosing a larger value of $m$ appears to give an improved lower bound on the amplification bound - i.e.\ it leads to a bound whose functional dependence on $\delta_0$ is better. However, the numerical constant $C_{\mathcal{U}, m}$ also depends on the value of $m$ in some uncontrolled way, so we cannot simply pass to the limit $m \rightarrow \infty$.}, we have shown the following:

\begin{corollary}[Bounds for the amplification time and the support of the data]
\label{corollary bounds for the amplification time}
	Let $(\mathcal{M}, g)$ be as in lemma \ref{lemma energy amplification with symmetry}. 
	
	Then, for any constant $C > 0$ and any open set $\mathcal{U}_0 \subset \Sigma_0$ such that $\mathcal{S} \subset \mathcal{U}_0$, there is a solution $\phi_{(C,\, \mathcal{U})}$ to the wave equation $\Box_g \phi_{(C, \, \mathcal{U})} = 0$ arising from smooth, compactly supported initial data such that
	\begin{equation}
	\sup_{\tau \in [0, \tau^*]} \frac{ \mathcal{E}^{(N)}_{\mathcal{U}}[\phi_{(C, \, \mathcal{U}})](\tau) }{\mathcal{E}^{(N)}[\phi_{(C, \, \mathcal{U})}](0)} \geq C 
	\end{equation}
	where $\tau^*$ is given by
	\begin{equation}
		\tau^* = C_{(\mathcal{U}, \, m)} \exp \left( C^{20 + \frac{328}{3m}} \right)
	\end{equation}
	for some constant $C_{(\mathcal{U}, \, m)}$ depending only on the set $\mathcal{U}$ and the positive integer $m$.
	
	Moreover, the initial data for this solution is supported only in the intersection of the causal past of the set $\mathcal{U} \cap \Sigma_{\tau^*}$ with the initial hypersurface $\Sigma_0$.
\end{corollary}

Note that the solution $\phi_{(C, \, \mathcal{U})}$ appearing in the corollary is the function $T\phi$ from the calculations above. Recall that, since $T$ is a Killing vector field, if $\phi$ is a solution to the wave equation then so is $T\phi$.

\subsection{Solutions with unbounded local energy}

We have seen (in lemma \ref{lemma energy amplification with symmetry}) that, in the case where an extra symmetry is present, the local non-degenerate energy can be amplified by an arbitrarily large amount compared with the initial, total, non-degenerate energy. A natural question now arises: does there exist finite-energy initial data leading to a \emph{single} solution $\phi$ to the wave equation (as opposed to a sequence of solutions), such that the local energy of $\phi$ becomes arbitrarily large?

Using the extra symmetry, we give an affirmative answer to this question. Note, however, that the initial data we construct is not necessarily smooth, and indeed, in view of the non-degenerate energy bound with a loss of derivatives, it cannot have finite ``higher order'' energies. Moreover, it is not necessarily compactly supported either.

\begin{corollary}[A solution with unbounded local energy]
	Let $(\mathcal{M}, g)$ be as in lemma \ref{lemma energy amplification with symmetry}.
	
	Then there exists a (weak) solution $\phi$ to $\Box_g \phi = 0$ such that
	\begin{itemize}
		\item The initial $N$-energy of $\phi$ is one, that is, $\mathcal{E}^{(N)}[\phi](0) = 1$
		\item The local $N$-energy of $\phi$ is unbounded, that is,
		\begin{equation}
			\limsup_{\tau \rightarrow \infty} \ \mathcal{E}^{(N)}_{\mathcal{U}}[\phi](\tau) = \infty
		\end{equation}	
		where $\mathcal{U}$ is any $T$-invariant open set such that $\mathcal{S} \subset \mathcal{U}$.
	\end{itemize}
\end{corollary}

For ease of notation, let us fix an open set $\mathcal{U}$ intersecting the ergosurface. Let $(T\phi)_{n}$ be a solution to the wave equation constructed as in lemma \ref{lemma energy amplification with symmetry}, except that we continue evolving the solution to the past until the $N$-energy is bounded above\footnote{Recall that, when we solve the wave equation backwards and use the decay estimate \eqref{equation decay with extra symmetry 1}, we only obtain an \emph{upper bound} on the initial $N$-energy. We want to avoid scaling the solution up by a (potentially large) factor.} by $\frac{1}{(C_n)^3}$. Then we have the following lower bound on the local energy:
\begin{equation*}
	\mathcal{E}^{(N)}_{\mathcal{U}}[(T\phi)_{n}](\tau_{C_n}) \geq C_n
\end{equation*}
where $\tau_{C_n}$ is some time satisfying
\begin{equation*}
	\tau_{C_n} \sim  e^{(C_n)^{21}}
\end{equation*}

At the same time, for the solution $(T\phi)_{n}$ we have the decay estimate (see equation \eqref{equation decay with extra symmetry 1} with the choice $m = 1$)
\begin{equation*}
	\mathcal{E}^{(N)}_{\mathcal{U}}[(T\phi)_{n}](\tau)
	\lesssim 
	\frac{(C_n)^{41} }{ \left( \log(2+\tau) \right)^{\frac{2}{3}} }
\end{equation*}
which holds for \emph{all} $\tau$.

Finally note that, by applying the decay estimate to the solution as we evolve backwards in time, we find that, for all times $\tau$ such that
\begin{equation*}
	\tau \lesssim e^{C_n^{\frac{129}{2}}} 
\end{equation*}
we also have the bound
\begin{equation*}
	\mathcal{E}^{(N)}_{\mathcal{U}}[(T\phi)_{n}](\tau)
	\lesssim 
	\frac{1}{(C_n)^2}
\end{equation*}

Now, we can use these bounds to construct a solution with the desired properties. First, we choose the constants $C_n = 2^{2^n}$. We see that, at some time $\tau_n$ satisfying
\begin{equation*}
	\tau_n \sim e^{2^{21\cdot 2^n} }
\end{equation*}
we have
\begin{equation*}
	\mathcal{E}^{(N)}_{\mathcal{U}}[(T\phi)_{n}](\tau_n) \geq 2^{2^n}
\end{equation*}

Now, if $m \geq n$ and $n$ is large enough, then we have
\begin{equation*}
	e^{2^{\frac{129}{2} \cdot 2^m }} > e^{2^{21 \cdot 2^n }}
\end{equation*}
and so we have
\begin{equation*}
	\mathcal{E}^{(N)}_{\mathcal{U}}[(T\phi)_{m}](\tau_n) \lesssim \frac{1}{ 2^{2 \cdot 2^m}}
\end{equation*}

On the other hand, if $m \leq n$ then we can use the uniform decay estimate to show
\begin{equation*}
	\mathcal{E}^{(N)}_{\mathcal{U}}[(T\phi)_{m}](\tau_n)
	\lesssim 
	\frac{2^{41\cdot2^m} }{ 2^{21 \cdot 2^n} }
\end{equation*}

Again, if $n$ is sufficiently large, then this term is bounded, say by $\frac{1}{2^n}$.

Now, we set
\begin{equation*}
	\phi_{\infty} := \sum_{n = 1}^{\infty} (T\phi)_n
\end{equation*}
By the triangle inequality, and the calculations above, we see that there is a sequence of times $\tau_n$ such that
\begin{equation*}
	\mathcal{E}^{(N)}_{\mathcal{U}}[\phi_\infty](\tau_n) \rightarrow \infty
\end{equation*}
At the same time, the initial energy of this function $\phi_{\infty}$ is bounded by $\sum_{n = 0}^{\infty} \frac{1}{2^{3 \cdot 2^n}}$, which is clearly finite. Hence the series has a limit which is a weak solution of the wave equation. Note, however, that from the bound \eqref{equation boundedness with loss of derivatives}, the initial energy of $\tilde{T} \phi_{\infty}$ must be infinite.

\subsection{Higher derivatives}

Another natural question in the context of the work above is whether the instability we have discussed can be ``cured'' by looking at higher derivatives. We could compare the situation to the case of wave equations on a Schwarzschild black hole, where the well-known ``trapping'' phenomena means that an \emph{integrated local energy decay} statement cannot hold. In the language of this paper, this means that no statement of the form
\begin{equation*}
	\int_{0}^{\tau_1} \left( \int_{\Sigma_\tau\cap\mathcal{U}} \imath_{ {^{(N)}J[\phi]}} \dVol \right) \upd \tau \lesssim \mathcal{E}^{(N)}[\phi](0)
\end{equation*}
can hold (see \cite{Sbierski2013a} for a very general proof that this kind of statement cannot hold on spacetimes involving trapping). However, in the case of Schwarzschild black holes, it is possible to ``fix'' this problem by including higher derivatives on the right hand side, and indeed a statement of the following form can be shown to hold:
\begin{equation*}
	\int_{0}^{\tau_1} \left( \int_{\Sigma_\tau\cap\mathcal{U}} \imath_{ {^{(N)}J[\phi]}} \dVol \right) \upd \tau \lesssim \mathcal{E}^{(N)}[\phi](0) + \mathcal{E}^{(N)}[T\phi](0)
\end{equation*}

One might wonder whether a similar approach could be used to ``cure'' the instability discussed in this paper. In fact, we have already seen that, when an additional symmetry is present, the local energy \emph{can} be bounded in terms a higher order energy (equation \eqref{equation boundedness with loss of derivatives}). However, this is not a very satisfactory result, since it leaves open the possibility that the higher order energy is itself unbounded, and this, in turn, could be interpreted as a kind of instability (albeit of a ``weaker'' type).

Again, the additional symmetry enables us to resolve this issue. Let $\tilde{\phi} = T\phi$, where $\phi$ is the solution to the wave equation arising from the initial data we have constructed in subsection \ref{subsection constructing the initial data}, and where we pose the initial data at time $\tau_1$. Then it is not hard to see that $\tilde{\phi}$ satisfies
\begin{equation*}
\begin{split}
	\sum_{j + k \leq m} \mathcal{E}^{(N)}[T^j \Phi^k \tilde{\phi}](\tau_1) &= \mathcal{O}\left( (\delta_0)^{-2m - \frac{3}{20}} \right)
	\\
	\mathcal{E}^{(T)}[T^j \Phi^k \tilde{\phi}](\tau_1) &= \mathcal{O}\left( (\delta_0)^{-2m} \right) 
\end{split}
\end{equation*}
and so, following exactly the same arguments as before, we can show that there exists a solution to the wave equation $u_{C_1}$ and a time $\tau_{C_1}$ such that, for any $C_1 > 0$, we have
\begin{equation}
\label{equation higher derivs 1}
\begin{split}
		\sum_{j + k \leq m} \mathcal{E}^{(N)}[T^j \Phi^k u_{C_1}](0) &= 1
		\\
		\sum_{j + k \leq m} \mathcal{E}^{(N)}[T^j \Phi^k u_C](\tau_{C_1}) &\geq C_1
\end{split}
\end{equation}

Now, since at each point on the manifold $\mathcal{M}$, there is a timelike vector in the span of $T$ and $\Phi$, we can use standard elliptic estimates to show that there is some numerical constant $C_2$ such that
\begin{equation*}
\begin{split}
		\sum_{j \leq m} \mathcal{E}^{(N)}[\partial^j u_{C_1}](0) & \leq C_2
		\\
		C_1 \leq \sum_{j \leq m} \mathcal{E}^{(N)}[\partial^j u_{C_1}](\tau_{C_1}) &\leq C_2 C_1
\end{split}
\end{equation*}
where in the second line the first inequality follows from equation \eqref{equation higher derivs 1}. Hence, rescaling the solution by a factor of $(C_2)^{-1}$ and setting $C_3 = C_1 (C_2)^{-1}$ we have found a solution to the wave equation such that
\begin{equation*}
\begin{split}
		\sum_{j \leq m} \mathcal{E}^{(N)}[\partial^j u_{C_1}](0) & \leq 1
		\\
		\sum_{j \leq m} \mathcal{E}^{(N)}[\partial^j u_{C_1}](\tau_{C_1}) &\geq C_3
\end{split}
\end{equation*}
In other words, we have proved the following corollary:

\begin{corollary}[Higher order energies]
	Suppose that the same conditions holds as for lemma \ref{lemma energy amplification with symmetry}. Then, for any positive integer $m$, for any $C > 0$, there is a solution to the wave equation $u_{(C,m)}$ and a time $\tau_{(C,m)}$ such that
	\begin{equation}
		\frac{ \sum_{j \leq m} \mathcal{E}^{(N)}[\partial^j u_{(C,m)}](\tau_{(C,m)})} { \sum_{j \leq m} \mathcal{E}^{(N)}[\partial^j u_{(C,m)}](0)}  \geq C
	\end{equation}
\end{corollary}

This shows that we cannot fully ``escape'' the instability by looking at higher order energies, at least in the case where the extra symmetry is present. To be explicit: suppose that we know, initially, that the ``n-th order'' energy of some wave $\phi$ is small. In other words, we know that $\partial \partial^n \phi$ is small in $L^2$. Then, although it may be the case that $\partial \partial^{n-1} \phi$ is small in $L^2$ at all points in the future, but we can never guarantee that  $\partial \partial^n \phi$ is also small (in the same sense) at all points in the future.

\newpage

\appendix

\section{Nonexistence of manifolds with evanescent ergosurfaces and a globally timelike Killing vector field}

A curious corollary of the results we have proved above allows us to rule out a certain kind of smooth Lorentzian manifold, possessing certain symmetries and a particular asymptotic structure. Note that the statement of this corollary makes no reference to the wave equation: it is a result purely in Lorentzian differential geometry. Nevertheless, our proof makes use of the wave equation!

\begin{corollary}[Nonexistence of Lorentzian manifolds with an evanescent ergosurface and a globally timelike Killing vector field]
	There does not exist a smooth, Lorentzian manifold which is either
	\begin{itemize}
		\item asymptotically flat and possesses  an evanescent ergosurface of the first kind (condition \ref{condition asymp flat})
		\item asymptotically Kaluza-Klein and possesses  an evanescent ergosurface of the second kind (condition \ref{condition asymp KK})
	\end{itemize}
and which also posesses a \emph{uniformly timelike} Killing vector field $\hat{T}$.
\end{corollary}

\begin{proof}
	The proof is by contradiction. Suppose that such a manifold did exist. Let $\phi_C$ be the solution to the wave equation constructed in lemma \ref{lemma energy amplification with symmetry}. Then, since $\hat{T}$ is a Killing vector field, the associated energy is conserved, i.e.\ for all $\tau$,
	\begin{equation*}
		\mathcal{E}^{(\hat{T})}[\phi_C](\tau) = \mathcal{E}^{(\hat{T})}[\phi_C](0)
	\end{equation*}
	
	At the same time, by construction we have
	\begin{equation*}
		\mathcal{E}^{(N)}[\phi_C](\tau_C) = C \mathcal{E}^{(N)}[\phi_C](0)
	\end{equation*}
	
	But, since both $N$ and $\hat{T}$ are uniformly timelike, at all times $\tau$ and for any function $\phi$ we have
	\begin{equation*}
		\hat{c}\mathcal{E}^{(N)}[\phi](\tau) \leq \mathcal{E}^{(\hat{T})}[\phi](\tau) \leq \hat{C} \mathcal{E}^{(N)}[\phi_C](\tau)
	\end{equation*}
	for some constants $\hat{c}$ and $\hat{C}$.
	
	Combining these, we have
	\begin{equation*}
	\begin{split}
		\mathcal{E}^{(\hat{T})}[\phi_C](0)
		&=
		\mathcal{E}^{(\hat{T})}[\phi_C](\tau_C)
		\\
		&\geq 
		\hat{c} \mathcal{E}^{(N)}[\phi_C](\tau_C)
		\\
		&\geq 
		\hat{c} C \mathcal{E}^{(N)}[\phi_C](0)
		\\
		&\geq 
		\hat{c} \hat{C}^{-1} C \mathcal{E}^{(\hat{T})}[\phi_C](0)
	\end{split}
	\end{equation*}
	
	The constants $\hat{C}$ and $\hat{c}$ are independent of the solution $u_C$. Hence, we can choose $C > (\hat{c})^{-1} \hat{C}$, giving the required contradiction.
	
\end{proof}

We note here that an analagous proposition holds, with ``evanescent ergosurface'' replaced by ``ergoregion''. This follows immediately from the result of \cite{Moschidis2015a}.

\sloppy
\printbibliography

\end{document}